\documentclass[11pt, twoside]{article}

\usepackage[english]{babel}

\usepackage{amssymb}
\usepackage{amsmath}
\usepackage{amsthm}
\usepackage{mathrsfs}
\usepackage{titletoc}
\usepackage{color}
\usepackage{algorithm}
\usepackage{algorithmic}
\usepackage{graphicx}
\usepackage{subfigure}
\usepackage{booktabs}
\usepackage{tabularx}
\usepackage{multirow}
\usepackage{diagbox}
\usepackage{txfonts}
\usepackage{indentfirst}
\usepackage{anysize}

\usepackage{latexsym}

\usepackage[colorlinks=true,
  linkcolor=blue,
  citecolor=red,
  urlcolor=magenta,
  ]{hyperref}

\allowdisplaybreaks

\newtheorem{theorem}{Theorem}[section]
\newtheorem{lemma}[theorem]{Lemma}
\newtheorem{corollary}[theorem]{Corollary}
\newtheorem{proposition}[theorem]{Proposition}
\theoremstyle{definition}
\newtheorem{remark}[theorem]{Remark}

\newtheorem{definition}[theorem]{Definition}
\renewcommand{\appendix}{\par
   \setcounter{section}{0}%
   \setcounter{subsection}{0}%
   \setcounter{subsubsection}{0}%
   \gdef\thesection{\@Alph\c@section}%
   \gdef\thesubsection{\@Alph\c@section.\@arabic\c@subsection}%
   \gdef\theHsection{\@Alph\c@section.}%
   \gdef\theHsubsection{\@Alph\c@section.\@arabic\c@subsection}%
   \csname appendixmore\endcsname
 }

\numberwithin{equation}{section}

\begin{document}

\arraycolsep=1pt

\title{\bf\Large
Stable Image Reconstruction via Two-Parameter Power-Scale Variation
Minimization \footnotetext{\hspace{-0.35cm} 2020 {\it Mathematics Subject Classification}.
Primary 94A08; Secondary 90C31, 90C26, 94A15.
\endgraf {\it Key words and phrases.}
image reconstruction, power-scale variation, compressed sensing, restricted isometry property,
iteratively re-weighted least square.
\endgraf This project is partially supported by the National
Natural Science Foundation of China (Grant Nos. 12431006, 12371093, 12401116, and 12371094),
the Beijing Natural Science Foundation (Grant No. 1262011), and
the Fundamental Research Funds for the Central Universities (Grant No. 2253200028).}}
\author{Ziwei Li, Wengu Chen, Huanmin Ge, Limei Huo and Dachun
Yang\footnote{Corresponding author,
E-mail: \texttt{dcyang@bnu.edu.cn}/{\color{red}\today}/Final version.}
}
\date{}
\maketitle

\vspace{-0.8cm}

\begin{center}
\begin{minipage}{13cm}
{\small{\textbf{Abstract}}\quad In this article,
we introduce a power-scale variation (PSV$_{a,p}$)
with two tunable parameters: the sparsity-inducing exponent
$p\in(0,1]$ and the  scaling factor $a\in(0,\infty)$.
By minimizing the PSV$_{a,p}$, we establish stable reconstructions
in both the gradient and the image domains under the restricted isometry
property (RIP) framework. Furthermore, we design an
iteratively re-weighted least squares algorithm
IRLSPSV to solve the unconstrained
PSV$_{a,p}$ minimization. Numerical experiments
demonstrate its superior performance and broad applicability.
The main novelties are: (i) the PSV$_{a,p}$ minimization enjoys great
flexibility and wide applicability due to its two tunable
parameters $a$ and $p$,
(ii) as $a\to\infty$, the PSV$_{a,p}$ minimization reduces to the $p$-th
power total variation (TV$_p$) minimization and, even in this
limiting case, the established RIP condition for
image reconstruction is also new,
(iii) the derived RIP upper bound $\overline{\delta}$ is proved
to be asymptotically optimal in $a$ for gradient recovery,
(iv) sensitivity analysis confirms the distinct roles of $a$ and $p$,
thereby motivating a practical parameter tuning
scheme for the proposed model.
}
\end{minipage}
\end{center}

\vspace{0.2cm}

\tableofcontents

\section{Introduction}

Digital image reconstruction serves as a fundamental computational  pillar​ across diverse fields
ranging from medical imaging and remote sensing
to computational photography and autonomous systems.
However, reconstructions in the real world are often challenged by noise, incomplete measurements,
or ill-posed inverse problems, making stable and efficient solutions highly desirable.
To address these challenges,  the seminal
\emph{total variation (TV)} minimization model
proposed  by Rudin, Osher, and Fatemi \cite{ROF}
has paved the way for modern image processing.
Its strength lies in strong edge preservation​ and the promotion
of piecewise-constant solutions.
However, images reconstructed by TV method often suffer from staircasing artifacts.
To suppress them, high-order TV models have been extensively studied by
incorporating a higher-order regularization term into the original TV minimization model;
see, for example, \cite{BKP10,CMM01,LLCS17,SS08,TPHD20,YSS08}. Other methods such as
mean curvature regularization (\cite{ZC12}),
fractional or weighted Laplacian regularization (\cite{SX24,XY25}), and fractional  function space regularization
(\cite{XY26}) are also effectively
employed to image processing.

In parallel,
theoretical advances in multidimensional signal representations,
such as sampling theorems and fractional Fourier analysis
(\cite{asf17,ccls,s14,Z18,Z19,Z24}),
have provided important mathematical foundations for image reconstruction.
Beyond spatial regularity, recent efforts have also explored structured priors for multidimensional images, such as adaptive frame-based image denoising (\cite{bshyh2007,SHB16}),
low-rank representation (\cite{DZYZN25,HL25,lxhljml,tlzphjl,ZZWL26}),
fully-connected tensor network decomposition with group sparsity (\cite{TLZHJN25}),
patch-based tensor logarithmic Schatten-$p$ minimization
(\cite{LHLXJZ24}),
and generalized opponent-transformation total variation (\cite{MN25}).

Next, we focus on image reconstruction in the framework of compressed sensing.

\subsection{Image Reconstruction via Compressed Sensing}

It is well known that compressed sensing (CS) offers a powerful approach
for stable reconstruction from undersampled linear measurements
(see \cite{CRT06a, CRT06b, CT05,D06}).

In CS,  a vector $x\in\mathbb{R}^N$ is said to be \emph{$s$-sparse} for some positive integer $s$ if
$
\|x\|_{\ell_0}:=|\mathrm{supp}(x)| \le s\ll N.
$
To recover  $x$ from a  measurement $y\in\mathbb{R}^K$  characterized by $y=Ax+ \xi$,
a direct method is to search for the $\ell_0$ minimizer over the feasible solution set
$\{z\in\mathbb{R}^N: \|Az-y\|_2\le \epsilon\}$. Here, and thereafter,
$A\in\mathbb{R}^{M\times N}$ is a sensing matrix, $\xi$ denotes the noise,
$\epsilon\in[0,\infty)$ denotes the noise level,
and, for any $p\in(0,\infty)$,  $\|\cdot\|_p$ is defined by setting,
for any $x:=(x_1,\ldots,x_N)\in\mathbb{R}^N$,
$$\|x\|_{p}:=\left(\sum_{i=1}^N |x_i|^p\right)^{\frac1p}.$$
To circumvent the NP-hardness of the $\ell_0$ minimization
(see, for instance, \cite[Section 2.3]{FR13}),
 numerous innovative models have emerged as extensions of the
$\ell_1$ convex relaxation,
 offering diverse advantages in terms of sparsity promotion and robustness.
 For instance, the  $\ell_p$ relaxation with $p\in(0,1]$ (see, for example, \cite{C07,FL09,S12,WDZ15,WC13,ZL18}), the
 transformed $\ell_1$ (TL1) (see ​\cite{ZX17,ZX18}), the $\ell_1-\ell_2$ (see \cite{LYHX15,YLHX15}),
 and a  more general $\ell_1-\beta\ell_q$ framework (see \cite{HCGN23}).
To guarantee the effectiveness of these sparse recovery models,
the \emph{restricted isometry property (RIP)}, introduced by Cand\`{e}s   et al.  in \cite{CT05},
serves as a powerful theoretical framework.
Correspondingly,  deriving tight upper bounds for RIP
conditions is of significant interest; see \cite{CZ14,ZL19} for some related work.

Digital images  primarily consist of pixels with slowly varying intensities except at edges.
Consequently, they exhibit sparsity in the discrete gradient domain
due to the low-dimensionality  of the subset of pixels representing edges.
From the perspective of CS, one can recover the image by
searching for  images with the  sparsest gradients,
and the TV minimization serves as its $\ell_1$ convex relaxation in the gradient domain;
see, for example, \cite{NW13,NW13b,PMGC12,P15}.
To be precise, to reconstruct an $N\times N$ image $X$
from a noisy measurement $y=\mathcal{M}(X)+\xi$,
ones solve the following TV minimization
\begin{equation*}
\min_{Z\in\mathbb{R}^{N\times N}} \left\|Z\right\|_{\mathrm{TV}}
\quad\text{subject to}\ \ \left\|\mathcal{M}(Z)-y\right\|_2\le \epsilon,
\end{equation*}
where $\mathcal{M}$ is a linear measurement operator and $\epsilon\in[0,\infty)$.
Moreover, the \emph{TV semi-norm} here has the following two different definitions:
the \emph{anisotropic} version
$$\|X\|^{(\mathrm{a})}_{\mathrm{TV}}:=\sum_{i,j=1}^N
|(\nabla_1 X)_{i,j}|+|(\nabla_2 X)_{i,j}|,\ \ \forall\, X\in\mathbb{R}^{N\times N}$$
and the \emph{isotropic} version
$$\|X\|^{(\mathrm{i})}_{\mathrm{TV}}:=\sum_{i,j=1}^N
\sqrt{(\nabla_1 X)_{i,j}^2+(\nabla_2 X)_{i,j}^2},\ \ \forall\, X\in\mathbb{R}^{N\times N}. $$
Furthermore, ones proposed the variation-based $\ell_p$ relaxation with $p\in(0,1]$.
We refer to \cite{AM15,C07,SCBP15,YL15}
for the $p$-th power total variation (TV$_p$)
and to \cite{PMLC20,PZLA20} for its  weighted variants.
Other variation-based nonconvex relaxations can also be
found in \cite{LMS16,LLCZWY17}.

In the same spirit, Needell and Ward \cite{NW13,NW13b}
introduced the following RIP adapted to linear operators, which
provides a rigorous theoretical
framework for variation-based image reconstruction,
and using the TV minimization they established the  RIP  conditions
for image reconstruction.

\begin{definition}
Let $\mathcal{M}:\ \mathbb{R}^{N_1\times N_2}\to\mathbb{R}^K$ be a linear operator,
and let $s\in\mathbb{N}$ with $1\le s\le N_1N_2$.
The linear operator $\mathcal{M}$ is said to have the  \emph{restricted isometry property} (RIP)  of order $s$
if there exists $\delta\in[0,1)$ such that,
for any  $s$-sparse matrix $X:=[X_{i,j}]\in\mathbb{R}^{N_1\times N_2}$ (treated as a vector),
\begin{equation}\label{eq-RIP}
(1-\delta)\left\|X\right\|_2^2\le\left\|\mathcal{M}(X)
\right\|_2^2\le(1+\delta)\left\|X\right\|_2^2.
\end{equation}
The smallest $\delta$ satisfying \eqref{eq-RIP}
is called the \emph{restricted isometry constant} (RIC)
and denoted by $\delta_s$.
\end{definition}

Here, and hereafter, for any $p\in(0,\infty)$ and $X:=[X_{i,j}]\in\mathbb{R}^{N_1\times N_2}$,
$$\|X\|_p:=\left(\sum_{i,j}|X_{i,j}|^p\right)^{\frac1p}.$$

Moreover, the work \cite{NW13} has been extended to
the transformed total variation (TTV) in \cite{HCGN22}
and the $L_1-\beta L_q$ framework in \cite{HCGN23}.
In this article, we further extend the works in \cite{NW13}
and \cite{HCGN22}.

\subsection{Motivations}

Very recently, a two-parameter minimization model was 
introduced and applied to sparse signal recovery in \cite{lcgy}, i.e.,
 \begin{equation}\label{eq-unPap}
 \min_{x\in\mathbb{R}^N} \lambda P_{a,p}(x) + \frac{1}{2}\left\|Ax-y\right\|_2^2,
\end{equation}
where $a\in(0,\infty)$, $p\in(0,1]$, $\lambda\in(0,\infty)$  is the regularity parameter,
and the penalty function $P_{a,p}$ is defined by setting,
for any $x:=(x_1,\ldots,x_N)\in\mathbb{R}^N$,
\begin{equation*}
P_{a,p}(x):=\sum_{i=1}^N \rho_{a,p}(x_i):=  \sum_{i=1}^N \frac{(a+1)|x_i|^p}{a+|x_i|^p}.
\end{equation*}
This penalty function $P_{a,p}$ combines a rational damping
mechanism inherited from the TL1 and a fractional power nonlinearity inherited from $\ell_p$
because $P_{a,1}(\cdot)=\|\cdot\|_{\mathrm{TL1}}$ and $\displaystyle\lim_{a\to\infty}P_{a,p}(\cdot)=\|\cdot\|_p^p$,
where $\|\cdot\|_{\mathrm{TL1}}$ denotes the TL1 penalty.

Let us consider a power function $\phi_{a,p}(\cdot):=\frac{a+1}{a}|\cdot|^p$ defined on $\mathbb{R}$
with the same parameters $a$ and $p$.
Through a direct calculation, we obtain
\begin{equation*}
\lim_{t\to 0} \frac{\rho_{a,p}(t)}{\phi_{a,p}(t)}= 1\ \
\text{and}\ \
\lim_{t\to 0} \frac{\rho'_{a,p}(t)}{\phi'_{a,p}(t)}= 1.
\end{equation*}
This  implies that, in \eqref{eq-unPap},
 $P_{a,p}$ provides a  regularization strength
comparable to  $\frac{a+1}{a}\|\cdot\|_p^p$ near the origin,
where $p$ determines the non-convexity and $a$ determines
the multiplicative constant (i.e., scaling).
Moreover, in sparse recovery, the regularization effect
is primarily concentrated on the small-magnitude components
of the variable (i.e., $|x_i|\to 0$).
Accordingly, the parameters $a$ and $p$ play distinct roles in $P_{a,p}$,
and we call $P_{a,p}$ as a \emph{power-scale} penalty function,
the parameter $a$ as the \emph{scaling factor}, and the parameter
$p$ as the \emph{sparsity-inducing exponent}.
Note that the power-scale penalty function $P_{a,p}$
differs from  $\frac{a+1}{a}\|\cdot\|_p^p$
because they provide distinct regularization strength
when the optimization variable is away from zero.
To be precise,  away from zero
$\frac{a+1}{a}\|\cdot\|_p^p$ imposes a persistently strong regularization
whereas $P_{a,p}$ imposes a  gradually saturating one,
which can be seen from
$$
\lim_{t\to \infty} \rho_{a,p}(t)=a+1\ \ \text{and}\ \ \lim_{t\to \infty} \phi_{a,p}(t) =\infty.
$$

In image reconstruction, flat regions of images correspond to  near-zero gradients
while edges correspond to large gradients. This and the advantages of the power-scale
penalty function $P_{a,p}$ motivate us to apply the variation based on $P_{a,p}$
to the image reconstruction.

\subsection{Contributions}

In this article, we first introduce a \emph{power-scale variation (PSV$_{a,p}$)}, characterized by
two parameters: the \emph{scaling factor} $a\in(0,\infty)$ and
the \emph{sparsity-inducing exponent} $p\in(0,1]$.
We consider the following \emph{PSV$_{a,p}$ minimization problem}
\begin{equation*}
\min_{Z\in\mathbb{R}^{N\times N}} \left\|Z\right\|_{\mathrm{PSV}_{a,p}}
\quad\text{subject to}\ \ \left\|\mathcal{M}(Z)-y\right\|_2\le \epsilon,
\end{equation*}
where $\mathcal{M}$ denotes a linear measurement operator,
$y$ denotes a noisy measurement, and
$\epsilon\in[0,\infty)$ denotes the noise level.
It exhibits the following two key features:
\begin{itemize}
\item[(i)] the flexibility imparted by two parameters,
\item[(ii)] a gradually saturating penalty strength for large gradient magnitudes,
thereby facilitating edge preservation.
\end{itemize}
We also introduce the following two different definitions of
$\|Z\|_{\mathrm{PSV}_{a,p}}$, i.e.,
the \emph{anisotropic} version
\begin{equation}\label{eq-def-TTVp}
\|Z\|^{(\mathrm{a})}_{\mathrm{PSV}_{a,p}}:=\|\nabla Z\|^{(\mathrm{a})}_{\mathrm{PS}_{a,p}}
:=\sum_{i,j=1}^N \rho_{a,p} \left((\nabla_1 Z)_{i,j}\right)+\rho_{a,p}
\left((\nabla_2 Z)_{i,j}\right)
\end{equation}
and the \emph{isotropic} version
\begin{equation}\label{eq-def-iTTVp}
\|Z\|^{(\mathrm{i})}_{\mathrm{PSV}_{a,p}}:=
\|\nabla Z\|^{(\mathrm{i})}_{\mathrm{PS}_{a,p}}
:=\sum_{i,j=1}^N \rho_{a,p}\left(\sqrt{|(\nabla_1 Z)_{i,j}|^2+|(\nabla_2 Z)_{i,j}|^2}\right).
\end{equation}
They obviously satisfy $\|\cdot\|^{(\mathrm{i})}_{\mathrm{PSV}_{a,p}}
\le\|\cdot\|^{(\mathrm{a})}_{\mathrm{PSV}_{a,p}}
\le 2\|\cdot\|^{(\mathrm{i})}_{\mathrm{PSV}_{a,p}}$ and,
if $p=1$, as $a\to\infty$ they reduce to
$\|Z\|^{(\mathrm{a})}_{\mathrm{TV}}$
and $\|Z\|^{(\mathrm{i})}_{\mathrm{TV}}$, respectively.

Then, by minimizing the PSV$_{a,p}$, we establish stable reconstructions
in both the gradient and the image domains under the RIP framework.
Furthermore, we design an
iteratively re-weighted least squares algorithm
IRLSPSV to solve the unconstrained
PSV$_{a,p}$ minimization. Numerical experiments
demonstrate its superior performance and broad applicability.

The main novelties of these results include
\begin{itemize}
\item[\rm(i)]
the two parameters of the PSV$_{a,p}$ minimization play
different roles: $a$ determines the scaling while $p$ governs the non-convexity,
which  enhances the great flexibility and the wide applicability
of the proposed model across diverse scenarios,

\item[\rm(ii)] as $a\to\infty$, the PSV$_{a,p}$ minimization reduces to the $p$-th
power total variation (TV$_p$) minimization and, even in this
limiting case, the established RIP condition for
image reconstruction is also new,

\item[\rm(iii)] invoking the sparse convex combination technique,
we prove that the derived RIP upper bound $\overline{\delta}$
in \eqref{e2.14} is asymptotically optimal in $a$
for gradient recovery; to be precise, as $a\to\infty$, the reduced bound $\delta_p$ [see \eqref{eq-deltap}] is proved to be the tightest upper bound for
gradient recovery via the TV$_p$  minimization,

\item[\rm(iv)] sensitivity analysis confirms the distinct roles of $a$ and $p$,
thereby motivating a practical parameter tuning
scheme for the proposed model.
\end{itemize}

\subsection{Organization and Notation}

In Section \ref{sec-theory}, using PSV$_{a,p}$ minimization
we establish the stable recovery in the gradient domain under the RIP framework,
and derive its RIP upper bound.
Furthermore, this RIP upper bound in the limiting case
$a\to\infty$ is proved to be tightest for
gradient recovery using the TV$_p$ minimization.
Finally, invoking the strong Sobolev inequality, we
establish the stable image reconstruction.

Section \ref{sec-alg} is devoted to designing an IRLSPSV algorithm to solve
the unconstrained PSV$_{a,p}$ minimization by
using the combination of a modified  iteratively re-weighted least squares (IRLS) method,
the difference of convex functions (DC) algorithm, and the prime dual (PD) algorithm.
Besides, some related convergence analysis is provided.

In Section \ref{sec-exp}, numerical experiments are conducted to
evaluate the performance of the proposed PSV$_{a,p}$ minimization,
including three scenarios: the natural image reconstruction,
the magnetic resonance imaging (MRI) reconstruction,
and the X-ray computed tomography (CT) reconstruction.
Both  sensitivity experiments and comparative experiments
against the TV, the $L_1-\alpha L_2$, and the TTV methods
are also conducted to evaluate the performance of the proposed PSV$_{a,p}$ model.

Conclusions are given in Section \ref{s5}.

Finally, we make some notational conventions.
Throughout this article, we let $a\in(0,\infty)$ and $p\in(0,1]$ be two given parameters.
For any vector $x\in\mathbb{R}^{N}$ and any $s\in\mathbb{N}$ with $0<s\le N$,
we denote by $x_s$ the \emph{best $s$-term approximation} of $x$, which
keeps the largest $s$ elements in magnitude and lets the remainder be zero;
for any subset $T$ of $\{1,\ldots,N\}$, we use $T^\complement$ to denote its complement, and
we denote by $x_T$ the vector in $\mathbb{R}^{N}$, whose entry coincides with $x$
for the index in $T$ and is zero for the index in $T^\complement$.
Moreover, for any matrix $X\in\mathbb{R}^{N_1\times N_2}$,
we also define $X_s$ and $X_T$ by treating $X$ as a vector.
For any matrix $X:=[X_{i,j}]\in\mathbb{R}^{N_1\times N_2}$, we use $\|X\|_{\infty}$
to denote the entry-wise $\ell_\infty$ norm, i.e.,
$$\|X\|_{\infty}:= \max\{|X_{i,j}|:\ 1\le i \le N_1,\ 1\le j\le N_2\}.$$
For any matrices $X$ and $Y$, their inner product $\langle X, Y\rangle$ is defined by setting
$\langle X, Y\rangle:=\mathrm{trace}(X^T Y)$.
We also use $C$ to represent some positive constant
which is independent of the main parameters involved
and use $A\lesssim B$ to represent $A\leq CB$.
In all proofs we consistently retain the notation introduced
in the original theorem (or related statement).

\section{Stable Image Reconstruction Using PSV$_{a,p}$ Minimization }\label{sec-theory}

In this section, we establish the stable recovery in both the gradient
and the image domains, which consists of three subsections.
Subsection \ref{subsec-pre} contains some basic notions and lemmas.
The main results are given in Subsection \ref{subsec-Ima-re}, while
their proofs are presented in Subsection \ref{subsec-proof}.

\subsection{Preliminaries}\label{subsec-pre}

This subsection contains some basic knowledge
and several technique lemmas for image reconstruction.
We begin with some properties of the  function $\rho_{a,p}$,
which are a part of \cite[Lemma 2.4]{lcgy}.

\begin{lemma}\label{lem-TLp,prop}
Let $a\in(0,\infty)$ and $p\in(0,1]$. Then the following assertions hold.
\begin{enumerate}
\item[\rm(i)] $\rho_{a,p}$ is strictly increasing and concave in $[0,\infty)$
with
$\rho_{a,p}(0)=0$ and $\rho_{a,p}(1)=1$.

\item[\rm(ii)] For any $t\in\mathbb{R}$,
$\rho_{a,p}(t)\le \frac{a+1}{a}|t|^p$.

\item[\rm(iii)]
$|t|^p\le\rho_{a,p}(t)\le1$ if and only if $|t|\le 1$.

\item[\rm(iv)] For any $t_1,t_2\in\mathbb{R}$,
\begin{align*}
|\rho_{a,p}(t_1)-\rho_{a,p}(t_2)|&\le
\rho_{a,p}(t_1+t_2)\le\rho_{a,p}(|t_1|+|t_2|)\le \rho_{a,p}(t_1)+\rho_{a,p}(t_2)\\
&\le2 \rho_{a,p}\left(\frac{|t_1|+|t_2|}{2}\right).
\end{align*}
\end{enumerate}
\end{lemma}

We also recall two technique lemmas for sparse recovery. The first one is
a combination of \cite[Lemma 2.2]{ZL19} and \cite[Remark 2.13]{lcgy},
which is called the sparse convex combination lemma,
and the second one is a part of \cite[Lemma 5.3]{CZ13}.

\begin{lemma}\label{lem-sparse repre}
Let $p\in(0,1]$, $\alpha\in(0,\infty)$, and $s\in\mathbb{N}$.
For any given $u\in\mathbb{R}^N$ with $|{\mathop\mathrm{\,supp\,}}(u)|=n\ge s$,
$\|u\|_p^p\le s\alpha^p$, and $\|u\|_\infty\le \alpha$,
$u$ can be represented as a convex combination of finite $s$-sparse vectors,
i.e., $u=\sum_{i=1}^L \lambda_i v_i$ for some $L \in\mathbb{N}$,
where $\sum_{i=1}^L  \lambda_i =1$ with $\lambda_i\in(0,1]$ and $v_i$ is $s$-sparse with
${\mathop\mathrm{\,supp\,}}(v_i)\subset {\mathop\mathrm{\,supp\,}}(u)$.
Moreover,
\begin{equation*}
\sum_{i=1}^L  \lambda_i \left\|v_i\right\|_2^2\le
\min\left\{\frac{n}{s}\left\|u\right\|_2^2,\|u\|_\infty^p\left\|u\right\|_{2-p}^{2-p}\right\}.
\end{equation*}
\end{lemma}

\begin{lemma}\label{lem-Tc<T}
Let $k,l\in\mathbb{N}$ satisfy $k\le l$, $\lambda\ge0$, and $a_1\ge a_2\ge\ldots\ge a_l\ge0$.
If $\sum_{j=1}^k a_j + \lambda \ge \sum_{j=k+1}^l a_j$, then, for any $\alpha\in [1,\infty)$,
$$\sum_{j=k+1}^l a_j^\alpha  \le k\left[\left(\frac{1}{k}\sum_{j=1}^k a_j^\alpha \right)^{\frac{1}{\alpha}}
+ \frac{\lambda}{k}\right]^\alpha .$$
\end{lemma}

For any given matrices $X:=[X_{i,j}]\in\mathbb{R}^{N\times N_1}$ and $Y:=[Y_{i,j}]\in\mathbb{R}^{N\times N_2}$,
we use the symbol $[X,Y]$ to denote the \emph{horizontal concatenation matrix} $[Z_{i,j}]\in\mathbb{R}^{N\times (N_1+N_2)}$,
which is defined by setting, for any $i,j\in\mathbb{N}$ with $1\le i \le N$ and $1\le j\le N_1+N_2$,
\begin{align*}
Z_{i,j}:=
\begin{cases}
X_{i,j} \ \ &\text{if}\ 1\le j \le N_1,\\
Y_{i,j-N_1} \ \ &\text{if}\ N_1+1\le j \le N_1+N_2.
\end{cases}
\end{align*}
Then, for any given matrix $X:=[X_{i,j}]\in\mathbb{R}^{N\times N}$,
we have $[X,X^T]=[\overline{X}_{i,j}]=:\overline{X}$,
where
\begin{align}\label{eq-xij}
\overline{X}_{i,j}=\begin{cases}
X_{i,j}\ \ &\text{if}\ 1\le j \le N,\\
X_{j-N,i} \ \ &\text{if}\ N+1\le j \le 2N.\end{cases}
\end{align}
For any matrix $M\in \mathbb{R}^{N_1\times N_2}$, we also use
$M^0$ and $M_0$ to denote the padded matrices obtained
by adding a row of zeros, respectively, to the top and to the bottom of $M$.
Moreover, let $\mathcal{M}:\ \mathbb{R}^{(N-1)\times (2N) }\to\mathbb{R}^K$ be a  linear operator
identified by matrices $\{M_k\}_{k=1}^K$, i.e., defined via components by setting, for each $k\in\{1,\ldots,K\}$
and any matrix $X\in\mathbb{R}^{N_1\times N_2 }$,
$
(\mathcal{M}(X))_k:=\langle M_k,X\rangle.
$
We denote  by $\mathcal{M}^0:\ \mathbb{R}^{N\times (2N)}\to\mathbb{R}^K$
and $\mathcal{M}_0:\ \mathbb{R}^{N\times (2N)}\to\mathbb{R}^K$ the linear operators defined via components by setting,
for each $k\in\{1,\ldots,K\}$ and any matrix $X\in\mathbb{R}^{N\times (2N) }$,
\begin{equation}\label{eq-def,M0}
\big(\mathcal{M}^0(X)\big)_k=\left\langle (M_k)^0,X\right\rangle\ \
\text{and} \ \ \big(\mathcal{M}_0(X)\big)_k=\left\langle (M_k)_0,X\right\rangle.
\end{equation}

For a digital image with $N\times N$ block of pixels,
we denote it by a matrix $X:=[X_{i,j}]\in\mathbb{R}^{N\times N}$
with $X_{i,j}$ being its particular pixel.
Then its \emph{discrete directional derivatives}
are defined  by setting
$$\partial_1 X:=\left[(\partial_1 X)_{i,j}\right]\in \mathbb{R}^{(N-1)\times N}\ \ \text{and}\ \
\partial_2 X:=\left[(\partial_2 X)_{i,j}\right]\in \mathbb{R}^{N\times (N-1)},$$
where, for each $i\in\{1,\ldots,N-1\}$ and $j\in\{1,\ldots,N\}$,
\begin{equation}\label{def-DDx1}
(\partial_1 X)_{i,j}:=X_{i+1,j}-X_{i,j}
\end{equation}
and, for each $i\in\{1,\ldots,N\}$ and $j\in\{1,\ldots,N-1\}$,
\begin{equation}\label{def-DDx2}
(\partial_2 X)_{i,j}:=X_{i,j+1}-X_{i,j};
\end{equation}
moreover, the  \emph{discrete gradient transform}
$\nabla$
on $X$ is  defined by setting
$
\nabla X :=\left(\nabla_1 X,\nabla_2 X\right),
$
where
\begin{align*}
\nabla_1 X :=\left[(\nabla_1 X)_{i,j}\right]:=
\begin{pmatrix}
(\partial_1 X)_{1,1}  & \cdots & (\partial_1 X)_{1,N-1}& (\partial_1 X)_{1,N} \\
\vdots & \ddots & \vdots & \vdots \\
(\partial_1 X)_{N-1,1}  & \cdots  & (\partial_1 X)_{N-1,N-1} & (\partial_1 X)_{N-1,N} \\
0& \cdots & 0 & 0
\end{pmatrix}
\end{align*}
and
\begin{align*}
\nabla_2 X:=\left[(\nabla_2 X)_{i,j}\right]:=
\begin{pmatrix}
(\partial_2 X)_{1,1}& \cdots &(\partial_2 X)_{1,N-1} & 0\\
\vdots & \ddots & \vdots & \vdots \\
(\partial_2 X)_{N-1,1}\ & \cdots\ & (\partial_2 X)_{N-1,N-1}\ & 0\\
(\partial_2 X)_{N,1} & \cdots & (\partial_2 X)_{N,N-1} & 0
\end{pmatrix}.
\end{align*}

We can connect the measurements  of images
with those of their discrete directional derivatives 
through the following ​algebraic expression.

\begin{lemma}\label{lem-M0}
Let  $X\in\mathbb{R}^{N\times N}$,
$\overline{X} := [X,X^T]$,
and $M:=[M_{i,j}]\in\mathbb{R}^{(N-1)\times (2N)}$. Then
$$
\left\langle M, [\partial_1 X,(\partial_2 X)^T]\right\rangle
=\left\langle M^0, \overline{X}\right\rangle-\left\langle M_0,  
\overline{X}\right\rangle.
$$
\end{lemma}

\begin{proof} Using \eqref{eq-xij}, \eqref{def-DDx1}, and \eqref{def-DDx2},
we conclude that
\begin{align*}
\left\langle M^0, \overline{X}\right\rangle-\left\langle M_0,  \overline{X}\right\rangle
&=\sum_{\genfrac{}{}{0pt}{}{1\le i \le N-1}{1\le j\le N}} M_{i,j}\left(X_{i+1,j}-X_{i,j}\right)
+\sum_{\genfrac{}{}{0pt}{}{1\le i \le N-1}{N+1\le j\le 2N}} M_{i,j}\left(X_{j-N,i+1}-X_{j-N,i}\right)\\
&=\sum_{\genfrac{}{}{0pt}{}{1\le i \le N-1}{1\le j\le N}} M_{i,j}\left(\partial_1 X\right)_{i,j}
+\sum_{\genfrac{}{}{0pt}{}{1\le i \le N-1}{N+1\le j\le 2N}} M_{i,j}\left(\partial_2 X\right)_{j-N,i}\\
&=\left\langle M, \left[\partial_1 X,(\partial_2 X)^T\right]\right\rangle,
\end{align*}
which completes the proof of Lemma \ref{lem-M0}.
\end{proof}

Next, we recall the Haar wavelet transform. It is well known that
wavelets provide a powerful tool
for sparse representation of images;
see, for example, \cite{ss04,W97}.

The univariate Haar wavelet system consists of
a characteristic function  on the unit interval
\begin{equation}\label{def-h0}
h^0(t) :=\mathbf{1}_{[0,1)}(t), \ \ \forall\,t\in\mathbb{R},
\end{equation}
a mother wavelet
\begin{equation}\label{def-h1}
h^1(t) :=\mathbf{1}_{[0,\frac{1}{2})}(t)  - \mathbf{1}_{[\frac{1}{2},1)}(t),
\ \ \forall\,t\in\mathbb{R},
\end{equation}
and dyadic dilations and translates of $h^1$,
i.e., for any $\ell\in\mathbb{N}$ and $k\in\mathbb{Z}_+$ with $0\le k <2^\ell$,
$$
h_{\ell,k}(t):=2^{\ell/2}h^1(2^\ell t - k), \ \ \forall \,t\in\mathbb{R}.
$$
Moreover,
the bivariate Haar wavelet system can be constructed from the
univariate Haar wavelet system by the usual tensor product.
To be precise, let $V:=\{(0,1),(1,0),(1,1)\}$ and,
for any given $e:=(e_1,e_2)\in V$,
we define a bivariate function $h^e$ by setting
\begin{equation*}
h^e(z):=h^{e_1}(z_1)h^{e_2}(z_2),\ \ \forall\, z:=(z_1,z_2)\in\mathbb{R}^2,
\end{equation*}
where $h^0$ and $h^1$ are the same as in \eqref{def-h0} and \eqref{def-h1}.
Then the  bivariate Haar wavelet system
consists of the characteristic function $\mathbf{1}_Q$ of the unit cube $Q:=[0,1)^2$
and the dyadic dilations and translates $\{h^e_{\ell,\nu}:\ e\in V, \ell\in\mathbb{N},\nu\in\mathbb{Z}^2\cap 2^\ell Q\}$
which are defined by setting,
for any $e\in V$, $\ell\in\mathbb{N}$, and $\nu\in\mathbb{Z}^2\cap 2^\ell Q$,
\begin{equation*}
h^e_{\ell,\nu}(z):=2^\ell h^e\left(2^\ell z - \nu\right), \ \ \forall\, z \in\mathbb{R}^2.
\end{equation*}
We note that, for each $e\in V$, $h^e_{\ell,\nu}$ is supported in the dyadic cube
$Q_{\ell,\nu}:=2^{-\ell}(\nu+ Q)$.

Let $N:=2^n$. Then the bivariate Haar wavelet basis,
restricted to $N^2$ basis functions
\begin{equation*}
\left\{h^{0,0}:=\mathbf{1}_Q\right\}\cup
\left\{h^e_{\ell,\nu}: e\in V,\  \ell\in\mathbb{N},
\ 0\le \ell\le n-1,\ \nu\in\mathbb{Z}^2\cap 2^\ell Q\right\},
\end{equation*}
forms  an orthonormal basis of $\mathbb{R}^{N\times N}$.
As we know that the set of discrete images $X:=[X_{i,j}]$ with $N\times N$
block of pixels is isometrically isomorphic to the space
of piecewise-constant functions on $Q$
\begin{equation*}
\Sigma_{N}(Q):=\left\{f:\ f(z_1,z_2)=c_{i,j},\ 
\frac{i -1}{N}\le z_1 <\frac{i }{N},\ \frac{j-1}{N}\le z_2 <\frac{j}{N},
1\le i,j\le N,\ i,j\in\mathbb{N}\right\},
\end{equation*}
via $c_{i,j}=N X_{i,j}$,
we identify these $N^2$ basis functions
as discrete images
$\{H^{0,0}\}\cup\{H^e_{\ell,\nu}\}_{\ell,\nu,e}$.
In what follows, for any image $X\in\mathbb{R}^{N\times N}$,
we use $\mathcal{H}(X)$
to denote the \emph{discrete bivariate Haar transform} on $X$
which is an $N\times N$ matrix consisting of the components
$\{\langle X, H^{0,0}\rangle\}\cup\{\langle X, H^e_{\ell,\nu}\rangle\}_{\ell,\nu,e}$.

We end this subsection by
the following strong Sobolev inequality which
is precisely \cite[Theorem 9]{NW13}.

\begin{lemma}\label{lem-sobolev}
Let $s\in\mathbb{N}$ and
$\mathcal{B}: \mathbb{R}^{N\times N}\to \mathbb{R}^{K}$ be a linear operator
such that $\mathcal{B}\mathcal{H}^{-1}:\ \mathbb{R}^{N\times N}\to \mathbb{R}^{K}$
has the RIP with $\delta_{2s+1}<1$.
Suppose that $D\in\mathbb{R}^{N\times N}$ satisfies the tube constraint
$\|\mathcal{B}(D)\|_2\le \epsilon$ for some $\varepsilon>0$.
Then there exists a positive constant $C$ such that
$$
\|D\|_2 \le C\left[\left(\frac{\|D\|_{\mathrm{TV}}}
{\sqrt{s}}\right)\log\left(\frac{N^2}{s}\right)  + \epsilon\right].
$$
\end{lemma}

\subsection{Stable Image Reconstruction under the RIP Framework}\label{subsec-Ima-re}

In this subsection, we apply the proposed PSV$_{a,p}$ model to establish
the stable recovery in both the gradient and the image domains under the RIP framework,
deriving  an asymptotically optimal RIP upper bound in $a$ for gradient recovery.
For simplicity, we only consider the anisotropic case, and let
$\|\cdot\|_{\mathrm{PS}_{a,p}}:=\|\cdot\|^{(\mathrm{a})}_{\mathrm{PS}_{a,p}}$ and
$\|\cdot\|_{\mathrm{PSV}_{a,p}}:=\|\cdot\|^{(\mathrm{a})}_{\mathrm{PSV}_{a,p}}$ [see \eqref{eq-def-TTVp}].
For notational consistency, for any vector $x$, we also use
$\|x\|_{\mathrm{PS}_{a,p}}$ to denote $P_{a,p}(x).$

We start by explaining some crucial notation.
Let $\gamma\in[1,\infty)$.
In what follows,  we let $\eta_0$  be
the \emph{unique positive solution} of the equation
\begin{equation}\label{eq-equation}
\frac{p}{2}\eta^{\frac{2}{p}}+\eta-\left(1-\frac{p}{2}\right)
\left(\frac{a+1}{a}\gamma\right)^{\frac{p}{2-p}}=0
\end{equation}
and rewrite $\eta_0$ as $\overline{\eta}_0$ when $\gamma=1$.
We further define
\begin{equation*}
\mu_0:= \left(\frac{a+1}{a}\gamma \right)^{-\frac{p}{2-p}}\eta_0
\ \ \text{and}\ \
\delta(p,a,\gamma):=\frac{\mu_0}{2-p-\mu_0},
\end{equation*}
and, particularly,
 \begin{equation}\label{e2.14}
\overline{\mu}_0:= \left(\frac{a+1}{a}\right)^{-\frac{p}{2-p}}\overline{\eta}_0
\ \ \text{and}\ \
\overline{\delta}:=\delta(p,a,1).
\end{equation}
The uniqueness of $\eta_0$ follows from the monotonously increasing property  of the function
$f(\eta):=p\eta^{\frac{2}{p}}+2\eta,\ \forall\,\eta\in[0,\infty)$.
We also find $0<\eta_0<(1-\frac p2)(\frac{a+1}{a}\gamma)^{\frac{p}{2-p}}$,
and hence $\mu_0\in(0,1-\frac{p}{2})$.

The image reconstruction is based on a normalization procedure  as follows.
Let $X$ be an image, $\mathcal{M}$ be a linear measurement operator,
and $y=\mathcal{M}(X)+\xi$  be a noisy measurement with $\|\xi\|_2\le \epsilon$ for some $\epsilon\in[0,\infty)$. We claim that
there always exists $\beta\in[1,\infty)$ such that
$\|X_\beta\|_{\mathrm{PSV}_{a,p}}\le \frac{1}{2}$, where
$X_\beta:=\frac{X}{\beta}$.
Indeed, letting
\begin{equation}\label{eq-C}
\beta\ge a^{-\frac{1}{p}} \|\nabla X\|_{\infty} \left[2(a+1) |{\mathop\mathrm{\,supp\,}}(\nabla X) |-1\right]^{\frac{1}{p}},
\end{equation}
then, by Lemma \ref{lem-TLp,prop}(i), we obtain
$$
\|X_\beta\|_{\mathrm{PSV}_{a,p}}=\|\nabla X_\beta\|_{\mathrm{\mathrm{PS}_{a,p}}}
\le\left|{\mathop\mathrm{\,supp\,}}(\nabla X)\right|\rho_{a,p}\left(\frac{\|\nabla X\|_{\infty}}{\beta}\right)
=\left|{\mathop\mathrm{\,supp\,}}(\nabla X)\right|\frac{(a+1)\|\nabla X\|_{\infty}^p}{a\beta^p +\|\nabla X\|^p_{\infty}}
\le\frac{1}{2},
$$
which confirms the above claim.
Define $y_\beta:=\frac{y}{\beta}$ and $\epsilon_\beta:=\frac{\epsilon}{\beta}$.
Then  $X_\beta$ is a feasible solution of the following normalized PSV$_{a,p}$ minimization problem:
\begin{equation}\label{eq-problem}
\min_{Z\in\mathbb{R}^{N\times N}} \left\|Z\right\|_{\mathrm{PSV}_{a,p}}
\quad\text{subject to}\ \ \left\|\mathcal{M}(Z)-y_\beta\right\|_2\le \epsilon_\beta.
\end{equation}

We first establish the following stable recovery in the gradient domain.

\begin{theorem}\label{thm-stableG}
Let $K,n\in\mathbb{N}$, $N:=2^n$,  and $s\in\mathbb{N}$ with $s<N(N-1)$.
Let $\mathcal{A}:\ \mathbb{R}^{(N-1)\times (2N)}\to\mathbb{R}^{K}$ be a
linear operator having the RIP with
$\delta_{2s}< \overline{\delta}$ and
$\mathcal{M}:\ \mathbb{R}^{N\times N}\to\mathbb{R}^{K}$ be a linear operator
defined by setting, for any matrix $Z\in\mathbb{R}^{N\times N}$,
$
\mathcal{M}(Z):=(\mathcal{A}^0-\mathcal{A}_0)([Z,Z^T]).
$
Then, for any given image $X\in\mathbb{R}^{N\times N}$
and any noisy measurement $y=\mathcal{M}(X)+\xi$ with $\|\xi\|_2\le \epsilon$
for some $\epsilon\in[0,\infty)$,
\begin{equation}\label{eq-bound2}
\left\|\nabla X-\nabla \widehat{X}\right\|^p_2\le
C_1  \max\left\{\left[\frac{\|\nabla X-(\nabla X)_{s}\|^{p}_{p}}{\beta^{p-1}\sqrt{s}}\right]^p,
\left[2\left(\frac{a+1}{a}\right)\right]^{1-p} \frac{\|\nabla X-(\nabla X)_{s}\|^p_{p}}{\sqrt{s}^{2-p}}\right\}
+ C_2\epsilon^p
\end{equation}
and
\begin{align}\label{eq-bound1}
&\left\|\nabla X-\nabla \widehat{X}\right\|^p_1 \nonumber\\
&\quad\le
 C_3\max\left\{\left[ \frac{\|\nabla X-(\nabla X)_{s}\|^{p}_{p}}{\beta^{p-1}} \right]^p,
\left[2\left(\frac{a+1}{a}\right)\right]^{1-p} \frac{\|\nabla X-(\nabla X)_{s}\|^p_{p}}{s^{1-p}}\right\}
+ C_4 \sqrt{s}^p\epsilon^p ,
\end{align}
where $\widehat{X}:=\beta\widehat{X}_\beta$
with $\beta$ satisfying \eqref{eq-C} and $\widehat{X}_\beta$
being a minimizer of the normalized minimization problem \eqref{eq-problem}.
Here the positive constants
$$
C_1:=2^p\left(\frac{a+1}{a}\right)^p\left\{1+\widetilde{C}^p
\left[1+\left(\frac{a+1}{a}\right)^{\frac{2}{p}}\right]^{\frac{p}{2}}\right\},
$$
$$C_2:=2^p C^p\left[1+\left(\frac{a+1}{a}\right)^{\frac{2}{p}}\right]^{\frac{p}{2}},$$
$$
C_3:=2^p \left(\frac{a+1}{a}\right)^p \left\{1+\widetilde{C}^p \left[1+\left(\frac{a+1}{a}\right)^{\frac{1}{p}}\right]^p\right\},$$
and
$$C_4:=2^p C^p\left[1+\left(\frac{a+1}{a}\right)^{\frac{1}{p}}\right]^p
$$
with
\begin{align*}
C:=\frac{\sqrt{1+\delta_{2s}} \,\overline{\mu}_0(1-\overline{\mu}_0)(2-p)
+\overline{\mu}_0 (2-p-\overline{\mu}_0) \sqrt{(1-p)(\overline{\delta}-\delta_{2s})}}
{(2-p-\overline{\mu}_0)^2 (\overline{\delta}-\delta_{2s})}
\end{align*}
and
\begin{align*}
\widetilde{C}:=\frac{p[\delta_{2s}(2-p)]^{\frac{2-p}{p}}\overline{\mu}_0^2}
{2(1+\delta_{2s})^{\frac{2-2p}{p}}(2-p-\overline{\mu}_0)^2 (\overline{\delta}-\delta_{2s})}.
\end{align*}
\end{theorem}

In practical applications, natural images often have sparse discrete gradients.
For those images with sparse discrete gradients,
Theorem \ref{thm-stableG} implies the exact  recovery in the gradient domain
from a noiseless measurement. To be precise,
under the assumptions of Theorem \ref{thm-stableG},
if an image $X\in\mathbb{R}^{N\times N}$ has an $s$-sparse 
discrete gradient $\nabla X$ (treated as a vector),
then  $\nabla X=\nabla \widehat{X}$, where $\widehat{X}:=\beta\widehat{X}_\beta$
and $\widehat{X}_\beta$ is the minimizer of the normalized minimization:
\begin{equation}\label{eq-TTVp,exact}
\min_{Z\in\mathbb{R}^{N\times N}} \left\|Z\right\|_{\mathrm{PSV}_{a,p}}
\quad\text{subject to}\ \ y_\beta=\mathcal{M}(Z)
\end{equation}
with $y_\beta$ being a normalized noiseless measurement and $\beta$ satisfying \eqref{eq-C}.

Moreover, as $a\to \infty$, we have $\|\cdot\|_{\mathrm{PS}_{a,p}}\to\|\cdot\|_p^p$.
In this particular case, due to the scaling property of $\|\cdot\|_p$,
the minimization \eqref{eq-TTVp,exact} is equivalent to the following
\emph{$p$-th power total variation (TV$_p$)} minimization:
\begin{equation}\label{eq-TpV}
\min_{Z\in\mathbb{R}^{N\times N}} \left\|Z\right\|_{\mathrm{TV}_p}
\quad\text{subject to}\ \  \mathcal{M}(X)=\mathcal{M}(Z),
\end{equation}
where $\|Z\|_{\mathrm{TV}_p}:=\sum_{i,j} [|(\nabla_1 X)_{i,j}|^p+|(\nabla_2 X)_{i,j}|^p]$
is the anisotropic $p$-th power total variation;
see, for example, \cite{AM15}.
In addition, in this case, $\overline{\eta}_0$ reduces to the 
unique positive solution of the equation
\begin{equation}\label{eq-equation,p}
\frac{p}{2}\eta^{\frac{2}{p}}+\eta-  1 + \frac{p}{2} =0.
\end{equation}
To distinguish this $\overline{\eta}_0$ with the one in \eqref{e2.14},
we rewrite it as  $\zeta_0$.
Then, in this case, $\overline{\delta}$ in \eqref{e2.14} reduces to
\begin{equation}\label{eq-deltap}
\frac{\zeta_0}{2-p-\zeta_0}=:\delta_p.
\end{equation}
These, together with Theorem \ref{thm-stableG},
imply the following exact recovery in the gradient domain
via the TV$_p$ minimization for images with sparse discrete gradients.

\begin{corollary}\label{coro-stableG}
Let $K,n\in\mathbb{N}$, $N:=2^n$,  and $s\in\mathbb{N}$ with $s<N(N-1)$.
Let $\mathcal{A}:\ \mathbb{R}^{(N-1)\times (2N)}\to\mathbb{R}^{K}$ be a
linear operator having the  RIP with
$\delta_{2s}< \delta_{p}$ and
$\mathcal{M}:\ \mathbb{R}^{N\times N}\to\mathbb{R}^{K}$ be a linear operator
defined by setting, for any matrix $Z\in\mathbb{R}^{N\times N}$,
$
\mathcal{M}(Z):=(\mathcal{A}^0-\mathcal{A}_0)([Z,Z^T]).
$
Then, for any image $X\in\mathbb{R}^{N\times N}$ with an $s$-sparse discrete gradient,
$\nabla X=\nabla \widehat{X}$, where $\widehat{X}$ is a minimizer of \eqref{eq-TpV} and
$\delta_p$ is defined in \eqref{eq-deltap}.
\end{corollary}

Furthermore, the following theorem establishes  the optimality of the upper bound $\delta_{p}$
in Corollary \ref{coro-stableG} for gradient recovery via the TV$_p$ minimization,
which further implies that the upper bound $\overline{\delta}$ in Theorem \ref{thm-stableG}
is asymptotically optimal in $a$.

\begin{theorem}\label{prop-sharp}
For any $\varepsilon\in(0,1)$ and $s\ge\frac{8}{\varepsilon}$ with $\frac{s}{4}\in\mathbb{N}$,
there exist a linear operator $\mathcal{A}$, having the RIP  with
$\delta_{2s}< \delta_{p}+\varepsilon$, and an image $X_0\in\mathbb{R}^{N\times N}$
with an $s$-sparse discrete gradient such that $\nabla X_0$
cannot be recovered  by \eqref{eq-TpV}, where $N$ is big enough,
$\delta_p$ is defined in \eqref{eq-deltap}, and $\mathcal{M}$ in \eqref{eq-TpV} is
defined by setting, for any matrix $Z\in\mathbb{R}^{N\times N}$,
$\mathcal{M}(Z):=(\mathcal{A}^0-\mathcal{A}_0)([Z,Z^T]).$
\end{theorem}

Finally, by invoking the strong Sobolev inequality,
we establish the following stable  image reconstruction.

\begin{theorem}\label{thm-stable}
Let $n,K_1,K_2\in\mathbb{N}$, $N:=2^n$,  $K:=K_1 +K_2$, and $s\in\mathbb{N}$ with $2s+1<N^2$.
Let $\mathcal{A}:\ \mathbb{R}^{(N-1)\times (2N)}\to\mathbb{R}^{K_1}$ be a
linear operator having the RIP with $\delta_{2s}< \overline{\delta}$,
$\mathcal{B}:\ \mathbb{R}^{N\times N}\to\mathbb{R}^{K_2}$
be a linear operator such that the composite operator
$\mathcal{B}\mathcal{H}^{-1}:\ \mathbb{R}^{N\times N}\to\mathbb{R}^{K_2}$
has the RIP with $\delta_{2s+1}<1$, and
$\mathcal{M}:\ \mathbb{R}^{N\times N}\to\mathbb{R}^{K}$ be a linear operator
defined by setting, for any matrix $Z\in\mathbb{R}^{N\times N}$,
$\mathcal{M}(Z):=((\mathcal{A}^0-\mathcal{A}_0)
[Z,Z^T],\mathcal{B}(Z)).$
Then, for any given image $X\in\mathbb{R}^{N\times N}$
and any noisy measurement $y=\mathcal{M}(X)+\xi$ with $\|\xi\|_2\le \epsilon$ for some $\epsilon\in[0,\infty)$,
\begin{align}\label{eq-stablePSV}
\left\|X-\widehat{X}\right\|_2 \le
C_0\log\left(\frac{N^2}{s}\right)\left[\max\left\{ \frac{\beta^{1-p}\|\nabla X-(\nabla X)_{s}\|^{p}_{p}}{ \sqrt{s}},
\frac{\|\nabla X-(\nabla X)_{s}\|_{p}}{\sqrt{s}^{\frac{2}{p}-1}}\right\}
+  \epsilon\right],
\end{align}
where $\widehat{X}:=\beta\widehat{X}_\beta$ with $\widehat{X}_\beta$
being a minimizer of the normalized minimization  \eqref{eq-problem}
and $\beta$ satisfying \eqref{eq-C} and where $C_0$ is a positive constant depending only on
$a$, $p$, $s$, $N$, and $\mathcal{A}$.
\end{theorem}

\begin{remark}
\begin{enumerate}
\item[\rm(i)]
When $p=1$, the PSV$_{a,p}$ minimization  reduces to the TTV minimization
introduced in \cite{HCGN22}. In this case, the RIP condition
becomes $\delta_{2s}<\frac{1}{\sqrt{1+\frac{a+1}{a}}}$
which coincides with the one in \cite[Theorem 3.3]{HCGN22};
moreover, as $a\to\infty$, the RIP condition
approaches the condition $\delta_{2s}<\frac{\sqrt{2}}{2}$,
which improves the one
$\delta_{5s}<\frac{1}{3}$ obtained in \cite[Theorem 5]{NW13}.

\item[\rm(ii)] In \eqref{eq-stablePSV}, to achieve a better error bound,
we choose $\beta\in[1,\infty)$ as the smallest positive constant satisfying \eqref{eq-C}.

\item[\rm(iii)]
Let $\widehat{X}_\beta$ be a minimizer of  \eqref{eq-problem}.
Note that, in the limiting case $a\to \infty$,
$\widehat{X}:=\beta\widehat{X}_\beta$ is also a minimizer of
the following TV$_p$ minimization with noise:
\begin{equation}\label{eq-TVpnoise}
\min_{Z\in\mathbb{R}^{N\times N}} \left\|Z\right\|_{\mathrm{TV}_p}
\quad\text{subject to}\ \  \left\|\mathcal{M}(Z)-y\right\|_2\le \epsilon,
\end{equation}
where $y$ is a noisy measurement, $\mathcal{M}$ is a linear measurement
operator satisfying the same assumptions
as in Theorem \ref{thm-stable} with $\overline{\delta}$ replaced by $\delta_p$,
and $\epsilon\in[0,\infty)$ denotes the noise level.
Then, by Theorem \ref{thm-stable}, any minimizer of  \eqref{eq-TVpnoise}
satisfies \eqref{eq-stablePSV}.
Here, by the scaling property of $\|\cdot\|_{\mathrm{TV}_p}$,
we choose $\beta:=2^{\frac{1}{p}}\|\nabla X\|_p$ to ensure $\|X_\beta\|_{\mathrm{PSV}_{a,p}}\le \frac{1}{2}$.
\end{enumerate}
\end{remark}

\subsection{Proofs of Theorems \ref{thm-stableG}, \ref{prop-sharp}, and \ref{thm-stable}}\label{subsec-proof}

In this subsection, we focus on the detailed proofs of
Theorems \ref{thm-stableG}, \ref{prop-sharp}, and \ref{thm-stable}.

We first show the following key proposition.

\begin{proposition}\label{prop-u}
Let $\epsilon\in[0,\infty)$, $\gamma\in[1,\infty)$, $\sigma\in[0,\infty)$,
and $\mathcal{M}:\ \mathbb{R}^{N_1\times N_2}\to \mathbb{R}^K$ be a linear operator
having the RIP of order $2s$ with
$\delta_{2s}<\delta(p,a,\gamma)$ for some $s\in\mathbb{N}$ with $2s< N_1N_2$.
For any matrix $X\in\mathbb{R}^{N_1\times N_2}$ with $\|X\|_{\infty}\le 1$,
if $X$ satisfies the tube constraint $\|\mathcal{M}(X)\|_2\le \epsilon$
and  the  cone constraint
$\|X_{S^\complement}\|_{\mathrm{\mathrm{PS}_{a,p}}}\le \gamma \|X_{S}\|_{\mathrm{\mathrm{PS}_{a,p}}}+\sigma$
 for some index set $S$ with $|S|\le s$,
then there exist positive constants $c_1$, $c_2$, $c_3$, and $c_4$ such that
\begin{equation}\label{eq-bound2'}
\left\|X\right\|^p_2\le
c_1 \max\left\{ \frac{\sigma^p}{\sqrt{s}^{p}}, \frac{\sigma}{\sqrt{s}^{2-p}}\right\}
+ c_2\epsilon^p
\end{equation}
and
\begin{align}\label{eq-bound1'}
\left\|X\right\|^p_1\le
c_3\max\left\{\sigma^p,\frac{\sigma}{s^{1-p}}\right\}
+ c_4 \sqrt{s}^p \epsilon^p,
\end{align}
where
$$
c_1:=1+\widetilde{c}^p\left[1+\left(\frac{a+1}{a}\gamma\right)^{\frac{2}{p}}\right]^{\frac{p}{2}} ,\ \ c_2:=c^p\left[1+\left(\frac{a+1}{a}\gamma\right)^{\frac{2}{p}}\right]^{\frac{p}{2}},
$$
$$
c_3:=1+ \widetilde{c}^p\left[1+\left(\frac{a+1}{a}\gamma\right)^{\frac{1}{p}}\right]^p,
\ \
c_4:=c^p \left[1+\left(\frac{a+1}{a}\gamma\right)^{\frac{1}{p}}\right]^p ,
$$
\begin{align}\label{eq-def-c}
c:=\frac{\sqrt{1+\delta_{2s}} \,\mu_0(1-\mu_0)(2-p)
+\mu_0 (2-p-\mu_0) \sqrt{(1-p)[\delta(p,a,\gamma)-\delta_{2s}]}}
{(2-p-\mu_0)^2 [\delta(p,a,\gamma)-\delta_{2s}]},
\end{align}
and
\begin{align}\label{eq-def-c'}
\widetilde{c}:=\frac{p[\delta_{2s}(2-p)]^{\frac{2-p}{p}}\mu_0^2}
{2(1+\delta_{2s})^{\frac{2-2p}{p}}(2-p-\mu_0)^2 [\delta(p,a,\gamma)-\delta_{2s}]}.
\end{align}
\end{proposition}

\begin{remark}
Let $c$ and $\widetilde{c}$ be defined the same as in
\eqref{eq-def-c} and \eqref{eq-def-c'}.
When $p=1$, \eqref{eq-bound2'} and \eqref{eq-bound1'} reduce, respectively, to
\begin{equation*}
\left\|X\right\|_2\le\left[1+\widetilde{c}\sqrt{1
+\left(\frac{a+1}{a}\gamma\right)^{2}}\right]\frac{\sigma}{\sqrt{s}} +
 c\sqrt{1+\left(\frac{a+1}{a}\gamma\right)^{2}} \epsilon
\end{equation*}
and
\begin{align*}
\left\|X\right\|_1\le \left[1+\widetilde{c} \left(1+\frac{a+1}{a}\gamma\right)\right] \sigma
+ c \left(1+\frac{a+1}{a}\gamma\right)   \sqrt{s }\epsilon,
\end{align*}
where the positive constants $c$ and $\widetilde{c}$ are  simplified  as
\begin{align*}
c:=\frac{\sqrt{1+\delta_{2s}} \,\mu_0}
{(1-\mu_0) [\delta(1,a,\gamma)-\delta_{2s}]}
\ \ \text{and}\ \
\widetilde{c}:=\frac{\delta_{2s}\mu_0^2}
{2(1-\mu_0)^2 [\delta(1,a,\gamma)-\delta_{2s}]}.
\end{align*}
Note that, in this case,
$\frac{\mu_0}{1-\mu_0}=\frac{1}{\sqrt{1+\frac{a+1}{a}\gamma}}$
and
$\delta(1,a,\gamma)=\frac{1}{\sqrt{1+\frac{a+1}{a}\gamma}}.$
Inserting them into the above estimates yields
\begin{align*}
\left\|X\right\|_2\le
\frac{\sqrt{1+(\frac{a+1}{a}\gamma)^2}\sqrt{1+\delta_{2s}}}
{1-\sqrt{1+\frac{a+1}{a}\gamma}\,\delta_{2s}} \epsilon
+ \left[1+ \frac{\sqrt{1+(\frac{a+1}{a}\gamma)^2} \,\delta_{2s}}
{2\sqrt{1+\frac{a+1}{a}\gamma}-2(1+\frac{a+1}{a}\gamma)\,\delta_{2s}}\right] \frac{\sigma}{\sqrt{s}},
\end{align*}
which recovers the upper bound with respect to the corresponding
$\ell_2$-norm in \cite[Proposition 3.1]{HCGN22},
and a tighter upper bound
\begin{align*}
\left\|X\right\|_1\le
\frac{(1+\frac{a+1}{a}\gamma) \sqrt{1+\delta_{2s}}}{1-\sqrt{1+\frac{a+1}{a}\gamma}\,\delta_{2s}} \sqrt{s }\epsilon
+ \left[1+ \frac{(1+\frac{a+1}{a}\gamma ) \sqrt{1+\frac{a+1}{a}\gamma}\,
\delta_{2s}}{2-2\sqrt{1+\frac{a+1}{a}\gamma}\ \delta_{2s}}\right] \sigma
\end{align*}
compared with the upper bound with respect to the corresponding
$\ell_1$-norm in \cite[Proposition 3.1]{HCGN22}.
\end{remark}

\begin{proof}[Proof of Proposition \ref{prop-u}]
Let $T$ be the index set of the largest $s$
entries of $X$ in magnitude.
By the assumptions and the increasing property of $\rho_{a,p}$,
we have
\begin{equation*}
\left\|X_{T^\complement}\right\|_{\mathrm{\mathrm{PS}_{a,p}}}\le \gamma \left\|X_{T}\right\|_{\mathrm{\mathrm{PS}_{a,p}}} +\sigma.
\end{equation*}
From this, $\|X\|_\infty\le 1$,
Lemma \ref{lem-TLp,prop}(iii), the choice of $T$,
$\gamma\ge 1$,  and $\sigma\ge 0$,
we infer that
\begin{equation}\label{eq-DScfz}
\left\|X_{T^\complement}\right\|^p_\infty \le\rho_{a,p}
\left(\|X_{T^\complement}\|_\infty\right)\le\frac{\|X_{T}\|_{\mathrm{\mathrm{PS}_{a,p}}}}{s}
\le\frac{\gamma \left\|X_{T}\right\|_{\mathrm{\mathrm{PS}_{a,p}}} + \sigma}{s}
\end{equation}
and
\begin{equation}\label{eq-DSc1}
\left\|X_{T^\complement}\right\|_p^p\le \left\|X_{T^\complement}\right\|_{\mathrm{\mathrm{PS}_{a,p}}}
\le\gamma \left\|X_{T}\right\|_{\mathrm{\mathrm{PS}_{a,p}}} + \sigma.
\end{equation}
Combining \eqref{eq-DScfz} and \eqref{eq-DSc1} and
applying Lemma \ref{lem-sparse repre} with $\alpha:=(\frac{\gamma \|X_{T}\|_{\mathrm{\mathrm{PS}_{a,p}}}+\sigma}{s})^{\frac1p}$,
we obtain a convex combination
$X_{T^\complement}=\sum_{i=1}^L \lambda_i X^{(i)}$, where
$L \in\mathbb{N}$, $\sum_{i=1}^L  \lambda_i =1$ with each $\lambda_i\in(0,1]$,
$X^{(i)}$ is $s$-sparse, $\mathrm{supp}\,(X^{(i)})\subset T^\complement$, and
\begin{equation}\label{eq-moreover}
\sum_{i=1}^L  \lambda_i \left\|X^{(i)}\right\|_2^2
\le \left\|X_{T^\complement}\right\|_\infty^p\left\|X_{T^\complement}\right\|_{2-p}^{2-p}.
\end{equation}
By the choice of $T$ and H\"{o}lder's  inequality, we find that
$$
\left\|X_{T^\complement}\right\|^p_\infty \le\frac{\|X_T\|^p_p}{s}\ \ \text{and}\ \
\left\|X_{T}\right\|_{p}^p\le s^{1-\frac{p}{2}}\left\|X_{T}\right\|_{2}^p.
$$
Applying these, H\"{o}lder's  inequality, \eqref{eq-DSc1}, and Lemma \ref{lem-TLp,prop}(ii),
\eqref{eq-moreover} can be further estimated as follows
\begin{align}\label{eq-est1}
\sum_{i=1}^L  \lambda_i \left\|X^{(i)}\right\|_2^2
&\le \frac{\|X_{T}\|_{p}^{p}}{s}
\left(\|X_{T^\complement}\|_{p}^{p}\right)^{\frac{p}{2-p}}
\left(\|X_{T^\complement}\|_{2}^{2}\right)^{\frac{2-2p}{2-p}}\nonumber\\
&\le \frac{\|X_{T}\|_{p}^{p}}{s} \left(\gamma \left\|X_{T}\right\|_{\mathrm{\mathrm{PS}_{a,p}}}+\sigma\right)^{\frac{p}{2-p}}
\left(\|X_{T^\complement}\|_{2}^{2}\right)^{\frac{2-2p}{2-p}}\nonumber\\
&\le \frac{\|X_{T}\|_{p}^{p}}{s} \left(\frac{a+1}{a}\gamma\|X_{T}\|^p_p+\sigma\right)^{\frac{p}{2-p}}
\left(\|X_{T^\complement}\|_{2}^{2}\right)^{\frac{2-2p}{2-p}}\nonumber\\
&\le \left(\frac{a+1}{a}\gamma\|X_{T}\|^2_2
+\frac{\sigma}{s^{\frac{2-p}{2}}}\|X_{T}\|_{2}^{2-p} \right)^{\frac{p}{2-p}}
\left(\|X_{T^\complement}\|_{2}^{2}\right)^{\frac{2-2p}{2-p}}=:\Pi .
\end{align}

Let $\mu\in\mathbb{R}$ be a constant which is determined later
and, for each $i\in\{1,\ldots,L \}$, define $Z^{(i)}:=X_T+\mu X^{(i)}$.
By some simple computations, we obtain
\begin{align}\label{eq-equality}
&\sum_{i=1}^L  \lambda_i\left\|\mathcal{M}\,\left(\sum_{j=1}^L  \lambda_j Z^{(j)} -\frac p2 Z^{(i)} \right)\,\right\|_2^2
+\frac{1-p}{2}\sum_{i,j=1}^L  \lambda_i \lambda_j  \left\|\mathcal{M}\left(Z^{(i)}-Z^{(j)}\right)\right\|_2^2 \nonumber\\
&\quad=(1-p)\left\|\mathcal{M}\left(\sum_{j=1}^L  \lambda_j Z^{(j)}\right)\right\|_2^2
+\frac{p^2}{4}\sum_{i=1}^L  \lambda_i \left\|\mathcal{M}(Z^{(i)})\right\|_2^2\nonumber\\
&\quad\quad+(1-p)\left[\sum_{i=1}^L \lambda_i \left\|\mathcal{M}(Z^{(i)})\right\|_2^2
-\left\|\mathcal{M}\left(\sum_{i=1}^L  \lambda_i Z^{(i)}\right)\right\|_2^2\right]\nonumber\\
&\quad=\left(1-\frac{p}{2}\right)^2\sum_{i=1}^L \lambda_i \left\|\mathcal{M} (Z^{(i)})\right\|_2^2,
\end{align}
where we used the fact $\sum_{j=1}^L  \lambda_j =1$ in the first equality.
For simplicity, we denote by $\mathrm{LHS}$ the left-hand side
and by $\mathrm{RHS}$ the right-hand side of \eqref{eq-equality}.
As $\sum_{j=1}^L  \lambda_j =1$, we have, for each $i\in\{1,\ldots,L \}$,
$$
\sum_{j=1}^L  \lambda_j Z^{(j)} -\frac p2 Z^{(i)}=\mu X +\left(1-\frac p2 -\mu\right) X_T -\frac{p\mu}{2}X^{(i)}.
$$
Thus, from this,  $X_{T^\complement}=\sum_{i=1}^L \lambda_i X^{(i)}=X-X_T$,
the Cauchy--Schwarz inequality,
the fact that $X^{(i)}-X^{(j)}$ is $2s$-sparse, the RIP, and the tube constraint, it follows that
\begin{align*}
\mathrm{LHS}
&=\sum_{i=1}^L  \lambda_i\left\|\mathcal{M}\,\left(\left[1-\frac p2 -\mu\right] X_T -\frac{p\mu}{2}X^{(i)}\right)\,\right\|_2^2
+\frac{1-p}{2}\mu^2\sum_{i,j=1}^L  \lambda_i \lambda_j  \left\|\mathcal{M}(X^{(i)}-X^{(j)})\right\|_2^2\\
&\quad+\mu\langle \mathcal{M}(X), \mu(1-p) \mathcal{M}(X)+(2-p)(1-\mu)\mathcal{M}(X_T)\rangle\\
&\le(1+\delta_{2s})\left[\left(1-\frac p2 -\mu\right) ^2\left\|X_T\right\|_2^2
+\frac{p^2\mu^2}{4}\sum_{i=1}^L  \lambda_i\left\|X^{(i)}\right\|_2^2
+\frac{1-p}{2}\mu^2\sum_{i,j=1}^L  \lambda_i \lambda_j  \left\|X^{(i)}-X^{(j)}\right\|_2^2\right]\\
&\quad+\mu^2(1-p) \epsilon^2
+\mu(1-\mu)(2-p) \epsilon\sqrt{1+\delta_{2s}}  \left\|X_T\right\|_2\\
&=(1+\delta_{2s})\left[\left(1-\frac p2 -\mu\right) ^2\left\|X_T\right\|_2^2
+\frac{p^2\mu^2}{4}\sum_{i=1}^L  \lambda_i\left\|X^{(i)}\right\|_2^2\right]\\
&\quad+(1+\delta_{2s})(1-p)\mu^2\left(\sum_{i=1}^L   \lambda_i  \left\|X^{(i)}\right\|_2^2
-\left\|X_{T^\complement}\right\|_2^2\right)\\
&\quad+\mu^2(1-p) \epsilon^2
+\mu(1-\mu)(2-p) \epsilon\sqrt{1+\delta_{2s}}  \left\|X_T\right\|_2.
\end{align*}
For the right-hand side of \eqref{eq-equality},
by $\mathrm{supp}\,(X^{(i)})\subset T^\complement$, we have
\begin{align*}
\mathrm{RHS}
\ge (1-\delta_{2s}) \left(1-\frac{p}{2}\right)^2\sum_{i=1}^L  \lambda_i\left\| Z^{(i)}\right\|_2^2
= (1-\delta_{2s}) \left(1-\frac{p}{2}\right)^2\left(\left\| X_T \right\|_2^2
+\mu^2\sum_{i=1}^L  \lambda_i\left\|X^{(i)}\right\|_2^2\right).
\end{align*}
Collecting these two estimates and using \eqref{eq-est1}, we conclude that
\begin{align*}
0&\le2\left(1-\frac p2\right)^2 \delta_{2s}\mu^2 \Pi
-(1+\delta_{2s})(1-p)\mu^2\left\|X_{T^\complement}\right\|_2^2
+\sqrt{1+\delta_{2s}} \,\mu(1-\mu)(2-p)\epsilon\left\|X_T\right\|_2\\
&\qquad+\mu^2(1-p) \epsilon^2+
\left[ (1+\delta_{2s})\left(1-\frac p2 -\mu\right)^2-(1-\delta_{2s})
\left(1-\frac{p}{2}\right)^2\right]\left\|X_T\right\|_2^2,
\end{align*}
where the right-hand side   is a function of $\|X_{T^\complement}\|_{2}^{2}$.
By calculating the derivative of this function, we find that it can achieve its maximum
only when
$$
\left\|X_{T^\complement}\right\|_{2}^{2}=
\left(\frac{a+1}{a}\gamma\left\|X_{T}\right\|_{2}^2 + \frac{\sigma}{s^{\frac{2-p}{2}}}\left\|X_{T}\right\|_{2}^{2-p}\right)
\left[\frac{(2-p)\delta_{2s}}{1+\delta_{2s}}\right]^{\frac{2-p}{p}}.
$$
Thus, we obtain
\begin{align*}
0&\le\left[(1+\delta_{2s})\left(1-\frac p2 -\mu\right)^2 -(1-\delta_{2s})
\left(1-\frac{p}{2}\right)^2
+\frac{a+1}{a}\gamma\frac p2\frac{[\delta_{2s}(2-p)]^{\frac{2-p}{p}}}
{(1+\delta_{2s})^{\frac{2-2p}{p}}}\mu^2\right]\left\|X_T\right\|_2^2\\
&\quad+ \frac p2\frac{[\delta_{2s}(2-p)]^{\frac{2-p}{p}}}
{(1+\delta_{2s})^{\frac{2-2p}{p}}}\mu^2 \frac{\sigma}{s^{\frac{2-p}{2}}}\left\|X_T\right\|_2^{2-p}
+\sqrt{1+\delta_{2s}} \,\mu(1-\mu)(2-p)\epsilon\left\|X_T\right\|_2+\mu^2(1-p) \epsilon^2.
\end{align*}
By $\delta_{2s}<\delta(p,a,\gamma)$, we have $\frac{(2-p)\delta_{2s}}{1+\delta_{2s}}<\mu_0$.
Using this and \eqref{eq-equation}
and setting
$$
\mu:=\mu_0=\left(\frac{a+1}{a}\gamma\right)^{-\frac{p}{2-p}}\eta_0,
$$
we conclude that
\begin{align*}
0&\le\left\{\left[\left(1-\frac{p}{2}-\mu_0\right)^2-\left(1-\frac{p}{2}\right)^2
+\frac p2 \frac{a+1}{a}\gamma\mu_0^{\frac{2+p}{p}}\right]\right.\nonumber\\
&\quad\left.
+\left[\left(1-\frac{p}{2}-\mu_0\right)^2+\left(1-\frac{p}{2}\right)^2
+\frac p2 \frac{a+1}{a}\gamma\mu_0^{\frac{2+p}{p}}\right]\delta_{2s}\right\}\left\|X_T\right\|_2^2\nonumber\\
&\quad+ \sqrt{1+\delta_{2s}} \,\mu_0(1-\mu_0)(2-p)\epsilon\left\|X_T\right\|_2\nonumber\\
&\quad+\frac p2\frac{[\delta_{2s}(2-p)]^{\frac{2-p}{p}}}
{(1+\delta_{2s})^{\frac{2-2p}{p}}}\mu^2\frac{\sigma}{s^{\frac{2-p}{2}}}\left\|X_T\right\|_2^{2-p}
+\mu_0^2(1-p) \epsilon^2\nonumber\\
&=\left(1-\frac{p}{2}\right) \left[-\mu_0+(2-p-\mu_0)\delta_{2s}\right]\left\|X_T\right\|_2^2
+ \sqrt{1+\delta_{2s}} \,\mu_0(1-\mu_0)(2-p)\epsilon\left\|X_T\right\|_2\nonumber\\
&\quad+\frac p2\frac{[\delta_{2s}(2-p)]^{\frac{2-p}{p}}}
{(1+\delta_{2s})^{\frac{2-2p}{p}}}\mu_0^2\frac{\sigma}{s^{\frac{2-p}{2}}}\left\|X_T\right\|_2^{2-p}
+\mu_0^2(1-p) \epsilon^2\nonumber\\
&\le -(2-p-\mu_0)^2 \left[\delta(p,a,\gamma)-\delta_{2s}\right]\left\|X_T\right\|_2^2
+ \sqrt{1+\delta_{2s}} \,\mu_0(1-\mu_0)(2-p)\epsilon\left\|X_T\right\|_2\nonumber\\
&\quad+\frac p2\frac{[\delta_{2s}(2-p)]^{\frac{2-p}{p}}}
{(1+\delta_{2s})^{\frac{2-2p}{p}}}\mu_0^2\frac{\sigma}{s^{\frac{2-p}{2}}}\left\|X_T\right\|_2^{2-p}
+\mu_0^2(1-p) \epsilon^2.
\end{align*}
Since $\|X_T\|_2^{2-p}\le(\sqrt{s}\|X\|_\infty)^{1-p}\|X_T\|_2\le\sqrt{s}^{1-p}\|X_T\|_2$,
it then follows that
\begin{align*}
0&\le-(2-p-\mu_0)^2 \left[\delta(p,a,\gamma)-\delta_{2s}\right]\left\|X_T\right\|_2^2
+\mu_0^2(1-p) \epsilon^2\\
&\quad+ \left[\sqrt{1+\delta_{2s}} \,\mu_0(1-\mu_0)(2-p)\epsilon
+\frac p2\frac{[\delta_{2s}(2-p)]^{\frac{2-p}{p}}}
{(1+\delta_{2s})^{\frac{2-2p}{p}}}\mu_0^2
\frac{\sigma}{\sqrt{s}}\right]
\left\|X_T\right\|_2,
\end{align*}
which, together with the quadratic formula and the elementary inequality
\begin{equation}\label{eq-elementary}
(x_1+x_2)^q\le x_1^q +x_2^q,\quad\forall\,q\in(0,1],\ \forall\,x_1, x_2 \in[0,\infty)
\end{equation}
with $q=\frac{1}{2}$, yields
\begin{align}\label{eq-XT}
\left\|X_T\right\|_2&\le
\frac{\sqrt{1+\delta_{2s}} \,\mu_0(1-\mu_0)(2-p)
+\mu_0 (2-p-\mu_0) \sqrt{(1-p)[\delta(p,a,\gamma)-\delta_{2s}]}}
{(2-p-\mu_0)^2 [\delta(p,a,\gamma)-\delta_{2s}]}\epsilon\nonumber\\
&\quad+\frac{p[\delta_{2s}(2-p)]^{\frac{2-p}{p}}\mu_0^2}
{2(1+\delta_{2s})^{\frac{2-2p}{p}}(2-p-\mu_0)^2 [\delta(p,a,\gamma)-\delta_{2s}]}
\frac{\sigma}{\sqrt{s}}\nonumber\\
&=c \epsilon + \widetilde{c} \frac{\sigma}{\sqrt{s}}.
\end{align}
By \eqref{eq-DSc1} and Lemma \ref{lem-TLp,prop}(i), we have
\begin{equation*}
\left\|X_{T^\complement}\right\|_p^p
\le\gamma \left\|X_{T}\right\|_{\mathrm{\mathrm{PS}_{a,p}}}+\sigma
\le\frac{a+1}{a}\gamma \|X_{T}\|_p^p+\sigma,
\end{equation*}
which, combined with Lemma \ref{lem-Tc<T} respectively in the cases
$\alpha=2/p$ and $\alpha=1/p$,
implies that
 \begin{equation}\label{eq-2/p}
\left\|X_{T^\complement}\right\|^2_2
\le\left(\frac{a+1}{a}\gamma \|X_{T}\|_2^p+ \frac{\sigma}{s^{1-\frac{p}{2}}}\right)^{\frac{2}{p}}
\end{equation}
and
 \begin{equation}\label{eq-1/p}
\left\|X_{T^\complement}\right\|_1
\le\left(\frac{a+1}{a}\gamma \|X_{T}\|_1^p+ \frac{\sigma}{s^{1-p}}\right)^{\frac{1}{p}}.
\end{equation}
Thus, combining \eqref{eq-2/p}, \eqref{eq-XT}, \eqref{eq-elementary}
with $q=p$ and the triangle inequality that
\begin{equation}\label{eq-tri'}
\left\|x^{(1)}+x^{(2)}+x^{(3)}\right\|_{\frac{2}{p}}\le \left\|x^{(1)}\right\|_{\frac{2}{p}}
+\left\|x^{(2)}\right\|_{\frac{2}{p}} +\left\|x^{(3)}\right\|_{\frac{2}{p}}
\end{equation}
for any vectors $x^{(j)}:=(x^{(j)}_1,x^{(j)}_2)$  with each $x^{(j)}_i\ge 0$,
we conclude that
\begin{align*}
\left\|X\right\|_2
&=\sqrt{\left\|X_{T}\right\|^2_2+\left\|X_{T^\complement}\right\|^2_2}\\
&\le \left\{\left(c \epsilon + \widetilde{c} \frac{\sigma}{\sqrt{s}}\right)^{p\cdot\frac{2}{p}}
+ \left[\frac{a+1}{a}\gamma \left( c \epsilon + \widetilde{c} \frac{\sigma}{\sqrt{s}}\right)^p
+ s^{\frac{p}{2}-1}\sigma\right]^{\frac{2}{p}}\right\}^{\frac{1}{2}} \\
&\le \left\{\left(c^p \epsilon^p + \widetilde{c}^p \frac{\sigma^p}{\sqrt{s}^p}\right)^{\frac{2}{p}}
+ \left(\frac{a+1}{a}\gamma c^p \epsilon^p + \frac{a+1}{a}\gamma\widetilde{c}^p \frac{\sigma^p}{\sqrt{s}^p}
+ s^{\frac{p}{2}-1}\sigma\right)^{\frac{2}{p}}\right\}^{\frac{p}{2}\cdot\frac{1}{p}} \\
&\le\left\{\left[1+\left(\frac{a+1}{a}\gamma\right)^{\frac{2}{p}}\right]^{\frac{p}{2}}c^p\epsilon^p
+ \left[1+\left(\frac{a+1}{a}\gamma\right)^{\frac{2}{p}}\right]^{\frac{p}{2}}  \widetilde{c}^p
\frac{\sigma^p}{\sqrt{s}^p} + s^{\frac{p}{2}-1}\sigma \right\}^{\frac{1}{p}},
\end{align*}
which implies \eqref{eq-bound2'}.

Furthermore, from \eqref{eq-1/p}, H\"{o}lder's inequality,
\eqref{eq-XT}, \eqref{eq-elementary} with $q=p$ therein, and \eqref{eq-tri'} with $\frac{2}{p}$ replaced by $\frac{1}{p}$,
we deduce that
\begin{align*}
\left\|X\right\|_1&
\le \left\|X_T\right\|_1+\left\|X_{T^\complement}\right\|_{1}
\le \sqrt{s }\left\|X_T\right\|_2 +\left(\frac{a+1}{a}\gamma \sqrt{s}^{p} \|X_{T}\|_2^p
+ \frac{\sigma}{s^{1-p}}\right)^{\frac{1}{p}}\\
&\le\left[\left(c^p \sqrt{s}^p \epsilon^p + \widetilde{c}^p \sigma^p \right)^{\frac{1}{p}}
+ \left(\frac{a+1}{a}\gamma c^p \sqrt{s}^p \epsilon^p + \frac{a+1}{a}\gamma \widetilde{c}^p \sigma^p
+  \frac{\sigma}{s^{1-p}}\right)^{\frac{1}{p}}\right]^{p\cdot\frac{1}{p}}\\
&\le \left\{\left[1+\left(\frac{a+1}{a}\gamma\right)^{\frac{1}{p}}\right]^p c^p \sqrt{s}^p \epsilon^p
+ \left[1+\left(\frac{a+1}{a}\gamma\right)^{\frac{1}{p}}\right]^p \widetilde{c}^p \sigma^p
+ \frac{\sigma}{s^{1-p}}\right\}^{\frac{1}{p}},
\end{align*}
which implies \eqref{eq-bound1'}.
This finishes the proof of Proposition \ref{prop-u}.
\end{proof}

Next, we prove Theorem \ref{thm-stableG}.

\begin{proof}[Proof of Theorem \ref{thm-stableG}]
Let $X_\beta:=\frac{X}{\beta}$ and define $D:=\widehat{X}_\beta-X_\beta$.
Let $S$ be the index set of the largest $s$ components of $\nabla X_\beta$ in magnitude.
Since $\widehat{X}_\beta$ is a minimizer, then, from the choice of $\beta$, we deduce that
$\|\widehat{X}_\beta\|_{\mathrm{PSV}_{a,p}}\le \|X_\beta\|_{\mathrm{PSV}_{a,p}}\le \frac{1}{2}$
and hence, by Lemma \ref{lem-TLp,prop}(iv),
$$
\|D\|_{\mathrm{PSV}_{a,p}}\le\left \|\widehat{X}_\beta\right\|_{\mathrm{PSV}_{a,p}}+ \|X_\beta\|_{\mathrm{PSV}_{a,p}}
\le 1.
$$
This, together with  Lemma \ref{lem-TLp,prop}(i), implies
 $\|\nabla D\|_{\infty}\le 1$.

We first show the cone constraint.
Using  the separability of $\|\cdot\|_{\mathrm{\mathrm{PS}_{a,p}}}$
and Lemma \ref{lem-TLp,prop}(iv),
we find that
\begin{align*}
&\left\|\left(\nabla X_\beta\right)_S\right\|_{\mathrm{\mathrm{PS}_{a,p}}}+
\left\|\left(\nabla X_\beta\right)_{S^\complement}\right\|_{\mathrm{\mathrm{PS}_{a,p}}}\\
& \quad=\|X_\beta\|_{\mathrm{PSV}_{a,p}}\ge\big\|\widehat{X}_\beta\big\|_{\mathrm{PSV}_{a,p}}=
\|\nabla X_\beta + \nabla D\|_{\mathrm{\mathrm{PS}_{a,p}}}\\
& \quad=\|(\nabla X_\beta)_S + (\nabla D)_S\|_{\mathrm{\mathrm{PS}_{a,p}}}
+\left\|(\nabla X_\beta)_{S^\complement} + (\nabla D)_{S^\complement}\right\|_{\mathrm{\mathrm{PS}_{a,p}}}\\
& \quad\ge\left\|\left(\nabla X_\beta\right)_S\right\|_{\mathrm{\mathrm{PS}_{a,p}}}
-\left\|\left(\nabla D\right)_S\right\|_{\mathrm{\mathrm{PS}_{a,p}}}
+\left\| \left(\nabla D\right)_{S^\complement}\right\|_{\mathrm{\mathrm{PS}_{a,p}}}
-\left\|\left(\nabla X_\beta\right)_{S^\complement}\right\|_{\mathrm{\mathrm{PS}_{a,p}}},
\end{align*}
which implies
\begin{equation}\label{eq-cone,D}
\left\| \left(\nabla D\right)_{S^\complement}\right\|_{\mathrm{\mathrm{PS}_{a,p}}}\le
\left\|\left(\nabla D\right)_S\right\|_{\mathrm{\mathrm{PS}_{a,p}}}
+2\left\|\left(\nabla X_\beta\right)_{S^\complement}\right\|_{\mathrm{\mathrm{PS}_{a,p}}}.
\end{equation}
Define $G:=[\partial_1 D,(\partial_2 D)^T]$.
Then $G$ has the same nonzero components as $\nabla D$ (treated as vectors)
and hence $\|G\|_{\infty}\le1$.
We use $L$ to denote the mapping of indices,
which maps the index of each nonzero component of $\nabla D$ to the corresponding  index of $G$.
By these, $|L(S)|\le|S|\le s$,
 and \eqref{eq-cone,D}, we further obtain the following cone constraint for $G$
\begin{align}\label{eq-cone,G}
\left\| G_{L(S)^\complement}\right\|_{\mathrm{\mathrm{PS}_{a,p}}}
&=\left\| (\nabla D)_{S^\complement}\right\|_{\mathrm{\mathrm{PS}_{a,p}}}
\le \left\|\left(\nabla D\right)_S\right\|_{\mathrm{\mathrm{PS}_{a,p}}}
+2\left\|\left(\nabla X_\beta\right)_{S^\complement}\right\|_{\mathrm{\mathrm{PS}_{a,p}}}\nonumber\\
&=\left\|G_{L(S)}\right\|_{\mathrm{\mathrm{PS}_{a,p}}}
+2\left\|\left(\nabla X_\beta\right)_{S^\complement}\right\|_{\mathrm{\mathrm{PS}_{a,p}}}.
\end{align}

Next, we prove the tube constraint.
By $\|y_\beta-\mathcal{M}(X_\beta)\|_2\le\epsilon_\beta$ and
$\|y_\beta-\mathcal{M}(\widehat{X}_\beta)\|_2\le\epsilon_\beta$,
we have
\begin{equation}\label{eq-tube1}
\|\mathcal{M}(D)\|_2\le \left\|y_\beta-\mathcal{M}(X_\beta)\right\|_2+
\left\|y_\beta-\mathcal{M}(\widehat{X}_\beta)\right\|_2\le2\epsilon_\beta.
\end{equation}
We let $\mathcal{A}$ be identified by matrices $\{A_k\}_{k=1}^{K}$ and let
$\mathcal{A}_0$ and $\mathcal{A}^0$  be the linear operators identified by
the padded matrices $\{(A_k)_0,(A_k)^0\}_{k=1}^{K}$ in the way of  \eqref{eq-def,M0}.
Applying Lemma \ref{lem-M0},  we  have, for any $k\in\{1,\ldots,K\}$,
\begin{equation}\label{eq-Ak1}
\left|\left\langle A_k, G\right\rangle\right|^2
=\left|\left\langle (A_k)^0-(A_k)_0, [D,D^T]\right\rangle\right|^2.
\end{equation}
Thus, using \eqref{eq-tube1} and  \eqref{eq-Ak1},
we obtain the tube constraint for $G$,
\begin{align*}
\left\| \mathcal{A} \left(G\right) \right\|_2^2
=\sum_{k=1}^{K} \left|\left\langle A_k,G\right\rangle\right|^2
= \sum_{k=1}^{K}\left|\left\langle (A_k)^0-(A_k)_0, [D,D^T]\right\rangle\right|^2
=\left\|\mathcal{M}(D)\right\|_2^2\le 4\epsilon_\beta^2.
\end{align*}
Applying this, $\|G\|_{\infty}\le1$, \eqref{eq-cone,G}, and Proposition  \ref{prop-u} with $X=G$,
$\mathcal{M}=\mathcal{A}$, $S=L(S)$, $\epsilon=2\epsilon_\beta$,
$\sigma=2\|(\nabla X_\beta)_{S^\complement}\|_{\mathrm{\mathrm{PS}_{a,p}}}$, and $\gamma=1$ therein,
 we conclude that
\begin{equation*}
\left\|\nabla X_\beta-\nabla \widehat{X}_\beta\right\|^p_2 =\|G\|^p_2\le
2^p c_1  \max\left\{\frac{\|(\nabla X_\beta)_{S^\complement}\|^p_{\mathrm{\mathrm{PS}_{a,p}}}}{\sqrt{s}^{p}},
\frac{2^{1-p}\|(\nabla X_\beta)_{S^\complement}\|_{\mathrm{\mathrm{PS}_{a,p}}}}{\sqrt{s}^{2-p}}\right\}
+ 2^p c_2\epsilon_\beta^p
\end{equation*}
and
\begin{align*}
\left\|\nabla X_\beta-\nabla \widehat{X}_\beta\right\|^p_1
&=\|G\|^p_1\\
&\le  2^p c_3\max\left\{\left\|(\nabla X_\beta)_{S^\complement}\right\|_{\mathrm{\mathrm{PS}_{a,p}}}^p,
\frac{2^{1-p}\|(\nabla X_\beta)_{S^\complement}\|_{\mathrm{\mathrm{PS}_{a,p}}}}{s^{1-p}}\right\}
+ 2^p c_4 \sqrt{s}^p\epsilon_\beta^p ,
\end{align*}
where $c_1,$ $c_2,$ $c_3$, and $c_4$ are the same as in Proposition  \ref{prop-u}.
Furthermore, using Lemma \ref{lem-TLp,prop}(ii),
we obtain
\begin{equation*}
\left\|\nabla X-\nabla \widehat{X}\right\|^p_2\le
2^p c_1  \max\left\{\left[\frac{a+1}{a}\beta^{1-p}\frac{\|(\nabla X_\beta)_{S^\complement}\|^{p}_{p}}{\sqrt{s}}\right]^p,
\frac{a+1}{a}\frac{2^{1-p}\|(\nabla X)_{S^\complement}\|^p_{p}}{\sqrt{s}^{2-p}}\right\}
+ C_2\epsilon^p
\end{equation*}
and
\begin{align*}
\left\|\nabla X-\nabla \widehat{X}\right\|^p_1 \le
2^p c_3\max\left\{\left[\frac{a+1}{a}\beta^{1-p} \|(\nabla X_\beta)_{S^\complement}\|^{p}_{p} \right]^p,
\frac{a+1}{a}\frac{2^{1-p}\|(\nabla X)_{S^\complement}\|^p_{p}}{s^{1-p}}\right\}
+ C_4 \sqrt{s}^p\epsilon^p,
\end{align*}
which implies both \eqref{eq-bound2} and \eqref{eq-bound1}.
This finishes the proof of Theorem \ref{thm-stableG}.
\end{proof}

Using Theorem \ref{thm-stableG}, we further show Theorem \ref{thm-stable}.

\begin{proof}[Proof of Theorem \ref{thm-stable}]
Let $S$ be the index set of the largest $s$ components of $\nabla X$ in magnitude.
Applying Lemma \ref{lem-sobolev}, $\|\mathcal{B}(D)\|_2\le\|\mathcal{M}(D)\|_2\le2\epsilon$,
and \eqref{eq-bound1}, we conclude that
\begin{align*}
\left\|X-\widehat{X}\right\|_2&\lesssim
\left(\frac{\|\nabla X-\nabla\widehat{X}\|_{1}}{\sqrt{s}}\right)\log\left(\frac{N^2}{s}\right)  + 2\epsilon\\
&\lesssim  \log\left(\frac{N^2}{s}\right)\left[\max\left\{\left[ \frac{\|\nabla X-(\nabla X)_{S}\|^{p}_{p}}{ \sqrt{s}\beta^{p-1}} \right]^p,
\left[2\left(\frac{a+1}{a}\right)\right]^{1-p} \frac{\|\nabla X-(\nabla X)_{S}\|^p_{p}}{\sqrt{s}^{2-p}}\right\}
+  \epsilon^p\right]^{\frac{1}{p}} + 2\epsilon\\
&\lesssim  \log\left(\frac{N^2}{s}\right)\left[\max\left\{\left[ \frac{\|\nabla X-(\nabla X)_{S}\|^{p}_{p}}{ \sqrt{s}\beta^{p-1}} \right]^p,
\frac{\|\nabla X-(\nabla X)_{S}\|^p_{p}}{\sqrt{s}^{2-p}}\right\}
+  \epsilon^p\right]^{\frac{1}{p}} ,
\end{align*}
which implies \eqref{eq-stablePSV}. This finishes the proof of Theorem \ref{thm-stable}.
\end{proof}

Finally, we prove Proposition \ref{prop-sharp} by a constructive proof inspired
by the proof of \cite[Theorem 1.4]{ZL19}.

\begin{proof}[Proof of Theorem \ref{prop-sharp}]
Let $\varepsilon\in(0,1)$, $k\in\mathbb{N}$ be  such that
$k\ge\frac{4}{\varepsilon}$, and $s:=4k$.
Then it suffices to show that there exist two images
$X_0,\widehat{X}\in\mathbb{R}^{N\times N}$ with $X_0 \neq\widehat{X}$
and a linear operator $\mathcal{A}$ satisfying $\delta_{2s}<\delta_p+\varepsilon$
such that $X_0$ has an $s$-sparse discrete gradient and
$\mathcal{M}( X_0)=\mathcal{M}(\widehat{X})$,
but $\|\widehat{X}\|_{\mathrm{TV}_p}\le\|X_0\|_{\mathrm{TV}_p}$,
where $\delta_p:=\frac{\zeta_0}{2-p-\zeta_0}$ with $\zeta_0$ being
the unique positive solution of \eqref{eq-equation,p} and
$\mathcal{M}$ is defined by setting, for any matrix $Z\in\mathbb{R}^{N\times N}$,
$$\mathcal{M}(Z):=(\mathcal{A}^0-\mathcal{A}_0)
\left(\left[Z,Z^T\right]\right).$$

Let $\widetilde{\ell}:=\frac{s}{\zeta_0}$ and $\ell$ be the largest integer smaller than
$\widetilde{\ell}$ such that $l:=\ell/4\in\mathbb{N}$.
Then $0\le\widetilde{\ell}-\ell<4$. By $\zeta_0\in(0,1)$, we also have $s<\widetilde{\ell}$.
Define $\overline{\Sigma}_0:=[\Sigma_0,\Sigma_0]$, where

\begin{equation*}
\Sigma_0:=\frac{1}{\sqrt{s+\ell \zeta_0^{\frac{2}{p}}}}\times
\begin{aligned}
&\hspace{2cm}
\begin{matrix}
{\scriptstyle k+1}& \qquad \qquad \quad\ & {\scriptstyle k+l+1}\\
\downarrow & \qquad \qquad \quad\ & \downarrow
\end{matrix} \\
& \begin{pmatrix} 0 & 1  & 1 & \cdots & 1 & 0 &0 &\cdots & 0 & 0& \cdots&0\\
0 & 0  & -1 &\cdots & -1 & 0 & 0 &\cdots &0 & 0& \cdots&0\\
\vdots & \vdots  & \vdots & \ddots & \vdots &\vdots &\vdots&\ddots & \vdots & \vdots& \ddots&\vdots\\
0 & 0  & 0 &\cdots & 0 & 0 & 0 & \cdots & 0 & 0& \cdots&0\\
0 & -1  & 0 &\cdots & 0 & \zeta_0^{\frac{1}{p}}& \zeta_0^{\frac{1}{p}} & \cdots & \zeta_0^{\frac{1}{p}}& 0& \cdots&0\\
0 & 0  & 0 &\cdots & 0 & 0 & -\zeta_0^{\frac{1}{p}} & \cdots & -\zeta_0^{\frac{1}{p}}& 0& \cdots&0\\
0 & 0  & 0 &\cdots & 0 & 0 & 0 & \ldots & 0& 0& \cdots&0\\
\vdots & \vdots  & \vdots &\ddots & \vdots & \vdots & \vdots& \ddots & \vdots& \vdots& \ddots&\vdots\\
0 & 0  & 0& \cdots& 0 & -\zeta_0^{\frac{1}{p}} & 0 & \cdots & 0& 0& \cdots&0\\
0 & 0  & 0 &\cdots & 0 & 0 & 0 & \cdots & 0 & 0& \cdots&0\\
\vdots & \vdots  & \vdots & \ddots & \vdots &\vdots &\vdots&\ddots & \vdots & \vdots& \ddots&\vdots\\
0 & 0  & 0 &\cdots & 0 & 0 & 0 & \cdots & 0 & 0& \cdots&0
\end{pmatrix}
\begin{matrix}
\quad \\
\leftarrow {\scriptstyle k+1}\\
\quad \\  \quad \\
\quad \\  \quad \\
\leftarrow {\scriptstyle k+l+1}
\end{matrix}
\end{aligned},
\end{equation*}
and define
$$
\mathcal{A}:\left\{
\begin{array}{rll}
\mathbb{R}^{(N-1)\times 2N}&\to&\mathbb{R}^{2N(N-1)},\\
Z&\mapsto& \mathcal{A}(Z):=\sqrt{\frac{2-p}{2-p-\zeta_0}}\left(Z-\left\langle
\overline{\Sigma}_0, Z\right\rangle \overline{\Sigma}_0\right).
\end{array}\right.
$$
Then, by $s=4k$ and $\ell=4l$, we have $\|\overline{\Sigma}_0\|_2=1$
and, moreover,
for any $2s$-sparse   matrix $Z\in\mathbb{R}^{(N-1)\times 2N}$,
\begin{align}\label{eq-sharp1x}
\left\| \mathcal{A} (Z)\right\|_2^2
&=\frac{2-p}{2-p-\zeta_0}\left(\left\|Z\right\|_2^2
-\left|\left\langle \overline{\Sigma}_0, Z\right\rangle\right|^2\right).
\end{align}
From  the Cauchy--Schwartz inequality,
$\zeta_0<1-\frac{p}{2}<1$, and $s>4$,
we infer that
\begin{align*}
0&\le\left|\left\langle \overline{\Sigma}_0, Z\right\rangle\right|^2
\le \left\|(\overline{\Sigma}_0)_{\mathrm{supp}\,(Z)}\right\|_2^2\left\|Z\right\|_2^2\\
&\le \frac{s+s \zeta_0^{\frac{2}{p}}}{s+\ell \zeta_0^{\frac{2}{p}}}\left\|Z\right\|_2^2
= \frac{1+ \zeta_0^{\frac{2}{p}}}{1+ \zeta_0^{\frac{2}{p}-1}}
\frac{s+s\zeta_0^{\frac{2}{p}-1}}{s+\ell\zeta_0^{\frac{2}{p}}} \left\|Z\right\|_2^2  \\
&=\frac{1+ \zeta_0^{\frac{2}{p}}}{1+\zeta_0^{\frac{2}{p}-1}}
\frac{1}{1-\frac{(\widetilde{\ell}-\ell)s^{\frac{2}{p}-1}}
{\widetilde{\ell}^{\frac{2}{p}}+\widetilde{\ell} s^{\frac{2}{p}-1}}}\left\|Z\right\|_2^2
=\frac{1+ \zeta_0^{\frac{2}{p}}}{1+\zeta_0^{\frac{2}{p}-1}}
\frac{1}{1-\frac{\widetilde{\ell}-\ell}
{s(\zeta_0^{-\frac{2}{p}}+\zeta_0^{-1})}}\left\|Z\right\|_2^2.
\end{align*}
Thus, when $\widetilde{\ell}-\ell\in(0,4)$,
\begin{align}\label{eq-sharp,RIP}
0&\le\left|\left\langle \overline{\Sigma}_0, Z\right\rangle\right|^2
<\frac{1+\zeta_0^{\frac{2}{p}}}{1+ \zeta_0^{\frac{2}{p}-1}}
\frac{1}{1-\frac{2}{s}}\left\|Z\right\|_2^2\notag\\
&=\frac{2\zeta_0}{2-p}\left(1+\frac{2}{s-2}\right)\left\|Z\right\|_2^2
\le\left(\frac{2\zeta_0}{2-p}+\frac{4}{s}\right)\left\|Z\right\|_2^2
\end{align}
and, when $\widetilde{\ell}=\ell$,
\begin{align}\label{eq-sharp,RIP'}
0&\le\left|\left\langle \overline{\Sigma}_0, Z\right\rangle\right|^2
\le \frac{1+ \zeta_0^{\frac{2}{p}}}{1+\zeta_0^{\frac{2}{p}-1}}\left\|Z\right\|_2^2
= \frac{2\zeta_0}{2-p} \left\|Z\right\|_2^2.
\end{align}
By inserting \eqref{eq-sharp,RIP} and \eqref{eq-sharp,RIP'} to \eqref{eq-sharp1x} and using
$\zeta_0<1-\frac{p}{2}$ and $s\ge\frac{16}{\varepsilon}$,
we further obtain
\begin{align*}
(1+\delta_p)\left\|Z\right\|_2^2&=\frac{2-p}{2-p-\zeta_0} \left\|Z\right\|_2^2 \ge
\left\|\mathcal{A}(Z)\right\|_2^2\\
&>\frac{2-p}{2-p-\zeta_0}\left(1-\frac{2\zeta_0}{2-p}-\frac{4}{s}\right)\left\|Z\right\|_2^2\\
&>\left(1-\frac{\zeta_0}{2-p-\zeta_0}-\frac{8}{s}\right)\left\|Z\right\|_2^2
\ge\left(1-\delta_p-\frac{\varepsilon}{2}\right)\left\|Z\right\|_2^2,
\end{align*}
which implies that the  operator
$\mathcal{A}$ has the RIP of order $2s$ with
$\delta_{2s}\le\delta_p+\frac{\varepsilon}{2} < \delta_p +\varepsilon$.

Now, we define two diagonal matrices $X_0:=[X_{i,j}]\in\mathbb{R}^{N\times N}$ and
$\widehat{X}:=[\widehat{X}_{i,j}]\in\mathbb{R}^{N\times N}$ by setting
$$
X_0:=\frac{1}{\sqrt{s+\ell \zeta_0^{\frac{2}{p}}}}E_0\ \ \text{and}\ \
\widehat{X}:=\frac{1}{\sqrt{s+\ell \zeta_0^{\frac{2}{p}}}}\widehat{E},
$$
where
\begin{equation*}
E_0:=
\begin{aligned}
&\hspace{2.2cm}
\begin{matrix}
{\scriptstyle k+1}\\
\downarrow
\end{matrix} \\
& {\setlength{\arraycolsep}{2pt}\begin{pmatrix}
0 & 0  & 0 & \cdots & 0 & 0 &\cdots & 0 \\
0 & -1 & -1 &\cdots & -1 & 0 &\cdots &0 \\
0 & -1 & 0 &\cdots & 0& 0 &\cdots &0 \\
\vdots & \vdots  & \vdots & \ddots & \vdots  &\vdots&\ddots & \vdots \\
0 & -1  & 0 &\cdots & 0 & 0 & \cdots & 0  \\
0 & 0  & 0 &\cdots & 0 & 0 & \cdots & 0 \\
\vdots & \vdots  & \vdots &\ddots & \vdots & \vdots & \ddots & \vdots\\
0 & 0  & 0 &\cdots & 0 & 0 & \cdots & 0
\end{pmatrix}}
\end{aligned}
\end{equation*}
and
\begin{equation*}
\widehat{E}:=
\begin{aligned}
&\hspace{1.4cm}
\begin{matrix}
{\scriptstyle k+2}& \qquad \  \quad & {\scriptstyle k+l+1}\\
\downarrow & \qquad \  \quad & \downarrow
\end{matrix} \\
&
{\setlength{\arraycolsep}{2pt}\begin{pmatrix}
0 & \cdots  & 0 & 0 & 0 & \cdots & 0 & 0 &\cdots & 0\\
\vdots & \ddots & \vdots & \vdots&\vdots &\ddots & \vdots & \vdots &\ddots &\vdots \\
0 & \cdots  & 0 & 0 & 0 & \cdots & 0 & 0 &\cdots & 0\\
0 & \cdots & 0 & \zeta_0^{\frac{1}{p}}&\zeta_0^{\frac{1}{p}} & \cdots &\zeta_0^{\frac{1}{p}} & 0 &\cdots & 0\\
0 & \cdots  & 0 & \zeta_0^{\frac{1}{p}}& 0 & \cdots & 0 & 0 &\cdots & 0\\
\vdots & \ddots  & \vdots &\vdots & \vdots &\ddots & \vdots & \vdots & \ddots & \vdots\\
0 & \cdots  & 0 & \zeta_0^{\frac{1}{p}}& 0&\cdots & 0 & 0 &\cdots & 0\\
0 & \cdots  & 0 & 0&0& \cdots &0 & 0 &\cdots & 0\\
\vdots & \ddots  & \vdots &\vdots & \vdots& \ddots & \vdots & \vdots & \ddots & \vdots\\
0 & \cdots  & 0 & 0 & 0 & \cdots & 0 & 0 &\cdots & 0
\end{pmatrix}}
\end{aligned}
\end{equation*}
By a direct   calculation, we have
\begin{equation*}
\partial_2 E_0=
\begin{aligned}
&\hspace{2.1cm}
\begin{matrix}
{\scriptstyle k+1}\\
\downarrow
\end{matrix} \\
& {\setlength{\arraycolsep}{3pt} \begin{pmatrix}
0 & 0  & 0 & \cdots & 0 & 0 &\cdots & 0 \\
-1 & 0 & 0 &\cdots &  1 & 0 &\cdots &0 \\
-1 & 1 & 0 &\cdots & 0& 0 &\cdots &0 \\
\vdots & \vdots  & \vdots & \ddots & \vdots  &\vdots&\ddots & \vdots \\
-1 & 1  & 0 &\cdots & 0 & 0 & \cdots & 0  \\
0 & 0  & 0 &\cdots & 0 & 0 & \cdots & 0 \\
\vdots & \vdots  & \vdots &\ddots & \vdots & \vdots & \ddots & \vdots\\
0 & 0  & 0 &\cdots & 0 & 0 & \cdots & 0
\end{pmatrix}}
\end{aligned}
\end{equation*}
and
\begin{equation*}
\partial_2\widehat{E}=
\begin{aligned}
&\hspace{1.4cm}
\begin{matrix}
{\scriptstyle k+1}& \qquad  \qquad \quad \ & {\scriptstyle k+l+1}\\
\downarrow & \qquad  \qquad \quad \ & \downarrow
\end{matrix} \\
&
{\setlength{\arraycolsep}{2pt}\begin{pmatrix}
0 & \cdots  & 0 & 0 & 0  & 0 & \cdots & 0 & 0 &0 &\cdots & 0\\
\vdots & \ddots & \vdots & \vdots &\vdots&\vdots &\ddots & \vdots & \vdots & \vdots &\ddots &\vdots \\
0 & \cdots  & 0 & 0 & 0 & 0 & \cdots & 0 & 0 & 0 &\cdots & 0\\
0 & \cdots & 0 & \zeta_0^{\frac{1}{p}}& 0  & 0& \cdots & 0 & -\zeta_0^{\frac{1}{p}} & 0 &\cdots & 0\\
0 & \cdots  & 0 & \zeta_0^{\frac{1}{p}}& -\zeta_0^{\frac{1}{p}} & 0 & \cdots & 0&0 &0&\cdots & 0\\
\vdots & \ddots  & \vdots &\vdots & \vdots &\vdots&\ddots & \vdots & \vdots&\vdots & \ddots & \vdots\\
0 & \cdots  & 0 & \zeta_0^{\frac{1}{p}}& -\zeta_0^{\frac{1}{p}}& 0&\cdots & 0 & 0 &0 &\cdots & 0\\
0 & \cdots  & 0 & 0&0& 0&\cdots &0 & 0 & 0 &\cdots & 0\\
\vdots & \ddots  & \vdots &\vdots & \vdots &\vdots & \ddots & \vdots & \vdots &\vdots & \ddots & \vdots\\
0 & \cdots  & 0 & 0 & 0 & 0& \cdots & 0 & 0 &0 &\cdots & 0
\end{pmatrix}}
\end{aligned}
\end{equation*}
as well as
$\partial_1 E_0 = (\partial_2 E_0)^T$ and $\partial_1 \widehat{E} = (\partial_2 \widehat{E})^T.$
Note that
$\partial_1(\widehat{X}-X_0)=\Sigma_0$ and $\partial_2 (\widehat{X}-X_0 )=\Sigma_0^T$.
Applying  this, Lemma \ref{lem-M0},
and  $\|\overline{\Sigma}_0\|_2=1$, we obtain
\begin{align*}
\mathcal{M}(\widehat{X})-\mathcal{M}(X_0)
&=\left(\mathcal{A}^0-\mathcal{A}_0\right)\left(\left[\widehat{X}-X_0,(\widehat{X}-X_0)^T\right]\right)\\
&=  \mathcal{A} \left(\left[\partial_1\left(\widehat{X}- X_0\right),
\left[\partial_2\left(\widehat{X}- X_0\right)\right]^T\right]\right)
=\mathcal{A} \left(\overline{\Sigma}_0 \right)=0.
\end{align*}
However, since $\ell\zeta_0<s$, we have
$\|\nabla \widehat{X}\|_{p}^p/\|\nabla X_0\|_{p}^p =\ell\zeta_0/s<1$
and hence $\|\nabla \widehat{X}\|_{p}^p<\|\nabla X_0\|_{p}^p ,$
which completes the proof of Theorem \ref{prop-sharp}.
\end{proof}

\section{Algorithm}\label{sec-alg}

This section consists of two subsections.
It is well known that the \emph{iteratively re-weighted least squares}
(IRLS) method is a classical algorithm
for the $\ell_p$ ($p\in(0,1]$)  minimization; see, for example,  \cite{DDFG10,LXY13}.
Recently, a modified IRLS method was employed  in \cite{lcgy}
to solve the unconstrained two-parameter $P_{a,p}$ minimization.
Inspired by this, in Subsection \ref{sec-Alo}, based on the PSV$_{a,p}$
we design an iteratively re-weighted least squares algorithm
IRLSPSV to solve the following unconstrained PSV$_{a,p}$ minimization
\begin{equation}\label{eq-unproblem}
\min_{X\in\mathbb{R}^{N\times N}} Q(X) :=
\min_{X\in\mathbb{R}^{N\times N}}
\lambda\left\|X\right\|_{\mathrm{PSV}_{a,p}}+ \frac{1}{2}\left\|\Phi X-y\right\|_2^2,
\end{equation}
where  $\lambda\in(0,\infty)$ is the regularity parameter, $\Phi$ is a
measurement matrix, and $y$ is a measurement.
Subsection \ref{subsec-convergIRLS} is devoted to establishing
some related convergence results.

In the following algorithm and subsequent experiments,
we focus on the isotropic version of the
semi-norm $\|\cdot\|_{\mathrm{PSV}_{a,p}}$, namely
$\|X\|_{\mathrm{PSV}_{a,p}}:=\|X\|^{(\mathrm{i})}_{\mathrm{PSV}_{a,p}}$ [see \eqref{eq-def-iTTVp}],
due to its higher directional neutrality
compared to the anisotropic case.

\subsection{IRLSPSV Algorithm for Unconstrained PSV$_{a,p}$ Minimization}\label{sec-Alo}

For any given $a\in(0,\infty)$ and $p\in(0,1]$,
we define a \emph{functional $\mathcal{G}_{a,p}$} by setting,
for any image $X\in {\mathbb R}^{N\times N}$, any
$\omega:=[\omega_{i,j}]\in\mathbb{R}^{N\times N}$
with each $\omega_{i,j}>0$, and any $\epsilon>0$,
\begin{equation*}
\mathcal{G}_{a,p}(X,\omega,\epsilon)
:=\frac{p(a+1)}{2}\left[\sum_{i,j=1}^{N} \frac{|(\nabla X)_{i,j}|^2 +\epsilon^2}
{(a+|(\nabla X)_{i,j}|^p)^{\frac 2p}}\omega_{i,j}
+ \frac{2-p}{p}\omega_{i,j} ^{-\frac{p}{2-p}}\right].
\end{equation*}
Here, and thereafter,
$|(\nabla X)_{i,j}|:=\sqrt{(\nabla_1 X)_{i,j}^2+(\nabla_2 X)_{i,j}^2}$.

\begin{algorithm}[H]
\caption{IRLSPSV for  unconstrained PSV$_{a,p}$ minimization }\label{alg1}
\hspace*{1em}\textbf{Input:} $\Phi\in {\mathbb R}^{K\times N^2}$, $y\in {\mathbb R}^K$,
$\delta\in(0,1]$, and $s\in\mathbb N$ \\
\hspace*{1em}\textbf{Define:}  $k_{\mathrm{out}}\in\mathbb{N}$ and $\tau\ge1$\\
\hspace*{1em}\textbf{Initialize:} $X^0=\mathbf{0}$, $\epsilon_0=1$, and $n=0$ \\
\hspace*{1em}\textbf{While} $n<k_{\mathrm{out}}$ and \textbf{not converged}\\
\hspace*{2em} $w_{i,j}^{n}= \left[\left|(\nabla X^n)_{i,j}\right|^2+\epsilon_n^2\right]^{\frac{p-2}{2}}$ \\
\hspace*{2em} $X^{n+1}=\displaystyle\mathop{\arg\min\,}_{X\in\mathbb{R}^{N\times N}}
  \lambda (a+1)\sum_{i,j=1}^{N}w_{i,j}^{n}\frac{|(\nabla X)_{i,j}|^2}{a+|(\nabla X)_{i,j}|^p}+ \frac{1}{2}\|\Phi X-y\|_2^2$ \\
\hspace*{2em} $\epsilon_{n+1}=\min\left\{\epsilon_n ,  \delta r(\nabla X^{n+1})_{s+1}\right\}$ \\
\hspace*{2em} \textbf{if}    $\|\nabla X^{n+1}-\nabla X^{n}\|_2/\max\{1,\|\nabla X^{n}\|_2\}>\tau $\\
\hspace*{2em} \textbf{then return} $X^n$  \quad \textbf{break}\\
\hspace*{2em} $n=n+1$ \\
\hspace*{1em} \textbf{End while}\\
\hspace*{1em}\textbf{Output:} $X_{\mathrm{new}}=X^{n+1}$
\end{algorithm}

To solve \eqref{eq-unproblem}, the IRLSPSV algorithm adopts an alternating minimization
and  an approximation process from \cite{lcgy}​, as detailed in Algorithm \ref{alg1}.
Here, $r(\nabla X^{n+1})_{s+1}$ denotes the $(s+1)$-th largest component  in magnitude of
 the decreasing rearrangement of $\nabla X^{n+1}$ (treated as a vector).
Notably,  if $\epsilon_{n_0}=0$ for some $n_0\in\mathbb{N}$,
we stop the algorithm and define $X^k:=X^{n_0-1}$ and $w^k:=w^{n_0-1}$ for all $k\ge n_0$.
Then   $\{w^n\}_{n\in\mathbb{N}}$ is $\ell_\infty$-bounded because
$\|w^{n}\|_\infty \le  \epsilon_{n_0 -1}^{-(2-p)}$.

In each iteration, to update $X$, we need to
solve an unconstrained weighted sub-problem
\begin{equation}\label{eq-subproblem}
\min_{X\in \mathbb{R}^{N\times N}}  F_{w}(X)
:=\min_{X\in \mathbb{R}^{N\times N}}  \lambda (a+1)\sum_{i,j=1}^{N} w_{i,j}\frac{|(\nabla X)_{i,j}|^2}{a+|(\nabla X)_{i,j}|^p}
+ \frac{1}{2}\left\|\Phi X-y\right\|_2^2,
\end{equation}
where  $w:=[w_{i,j}]$ is an updated weight.
Although it is convex, since its explicit solution is hard to obtain,
we employ the difference of convex functions algorithm (DCA) (see, for example,
\cite{DL98}). To be precise, define
\begin{equation*}
\Psi_{w}(X) :=\sum_{i,j=1}^{N} w_{i,j} \frac{|(\nabla X)_{i,j}|^{p+2}}{a(a+|(\nabla X)_{i,j}|^p)},
\ \ \forall\,X:=[X_{i,j}]\in \mathbb{R}^{N\times N}.
\end{equation*}
Then $F_{w}$ has a DC (difference of convex functions) decomposition
\begin{equation}\label{eq-DCdecom}
F_{w} = G_{w}-H_{w},
\end{equation}
where,
for any $X:=[X_{i,j}]\in\mathbb{R}^{N\times N}$,
\begin{equation}\label{eq-def-Gw}
G_{w}(X):= \frac{\lambda (a+1)}{a}\sum_{i,j=1}^{N}w_{i,j}|(\nabla X)_{i,j}|^2
+ \frac{1}{2}\left\|\Phi X-y\right\|_2^2 + c\|X\|_2^2
\end{equation}
and
\begin{equation}\label{eq-def-Hw}
H_{w}(X):=\frac{\lambda (a+1)}{a} \Psi_{w}(X) + c\|X\|_2^2.
\end{equation}
The additional term $c\|X\|_2^2$ with $c\in(0,\infty)$  is used
to promote the convexity of $G_{w}$ and $H_{w}$.
Moreover, as both $G_{w}$ and $H_w$ are differentiable,
we can calculate their gradients
\begin{align}\label{eq-PG}
\partial G_{w}(X)
=\frac{2\lambda(a+1)}{a}\nabla^*\left(w\odot \nabla_1 X,w\odot \nabla_2 X\right)
+\Phi^T \Phi X-\Phi^T y + 2c X
\end{align}
and
\begin{align}\label{eq-PH}
\partial H_{w}(X)
&=2cX+\frac{2\lambda(a+1)}{a}\nabla^*\left(w\odot \nabla_1 X,w\odot \nabla_2 X\right)\nonumber\\
&\quad-2\lambda(a+1) \nabla^*\left(\sum_{i,j=1}^{N} \frac{w_{i,j} (\nabla_1 X)_{i,j}}{a+|(\nabla X)_{i,j}|^p}E_{i,j},
\sum_{i,j=1}^{N} \frac{w_{i,j} (\nabla_2 X)_{i,j}}{a+|(\nabla X)_{i,j}|^p}E_{i,j}\right)\nonumber\\
&\quad+\lambda p(a+1) \nabla^*\left(\sum_{i,j=1}^{N}\frac{w_{i,j} |(\nabla_1 X)_{i,j}|^p(\nabla_1 X)_{i,j}}{(a+|(\nabla X)_{i,j}|^p)^2}E_{i,j},
\sum_{i,j=1}^{N}\frac{w_{i,j} |(\nabla_2 X)_{i,j}|^p(\nabla_2 X)_{i,j}}{(a+|(\nabla X)_{i,j}|^p)^2}E_{i,j}\right),
\end{align}
where $E_{i,j}$ denotes the matrix with 1 at the $(i,j)$-entry and 0 elsewhere,
$\odot$ denotes the \emph{element-wise multiplication},
 and $\nabla^*$ denotes the adjoint operator of the gradient.

\begin{algorithm}[H]
\caption{DCA for the unconstrained weighted sub-problem  \eqref{eq-subproblem}}
\label{alg2}
\hspace*{1em}\textbf{Input:} $\Phi\in {\mathbb R}^{K\times N^2}$, $y\in {\mathbb R}^K$,
and $w^n \in \mathbb{R}^{N\times N}$\\
\hspace*{1em}\textbf{Define:} $\varepsilon_{\mathrm{mid}}>0$ and  $k_{\mathrm{mid}}\in\mathbb{N}$\\
\hspace*{1em}\textbf{Initialize:} $X^0=X^n$ \\
\hspace*{1em}\textbf{For} $k=0,1,2,\ldots,k_{\mathrm{mid}}$ \quad \textbf{do} \\
\hspace*{2em} $V^k =\partial H_{w^n}(X^k) $ \\
\hspace*{2em} $X^{k+1}=\displaystyle\mathop{\arg\min\,}\left\{X\in \mathbb{R}^{N^2} :\ G_{w^n}(X)-\langle X, V^k\rangle\right\}$ \\
\hspace*{2em} \textbf{if} $\|X^{k+1}-X^k\|_{2}/\max\{\|X^k\|_{2},1\}<\varepsilon_{\mathrm{mid}}$ \quad \textbf{then break} \\
\hspace*{1em}\textbf{Output:} $X_{\mathrm{new}}=X^{k+1}$
\end{algorithm}

The DCA for the unconstrained weighted sub-problem \eqref{eq-subproblem} is presented
 in  Algorithm \ref{alg2}.
In Algorithm \ref{alg2},  $G_{w^n}$ and $H_{w^n}$ are defined as in  \eqref{eq-def-Gw} and \eqref{eq-def-Hw}
with $w$ therein replaced by $w^n$. To update $X^{k+1}$,
we use the primal-dual (PD) algorithm (see, for example, \cite{SJP12}).
To be precise, let both $w^n:=[w^n_{i,j}]$ and $V^k$ be the updated
values generated in Algorithms \ref{alg1} and \ref{alg2}.
Then $G_{w^n}(X)-\langle X, V^k\rangle$ can be rewritten as
$f_1(u)+f_2(Z)+g(X),$ where
$$
f_1(u):=\frac{1}{2}\left\|u-y\right\|_2^2, \ \
f_2(Z):=\frac{\lambda (a+1)}{a}\sum_{i,j=1}^{N}w^n_{i,j}|Z_{i,j}|^2,\ \ \text{and}\ \
g(X):=c\|X\|_2^2 -\langle X, V^k\rangle
$$
with
$u=\Phi X$ and  $Z=\nabla X.$
The PD algorithm consists of four iterations:
\begin{align*}
\begin{cases}
P^{n+1}=\mathrm{Prox}_\sigma f_1^*\left(P^n+\sigma \Phi\overline{X}^n\right)\\
Q^{n+1}=\mathrm{Prox}_\sigma f_2^*\left(Q^n+\sigma \nabla\overline{X}^n\right)\\
X^{n+1}=\mathrm{Prox}_\tau g \left(X^n-\tau\Phi^T P^{n+1}-\tau\nabla^* Q^{n+1}\right)\\
\overline{X}^{n+1}=X^{n+1}+\theta(X^{n+1}-X^n)
\end{cases},
\end{align*}
where Prox denotes the classical proximal operator and
$f_1^*$ and $f_2^*$ denote the conjugate functions of $f_1$ and $f_2$, respectively.
We present the concrete PD algorithm in Algorithm \ref{alg3}.

\begin{algorithm}
\caption{PD for approximate solution $X^{k+1}$}
\label{alg3}
\hspace*{1em}\textbf{Input:} $\Phi\in {\mathbb R}^{K\times N^2}$, $y\in {\mathbb R}^K$,
$w^n:=[w^n_{i,j}]\in \mathbb{R}^{N\times N}$, $V^k\in \mathbb{R}^{N\times N}$, and $X^k\in \mathbb{R}^{N\times N}$\\
\hspace*{1em}\textbf{Define:} $\varepsilon_{\mathrm{inn}}>0$, $\theta>0$, $\sigma>0$,
$\tau>0$, and  $k_{\mathrm{inn}}\in\mathbb{N}$\\
\hspace*{1em}\textbf{Initialize:} $w:=w^n$,
$\overline{X}^0=X^k$, $P^0= \mathbf{0}$, $Q_x^0=\mathbf{0}$,
$Q_y^0=\mathbf{0}$, and $X_{\mathrm{prev}}^0=X^k$\\
\hspace*{1em}\textbf{For} $\ell=0,1,2,\ldots$ \quad \textbf{do} \\
\hspace*{2em} $P^{\ell+1}=  \frac{1}{1+\sigma}\left[P^\ell+\sigma (\Phi \overline{X}^\ell -y)\right]$ \\
\hspace*{2em} $(Q_x^{\ell+1})_{i,j}=\frac{2\lambda(a+1)w_{i,j}}{2\lambda(a+1)w_{i,j}+a\sigma}
\left[(Q_x^\ell)_{i,j}+\sigma(\nabla_1\overline{X}^\ell)_{i,j}\right]$ \\
\hspace*{2em} $(Q_y^{\ell+1})_{i,j}=\frac{2\lambda(a+1)w_{i,j}}{2\lambda(a+1)w_{i,j}+a\sigma}
\left[(Q_y^\ell)_{i,j}+\sigma(\nabla_2\overline{X}^\ell)_{i,j}\right]$ \\
\hspace*{2em} $X_{\mathrm{prev}}^{\ell+1}=\frac{1}{1+2\tau c}
\left\{X_{\mathrm{prev}}^\ell-\tau\left[\Phi^T P^{\ell+1}+\nabla^*\left(Q_x^{\ell+1},Q_y^{\ell+1}\right)\right]+\tau V^{k}\right\}$ \\
\hspace*{2em} $\overline{X}^{\ell+1}=X_{\mathrm{prev}}^{\ell+1}+\theta
\left(X_{\mathrm{prev}}^{\ell+1}-X_{\mathrm{prev}}^\ell\right)$ \\
\hspace*{2em} \textbf{if} $\|X_{\mathrm{prev}}^{\ell+1}-X_{\mathrm{prev}}^\ell\|_{2}/\max\{\|X_{\mathrm{prev}}^\ell\|_{2},1\}|< \varepsilon_{\mathrm{inn}}$
\quad \textbf{then break} \\
\hspace*{1em}\textbf{Output:} $X_{\mathrm{new}}=\overline{X}^{\ell+1}$
\end{algorithm}

\subsection{Convergence of DCA}\label{subsec-convergIRLS}

In this subsection, we do some convergence analysis on Algorithm \ref{alg2}
by using \cite[Proposition A.1]{DL98} on the DCA,
which is based on the following moduli of strong convexity (see \cite[(5)]{DL98}).

\begin{definition}
Let $f$ be a convex function defined on $\mathbb{R}^N$. The \emph{modulus of strong convexity}
$m(f)$ of $f$ is defined by setting
$$
m(f):=\sup\left\{\rho\in[0,\infty):\ f(\cdot)-\frac{\rho}{2}\|\cdot\|_2^2 \text{ is convex on } \mathbb{R}^N\right\}.
$$
\end{definition}

\begin{lemma}\label{lem-DC1}
Let $f=g-h$ be a DC decomposition with $m(g)>0$ and $m(h)>0$,
and let $\{X^n\}_{n\in\mathbb{N}}$ be a sequence generated by DCA. Then, for any $n\in\mathbb{N}$,
$$
\left\|X^{n+1}-X^n\right\|_2^2 \le \frac{2}{m(g)+m(h)} \left[ f(X^n) - f(X^{n+1})\right].
$$
\end{lemma}

Some convergence results on Algorithm \ref{alg2} read as follows.

 \begin{theorem}\label{thm-con-DCA}
Let $w:=[w_{i,j}]\in\mathbb{R}^{N\times N}$ be the input weight with each $w_{i,j}>0$,
$y$ the input measurement in Algorithm \ref{alg2},
and $\{X^n\}_{n\in\mathbb{N}}$
a sequence generated from Algorithm \ref{alg2}. Then
the following assertions hold.
\begin{enumerate}
\item[\rm(i)] $\{F_w(X^n)\}_{n\in\mathbb{N}}$ is decreasing and convergent.

\item[\rm(ii)] $\{X^n\}_{n\in\mathbb{N}}$ and $\{\nabla X^n\}_{n\in\mathbb{N}}$
have the asymptotic regularity, i.e.,
$$
\lim_{n\to\infty}\left\|X^{n+1}-X^n\right\|_2=0
\ \ \text{and}\ \
\lim_{n\to\infty}\left\|\nabla X^{n+1}-\nabla X^n \right\|_2=0.
$$
\item[\rm(iii)] If $\lambda>\frac{\|w^{-1}\|_\infty \| y \|_2^2}{2(a+1)}$
with $w^{-1}:=[w_{i,j}^{-1}]$,
then $\{\nabla X^n\}_{n\in\mathbb{N}}$ is $\ell_\infty$-bounded
and, for any accumulation point $X^*$ of $\{X^n\}_{n\in\mathbb{N}}$,
$\partial F_w(X^*)=0.$
\end{enumerate}
\end{theorem}

\begin{proof}
To prove (i), we let $F_w=G_w-H_w$ be a DC decomposition as in \eqref{eq-DCdecom}.
Then, by the definitions of $G_w$ and $H_w$,
we easily conclude $m(G_w)\ge 2c$ and $m(H_w)\ge 2c$ with $c>0$.
From this and Lemma \ref{lem-DC1}, it follows that, for any $n\in\mathbb{N}$,
$F_w(X^n) - F_w(X^{n+1})\ge 0$,
which implies the decreasing property of $\{F_w(X^n)\}_{n\in\mathbb{N}}$.
By this and  $F_w(X^n)\ge0$ for any $n\in\mathbb{N}$,
we obtain the convergence of $\{F_w(X^n)\}_{n\in\mathbb{N}}$, which proves (i).

Furthermore, by (i) and Lemma \ref{lem-DC1},
we conclude that
$$
\left\|X^{n+1}-X^n\right\|_2^2 \le \frac{F_w(X^n) - F_w(X^{n+1})}{2c} \to 0 
\quad\text{as}\ n\to\infty,
$$
and hence
$$
\left\|\nabla X^{n+1}-\nabla X^n\right\|_2^2 \le 2\left\|X^{n+1}-X^n\right\|_2^2 \to 0    \quad\text{as}\ n\to\infty,
$$
which implies (ii).

Finally, we prove (iii). By the decreasing property of $\{F_w(X^n)\}_{n\in\mathbb{N}}$
and $X^0=\mathbf{0}$, we have
$$
\lambda\sum_{i,j=1}^{N}w_{i,j} \frac{(a+1)|(\nabla X^n)_{i,j}|^2}{a+|(\nabla X^n)_{i,j}|^p}
+ \frac{1}{2}\left\|\Phi X^n-y\right\|_2^2 =F_w(X^n)\le F_w(X^0)=\frac{1}{2}\left\| y\right\|_2^2.
$$
This  implies that, for each $i,j\in\{1,\ldots,N\}$,
$$
 \frac{|(\nabla X^n)_{i,j}|^2}{a+|(\nabla X^n)_{i,j}|^p} \le\frac{\|w^{-1}\|_\infty \| y \|_2^2}{2\lambda (a+1)}.
$$
Thus, if $\lambda>\frac{\|w^{-1}\|_\infty \| y\|_2^2}{2(a+1)}$,
we have, for each $i,j\in\{1,\ldots,N\}$,
either $|(\nabla X^n)_{i,j}|\le 1$ or, by $|(\nabla X^n)_{i,j}|^2>|(\nabla X^n)_{i,j}|^p$ when $|(\nabla X^n)_{i,j}|>1$,
$$
|(\nabla X^n)_{i,j}|^2\le\frac{a\|w^{-1}\|_\infty \| y \|_2^2}{2\lambda (a+1)-\|w^{-1}\|_\infty \| y \|_2^2},
$$
which means that $\{\nabla X^n\}_{n\in\mathbb{N}}$ is $\ell_\infty$ bounded.

To prove the remainder, let $\{X^{n_k}\}_{k\in\mathbb{N}}$ be a subsequence of
$\{X^n\}_{n\in\mathbb{N}}$ with the limit point $X^*$.
Then, in the $(n_k -1)$-th step of Algorithm \ref{alg2}, by \eqref{eq-PG} and \eqref{eq-PH}, we have
\begin{align*}
0&= \partial G_w(X^{n_k})-V^{n_k-1}= \partial G_w(X^{n_k})- \partial H_w(X^{n_k-1})\\
&=\Phi^T (\Phi X^{n_k}-Y) + 2c (X^{n_k}-X^{n_k -1})
+\frac{2\lambda(a+1)}{a}\nabla^*\left(w\odot \nabla_1 X^{n_k},w\odot \nabla_2 X^{n_k}\right)\\
&\quad-\lambda(a+1)\left[\frac{2}{a}\nabla^*\left(w\odot \nabla_1 X^{n_k-1},w\odot \nabla_2 X^{n_k-1}\right)\right.\nonumber\\
&\quad+2\nabla^*\left(\sum_{i,j=1}^{N} \frac{w_{i,j} (\nabla_1 X^{n_k-1})_{i,j}}{a+|(\nabla X^{n_k-1})_{i,j}|^p}E_{i,j},
\sum_{i,j=1}^{N} \frac{w_{i,j} (\nabla_2 X^{n_k-1})_{i,j}}{a+|(\nabla X^{n_k-1})_{i,j}|^p}E_{i,j}\right)\nonumber\\
&\quad\left.- p \nabla^*\left(\sum_{i,j=1}^{N}
\frac{w_{i,j} |(\nabla_1 X^{n_k-1})_{i,j}|^p(\nabla_1 X^{n_k-1})_{i,j}}{[a+|(\nabla X^{n_k-1})_{i,j}|^p]^2}E_{i,j},
\sum_{i,j=1}^{N}\frac{w_{i,j} |(\nabla_2 X^{n_k-1})_{i,j}|^p(\nabla_2 X^{n_k-1})_{i,j}}{[a+|(\nabla X^{n_k-1})_{i,j}|^p]^2}E_{i,j}\right)\right] .
\end{align*}
Since $\|X^{n+1}-X^n\|_2\to0 $ as $n\to\infty$
and $\|X^{n_k}-X^*\|_\infty\to0 $  as $k\to \infty$,
we infer that
$0= \partial G_w(X^*) - \partial H_w(X^*)=\partial F_w(X^*),$
which completes the proof of Theorem \ref{thm-con-DCA}.
\end{proof}

\section{Numerical Experiments}\label{sec-exp}

In this section, we evaluate the performance of the proposed PSV$_{a,p}$ minimization
across three distinct scenarios: natural image reconstruction in Subsection \ref{subsec-NI},
magnetic resonance imaging (MRI) reconstruction in Subsection \ref{subsec-MRI},
and X-ray computed tomography (CT) reconstruction in Subsection \ref{subsec-CT}.
We employ the \emph{peak signal-to-noise ratio (PSNR)}, the
\emph{structural similarity (SSIM)} (see \cite{WBSS04}),
and the \emph{gradient magnitude similarity deviation (GMSD)}
(see \cite{XZMB14} ) as quality metrics
to quantitatively evaluate the fidelity of ​each reconstructed image.
Here, the PSNR measures pixel-wise accuracy, whereas the SSIM and the GMSD
evaluate the reconstruction quality
from a perceptual perspective that a higher SSIM means
greater global structure similarity
and a lower GMSD indicates
less local detail distortion.

We select the zero vector as the initial point of the IRLSPSV algorithm and,
if not otherwise specified, we adopt the following parameter settings: $\tau=1$,
$s=\lfloor0.9N^2+0.5\rfloor$, and $\delta$ adaptive in Algorithm \ref{alg1},
$c=0$ in Algorithm  \ref{alg2},
and $\theta=1$ and $\sigma=0.5=\tau$ in  Algorithm  \ref{alg3}.
Here, and thereafter, the symbol $\lfloor \alpha\rfloor$
denotes the largest integer not greater than $\alpha$.
The stopping criteria (maximum iteration steps $k_{\mathrm{inn}}$, $k_{\mathrm{mid}}$, $k_{\mathrm{out}}$,
and relative tolerances $\varepsilon_{\mathrm{inn}}$ and $\varepsilon_{\mathrm{mid}}$)
are tailored according to the specific application scenario.
In addition, two crucial parameters of the proposed model, $a$ and $p$, are
manually selected​ based on the specific testing scenario and image content
with the recommended ranges $a\in[0.1,50]$ and $p\in[0.1,1]$.
The regularity parameter $\lambda$ is also tuned manually,
whose value is given in scientific notation (e.g., 1.0e-3 = $1.0\times10^{-3}$)
with the recommended range from 1.0e-6  to 1.0e-2.
It is interesting to develop some adaptive algorithms
that automatically tune these parameters.

With the \emph{zero-padding (ZP)} reconstruction  as a baseline,
comparative  experiments are conducted between the proposed PSV$_{a,p}$ model against
the \emph{total variation (TV)} (see \cite{ROF}), the $L_1-\alpha L_2$ (see \cite{HCGN23,LLCZWY17}),
and the \emph{transformed total variation (TTV)} (see \cite{HCGN22}),
where the $L_1-\alpha L_2$ and the TTV models are solved by the DC algorithm.
Furthermore, for fair comparison, the regularity parameters for all competing methods
are empirically tuned to  attain comparable reconstruction performance across all experiments in this section.

All  experiments are conducted in MATLAB on a Thinkpad desktop with 32 GB of RAM
and 13-th Generation Intel Core i9-13900H Processor.

\subsection{Natural Image Reconstruction}\label{subsec-NI}

In this subsection, we apply the proposed PSV$_{a,p}$ model to the natural image reconstruction.
Both parameter-sensitivity experiments and comparative experiments
are conducted to evaluate the proposed model in terms of
\begin{enumerate}
\item[\rm(i)]  broad applicability through two-parameter optimization,
\item[\rm(ii)] consistently strong performance​ across multiple   metrics under various testing scenarios, and
\item[\rm(iii)] a reasonable trade-off  between performance and efficiency.
\end{enumerate}

The measurements are obtained by the Gaussian random sampling  in the Fourier domain,
using three natural images from set12: Cameraman, Starfish, and Bird. To be precise,
the operator $\Phi$ in \eqref{eq-unproblem} represents a composed operator $\Phi=SF$,
where $F$ denotes the Fourier matrix and $S$ the Gaussian measurement matrix.
Notably, in all experiments, a fixed Gaussian measurement
matrix is used to ensure reproducibility and fair comparison.
The stopping criteria of the IRLSPSV algorithm are $k_{\mathrm{inn}}=20=k_{\mathrm{mid}}$ and
$\varepsilon_{\mathrm{inn}}=10^{-5}=\varepsilon_{\mathrm{mid}}$ for the inner two loops
and $k_{\mathrm{out}}=20$ for the outer loop.

To assess the image dependency​ and the parameter sensitivity of the proposed model,
we conduct two sets of experiments on Cameraman, Starfish, and Bird
under sampling rates of 40\% and 60\%
in a noise-free  setting. To be precise, we examine the influence of $a$ and $p$
by fixing one parameter while varying the other.

\begin{table}[H]
\centering
\resizebox{0.96\textwidth}{!}{
\renewcommand{\arraystretch}{1.3}
\begin{tabular}{c|c|cccc|cccc}
\hline
\multirow{2}{*}{Image}&  &\multicolumn{4}{c|}{40\% sampling rate} & \multicolumn{4}{c}{60\% sampling rate} \\
\cline{2-10}
&$(a,p)$& $(0.5,1)$ & $(1,1)$& $(5,1)$ & $(50,1)$ & $(0.5,1)$ & $(1,1)$& $(5,1)$ & $(50,1)$ \\
\hline
\multirow{4}{*}{Cameraman} & $\lambda$& 5.0e-4&8.0e-4 & 1.2e-3& 1.3e-3
& 2.2e-4 & 5.0e-4& 1.0e-3 & 1.0e-3\\
& PSNR & $29.1604$ & 29.2213& \textbf{29.4173}& 29.2669& 34.0192& \textbf{34.9579}& 34.8449 & 33.9629\\
& SSIM & $0.8597$ & 0.8675& 0.8812& \textbf{0.8842} & 0.9368 &0.9475& \textbf{0.9536} & 0.9496\\
& GMSD & $0.1597$ & 0.1529& 0.1430 & \textbf{0.1407} & 0.0964 & 0.0902& 0.0885& \textbf{0.0853}\\
\hline
\multirow{4}{*}{Starfish} & $\lambda$ & 8.0e-4 & 1.3e-3& 2.0e-3
& 2.5e-3 & 8.0e-4 & 9.0e-4& 2.0e-3 & 2.5e-3\\
& PSNR & \textbf{18.4760} & 18.4255& 18.2247&18.1666& \textbf{26.1779} & 26.0045& 25.4824 & 25.1236\\
& SSIM & 0.7030 & 0.7067&\textbf{0.7131}&0.7129&0.8828&\textbf{0.8870}&0.8822&0.8785\\
& GMSD &0.2035&0.2041&\textbf{0.2029}&0.2039&0.1525&0.1482&\textbf{0.1475}&0.1476\\
\hline
\multirow{4}{*}{Bird} & $\lambda$ & 5.0e-4 & 7.0e-4& 1.5e-3
& 1.8e-3 & 1.0e-4 & 5.0e-4& 1.2e-3& 1.5e-3\\
& PSNR & $21.0641$ & 21.2201& 22.2660 & \textbf{22.4461} & 31.0837& 31.1717& 32.2085 & \textbf{32.6829}\\
& SSIM & 0.8164&0.8263&0.8438&\textbf{0.8479} &0.9468&0.9529&0.9561&\textbf{0.9570}\\
& GMSD & 0.1809 &0.1759 &0.1697&\textbf{0.1680}& 0.1145&0.1014&0.0982&\textbf{0.0977}\\
\hline
\end{tabular}}
\caption{Sensitivity analysis on parameter $a$ with $p=1$,
evaluated on noise-free measurements under sampling rates of 40\% and 60\%.}\label{table-a}
\end{table}

\begin{table}[H]
\centering
\resizebox{0.96\textwidth}{!}{
\renewcommand{\arraystretch}{1.3}
\begin{tabular}{c|c|cccc|cccc}
\hline
\multirow{2}{*}{Image}&  &\multicolumn{4}{c|}{40\% sampling rate} & \multicolumn{4}{c}{60\% sampling rate} \\
\cline{2-10}
&$(a,p)$& $(1,1)$ & $(1,0.7)$& $(1,0.6)$ & $(1,0.5)$ & $(1,1)$ & $(1,0.7)$& $(1,0.6)$ & $(1,0.5)$\\
\hline
\multirow{4}{*}{Cameraman} & $\lambda$& 8.0e-4&1.2e-4 & 6.0e-5& 5.0e-5
& 5.0e-4 & 1.0e-4& 6.0e-5 & 3.0e-5\\
& PSNR & $29.2213$ & \textbf{30.6976}& 29.4443 & 27.8209 & 34.9579 & 38.9139& \textbf{39.5400} & 37.9420\\
& SSIM & $0.8675$ & 0.9174& \textbf{0.9228} & 0.9035 & 0.9475 &0.9731& \textbf{0.9761} & 0.9698\\
& GMSD & $0.1529$ & \textbf{0.1230}& \textbf{0.1230} & 0.1317 & 0.0902 & 0.0716& 0.0694 & \textbf{0.0644}\\
\hline
\multirow{4}{*}{Starfish} & $\lambda$ & 1.3e-3 & 3.0e-4& 2.0e-4
& 1.0e-4 & 9.0e-4 & 3.0e-4& 2.0e-4 & 1.0e-4\\
& PSNR & 18.4255 & 19.4001& 19.5952&\textbf{19.7168}& 26.0045 & 27.2889& \textbf{27.5102} & 27.1477\\
& SSIM & 0.7067 & 0.7592&0.7647&\textbf{0.7787}&0.8870&0.9143&\textbf{0.9191}&0.9186\\
& GMSD &0.2041&0.2059&\textbf{0.2034}&0.2088&0.1482&0.1347&0.1325&\textbf{0.1275}\\
\hline
\multirow{4}{*}{Bird} & $\lambda$ & 7.0e-4 & 1.8e-4& 1.0e-4
& 6.0e-5 & 5.0e-4 & 1.6e-4& 1.0e-4& 5.0e-5\\
& PSNR & $21.2201$ & \textbf{22.9646}& 22.5327 & 22.0887 & 31.1717 & 33.8538& 34.4672 & \textbf{34.6306}\\
& SSIM & 0.8263&0.8874&0.8929&\textbf{0.8942} &0.9529&0.9717&0.9746&\textbf{0.9773}\\
& GMSD & 0.1759 &0.1516 &\textbf{0.1496}&0.1504& 0.1014&0.0809 &0.0785&\textbf{0.0726}\\
\hline
\end{tabular}}
\caption{Sensitivity analysis on parameter $p$ with $a=1$,
evaluated on noise-free measurements under sampling rates of  40\% and 60\%.}\label{table-p}
\end{table}

As shown in Tables \ref{table-a} and \ref{table-p},
no single fixed parameter setting achieves the best reconstruction performance​
across all quality metrics and all test images simultaneously.
For instance, under the sampling rate of 40\%, while the parameter pair $(1,0.5)$
yields the highest PSNR for Starfish,
it fails to match​ the performance on Cameraman and Bird compared to their respective best settings.
This observation confirms the image dependency of the proposed model,
while also demonstrating its broad applicability, i.e.,​ it
can achieve high-quality reconstruction for diverse images​ simply by tuning the parameters​ accordingly.

\begin{figure}[t]
  \centering
  \begin{tabular}{ccccc}
  \includegraphics[width=2.5cm]{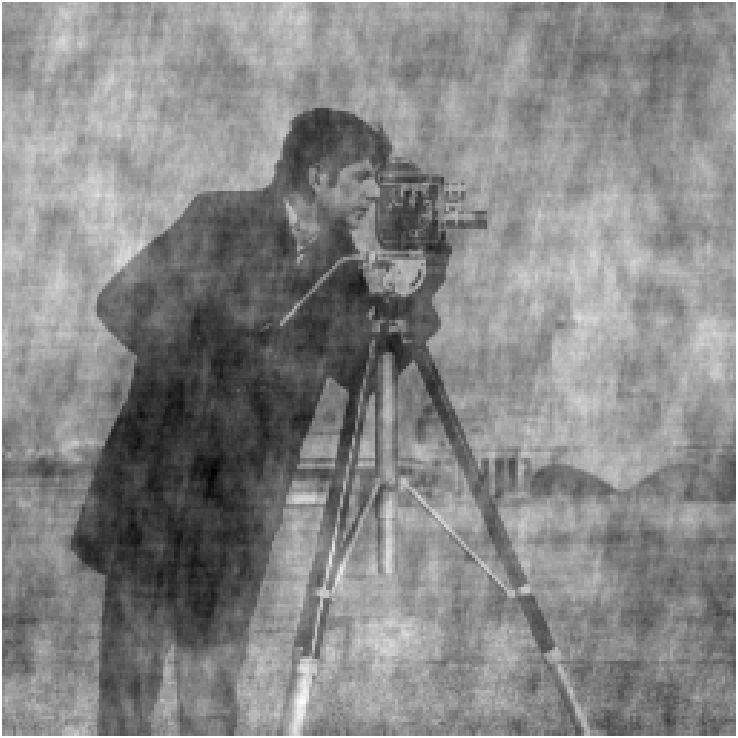}&
  \includegraphics[width=2.5cm]{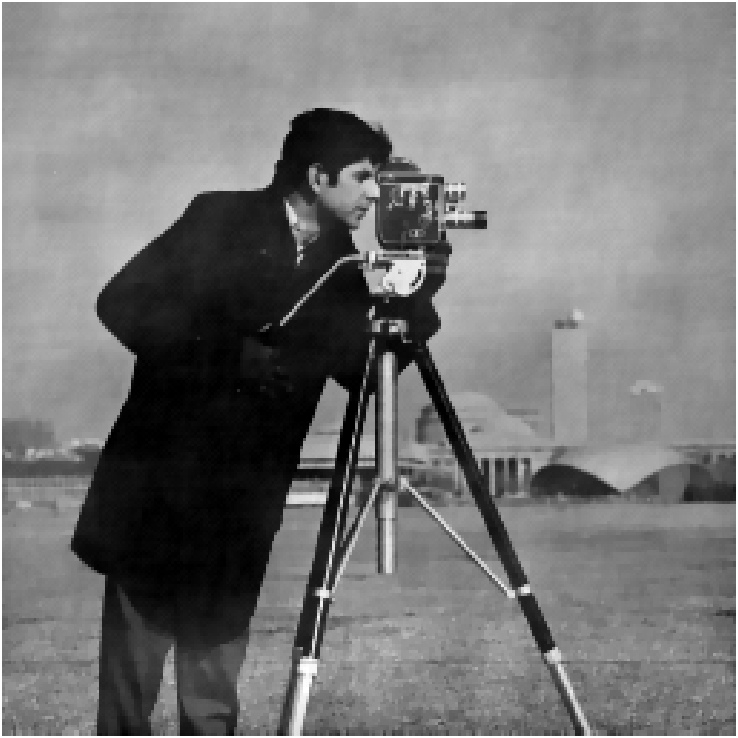}&
  \includegraphics[width=2.5cm]{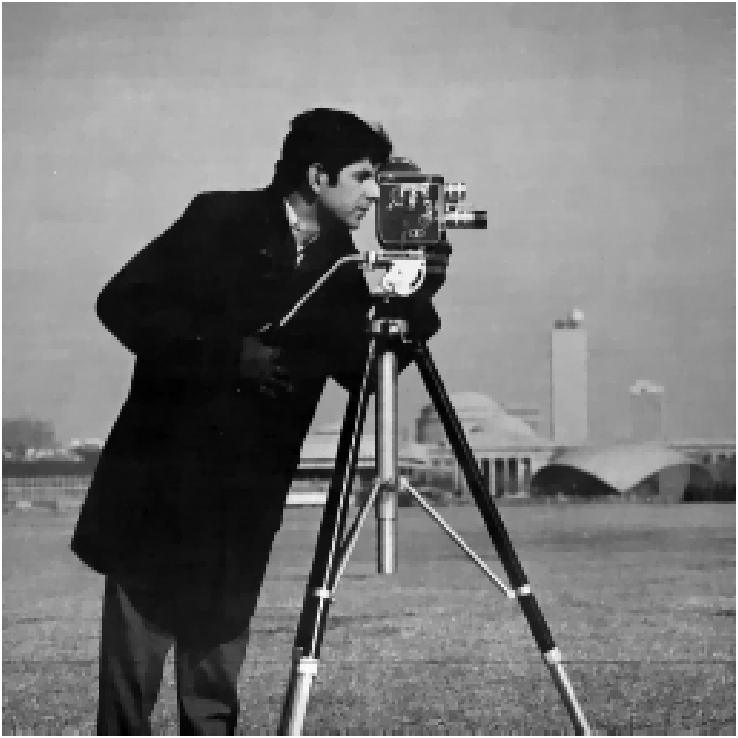}&
  \includegraphics[width=2.5cm]{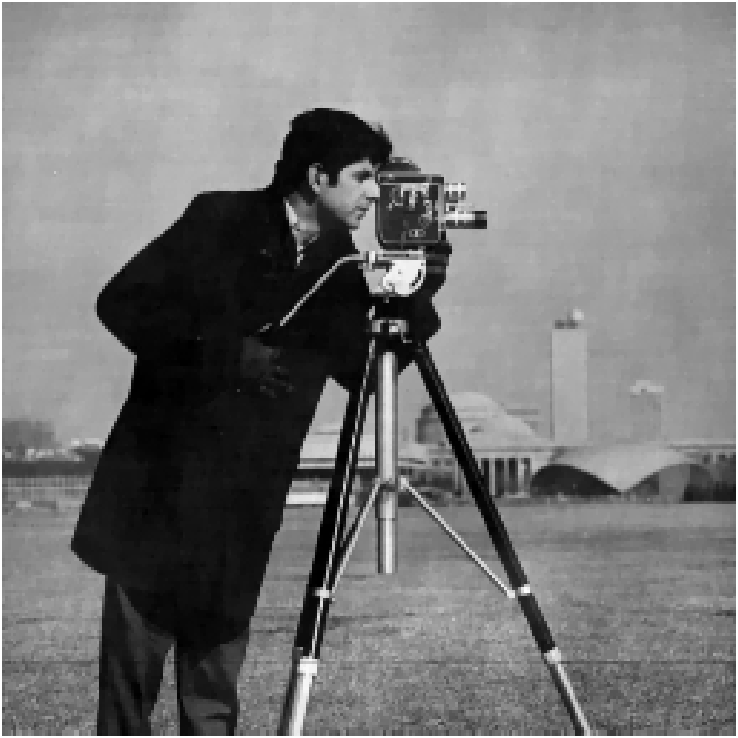}&
  \includegraphics[width=2.5cm]{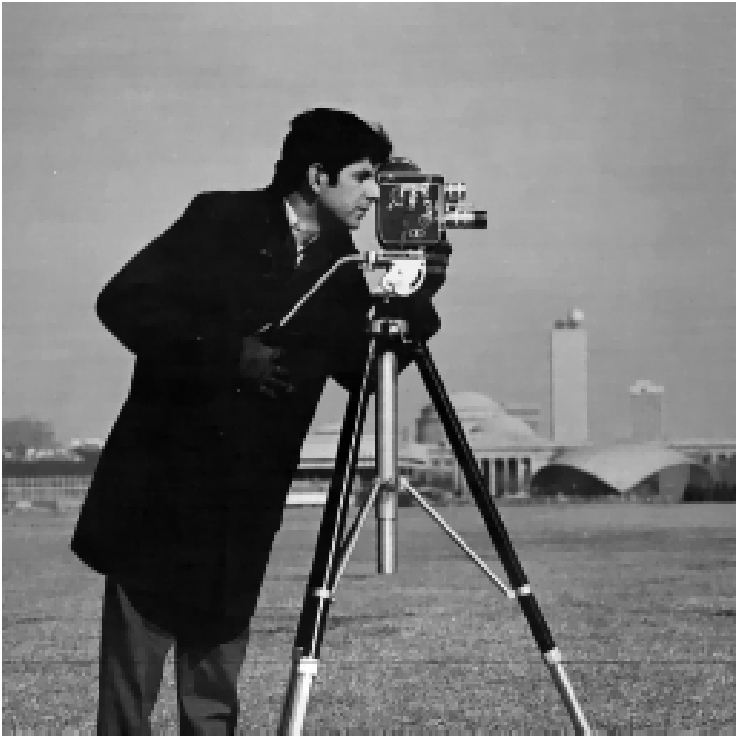}\\
  \includegraphics[width=2.5cm]{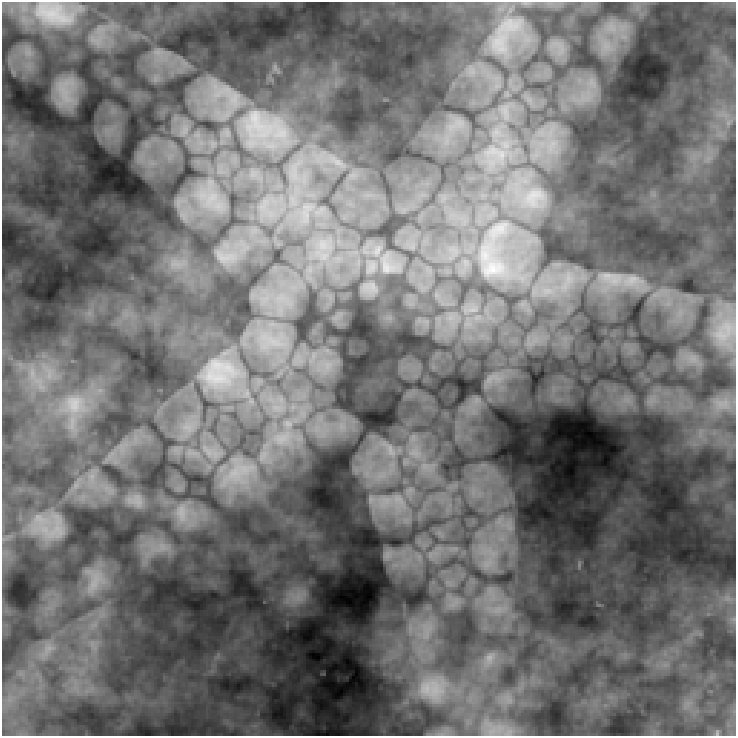}&
  \includegraphics[width=2.5cm]{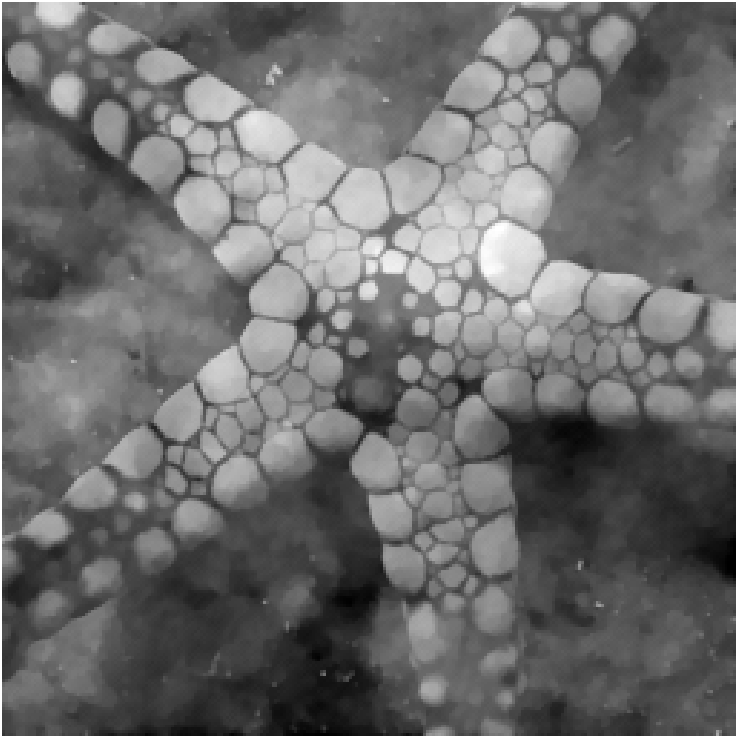}&
  \includegraphics[width=2.5cm]{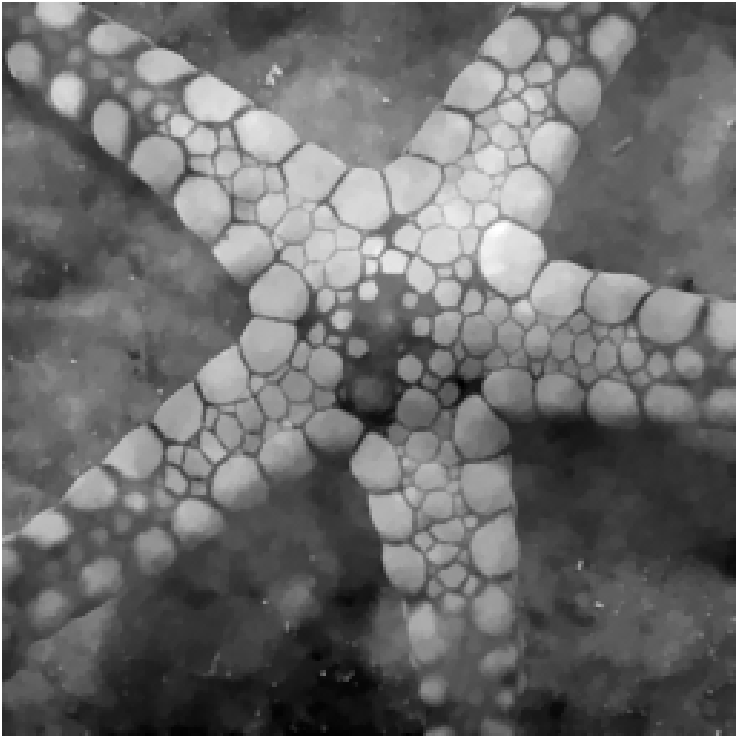}&
  \includegraphics[width=2.5cm]{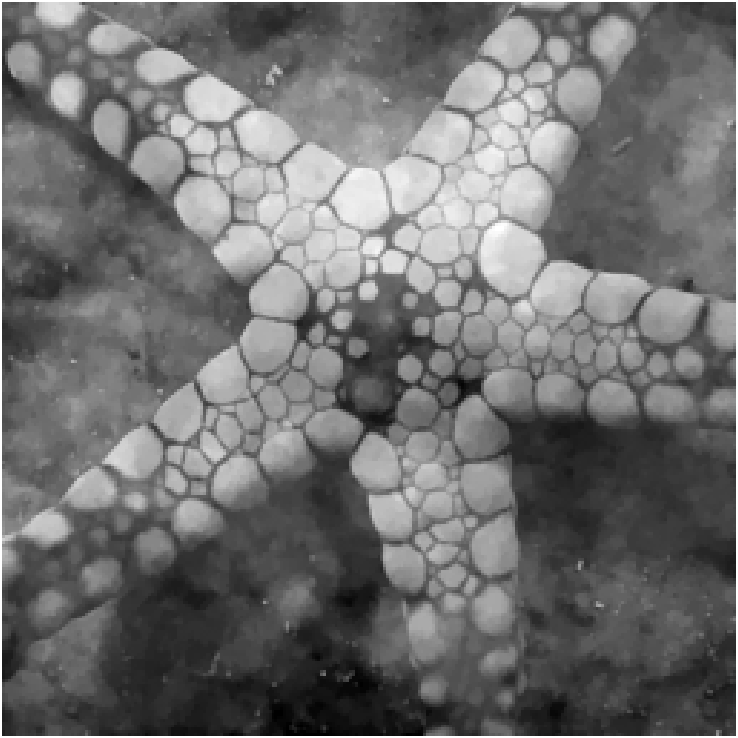}&
  \includegraphics[width=2.5cm]{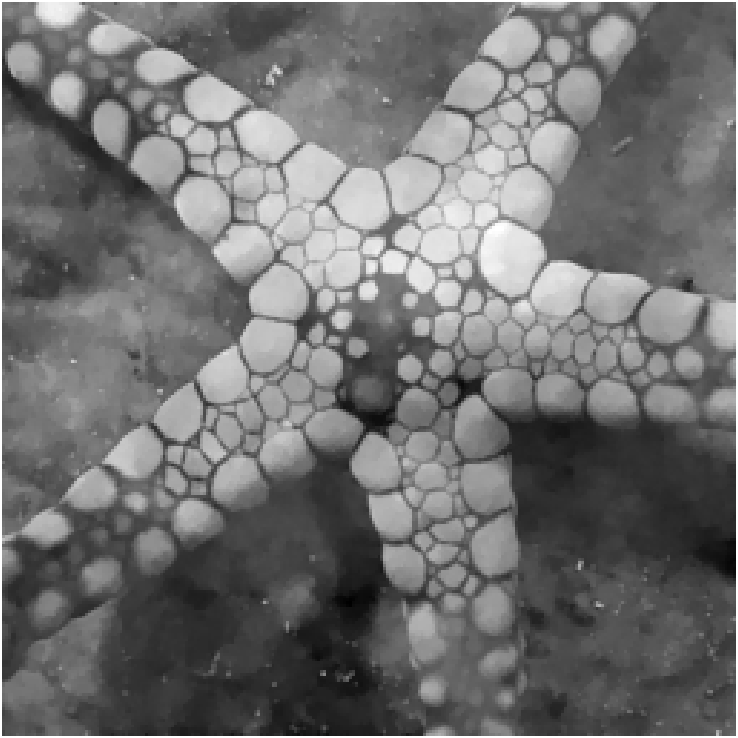}\\
  \includegraphics[width=2.5cm]{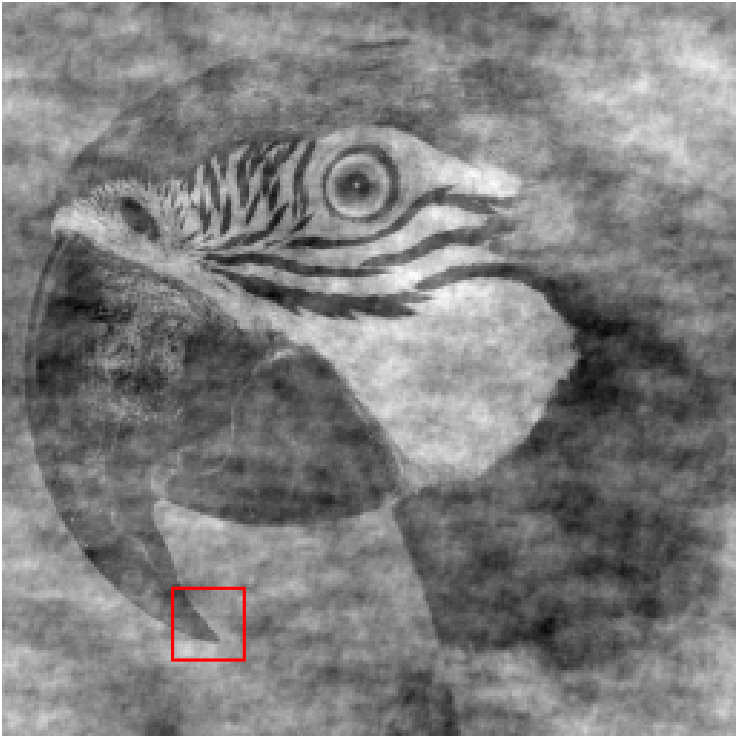}&
  \includegraphics[width=2.5cm]{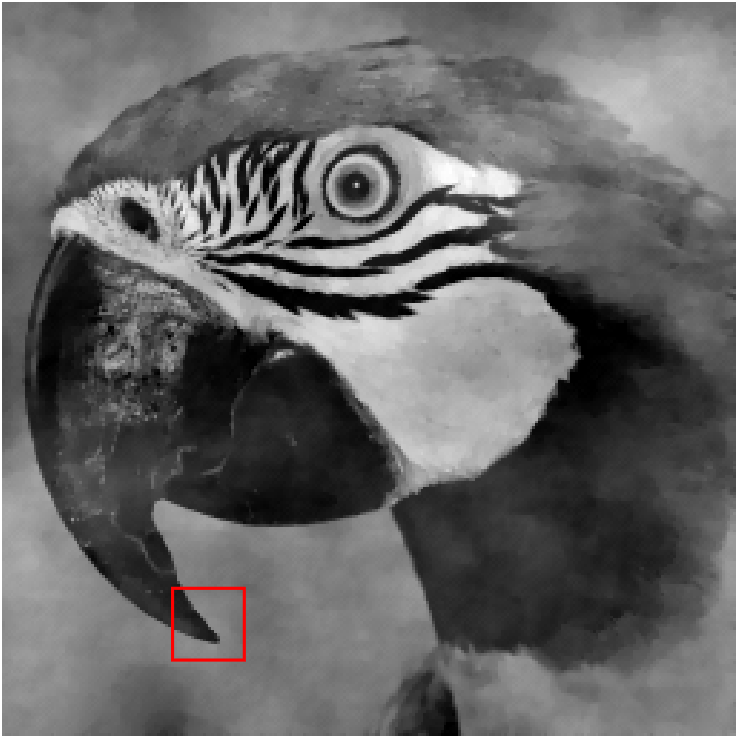}&
  \includegraphics[width=2.5cm]{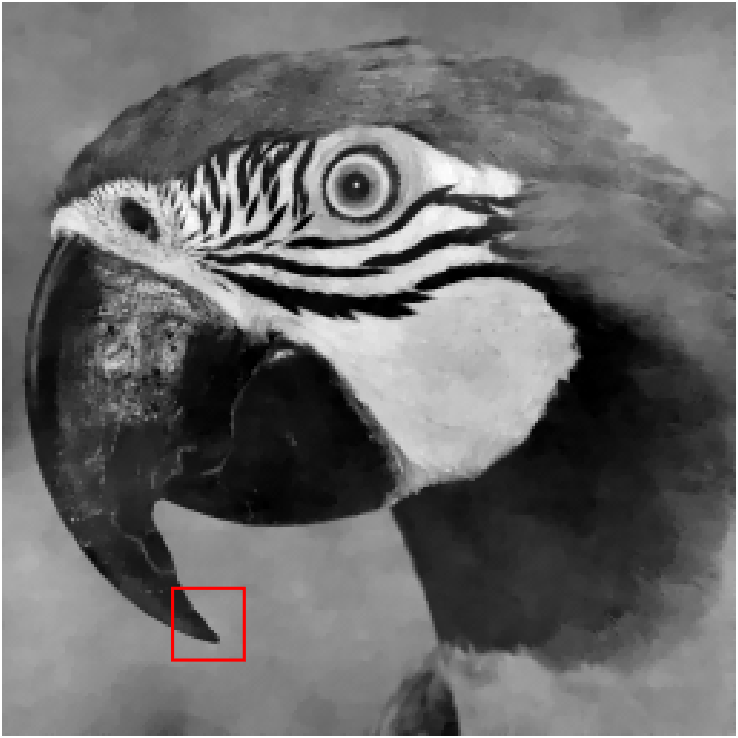}&
  \includegraphics[width=2.5cm]{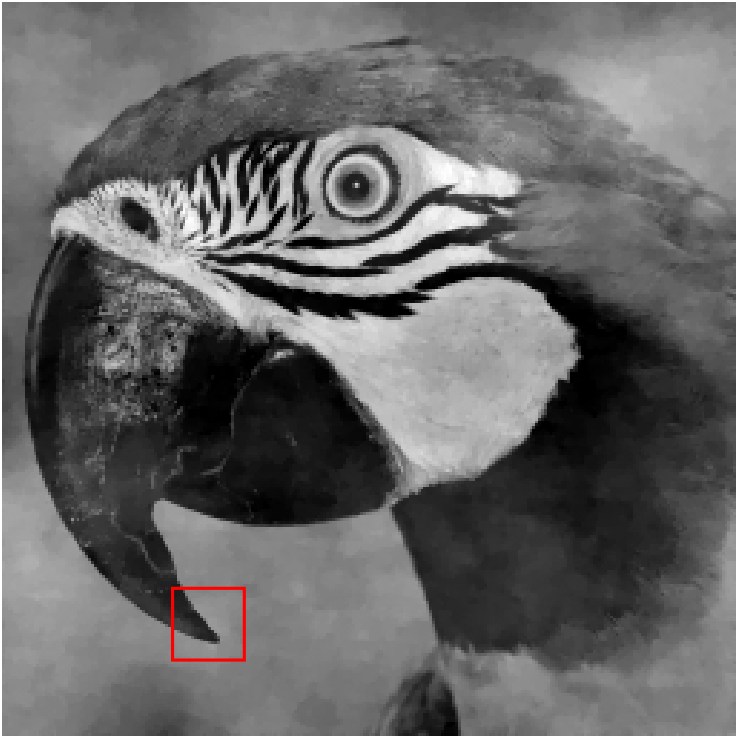}&
  \includegraphics[width=2.5cm]{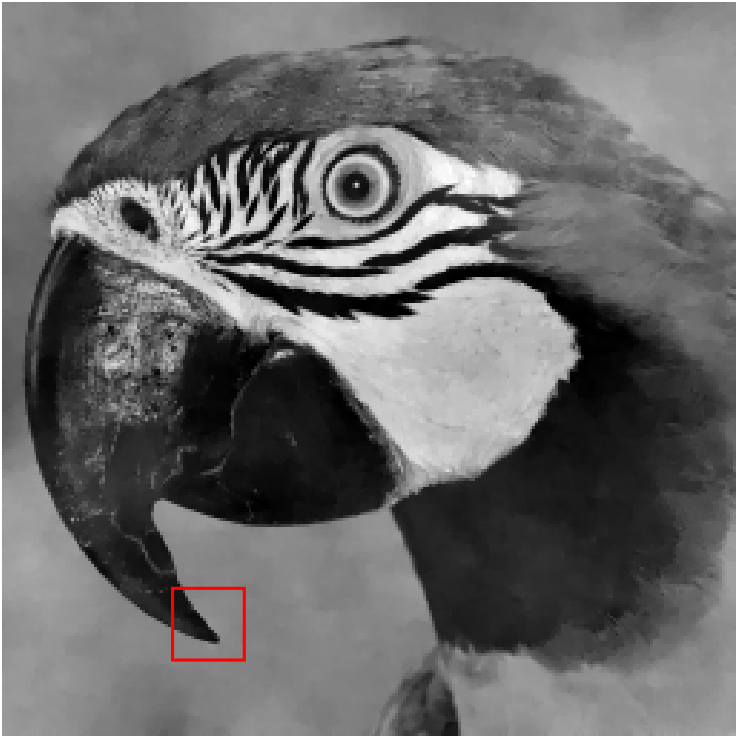}\\
  ZP& TV &$L_1-\alpha L_2$&TTV&PSV$_{a,p}$
  \end{tabular}
  \caption{Reconstructed images from noise-free data under 40\% sampling rates via  different methods. }\label{fig-noiseless40}
\end{figure}

\begin{figure}[H]
  \centering
  \begin{tabular}{ccccc}
  \includegraphics[width=2.5cm]{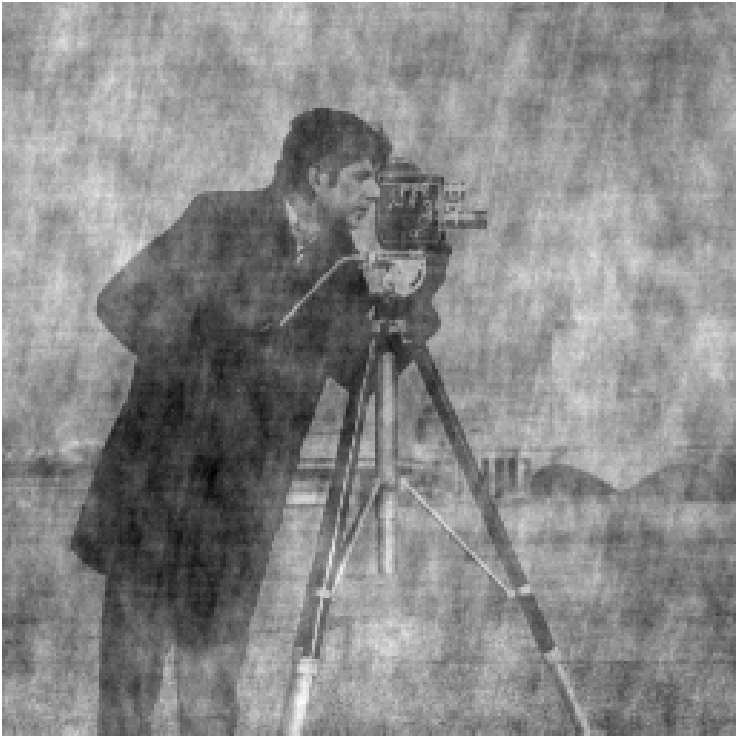}&
  \includegraphics[width=2.5cm]{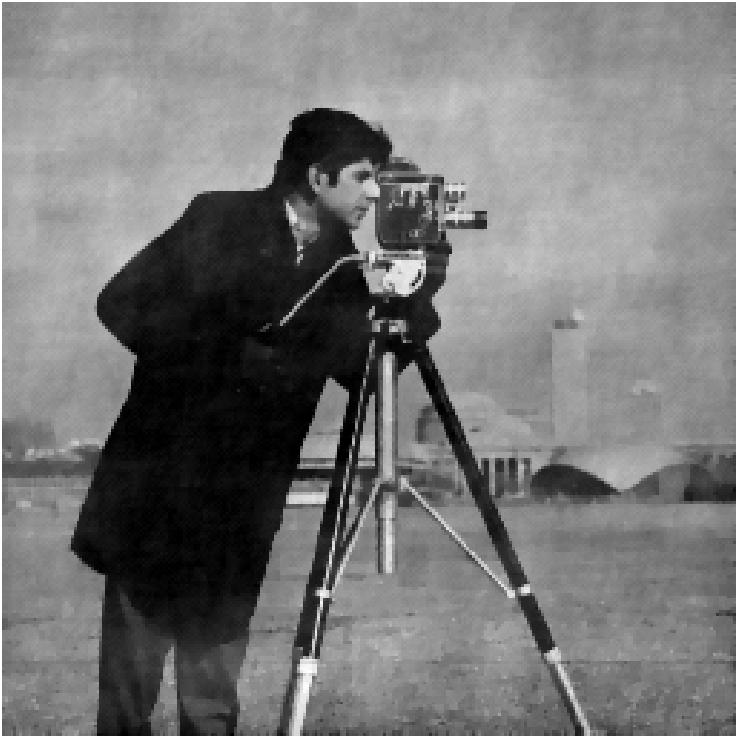}&
  \includegraphics[width=2.5cm]{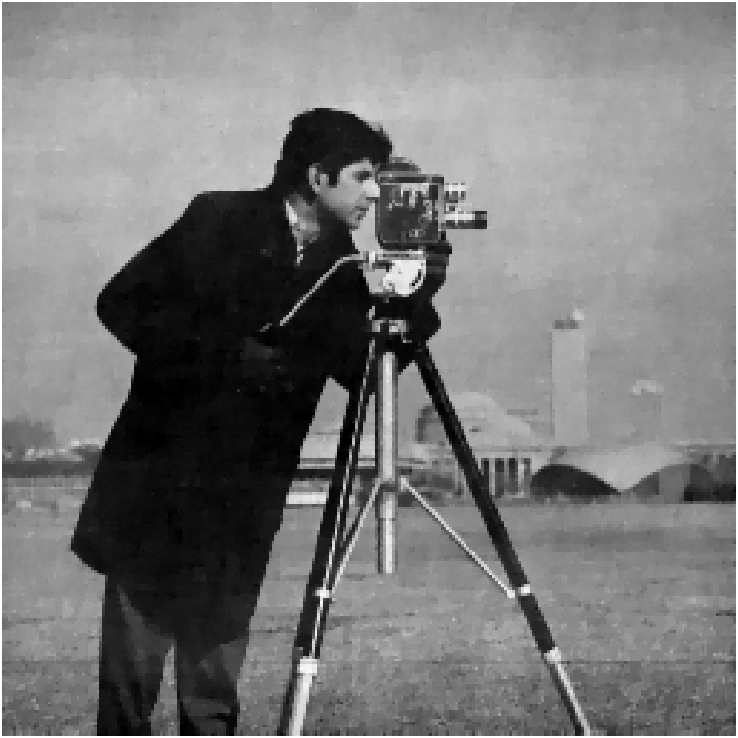}&
  \includegraphics[width=2.5cm]{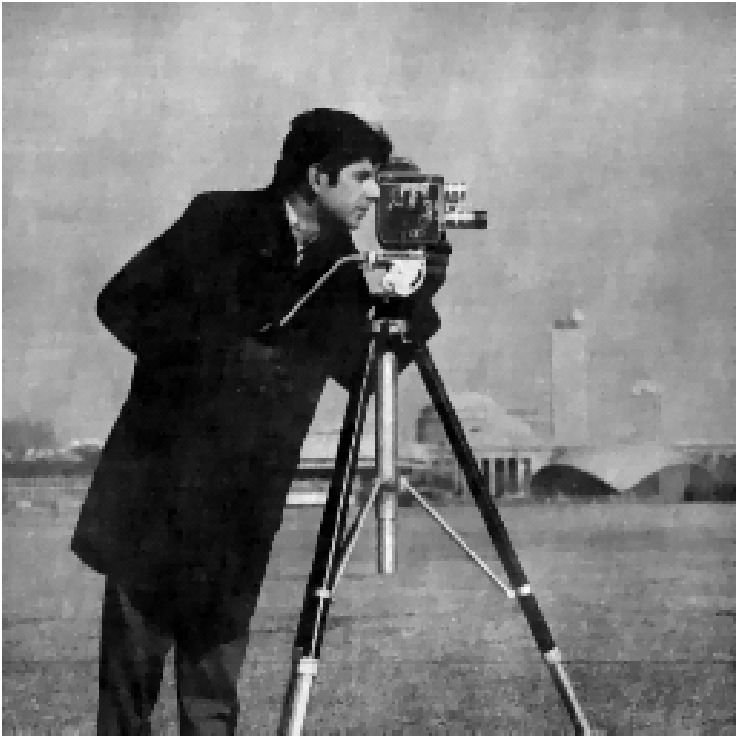}&
  \includegraphics[width=2.5cm]{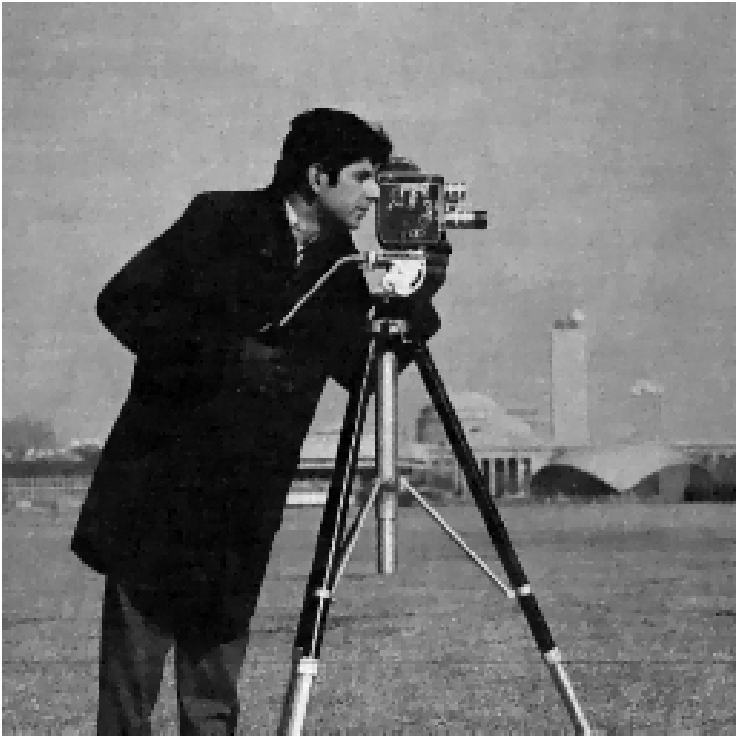}\\
  \includegraphics[width=2.5cm]{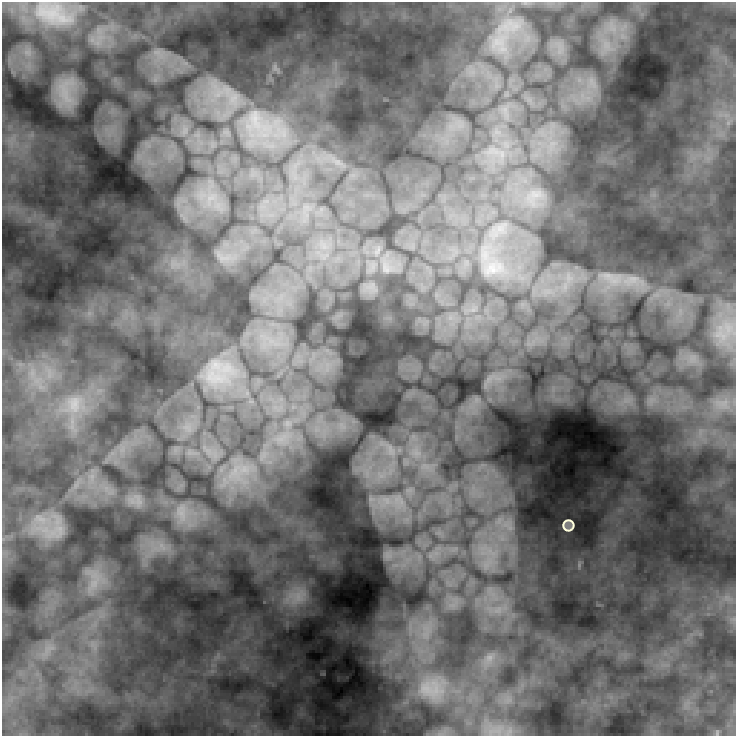}&
  \includegraphics[width=2.5cm]{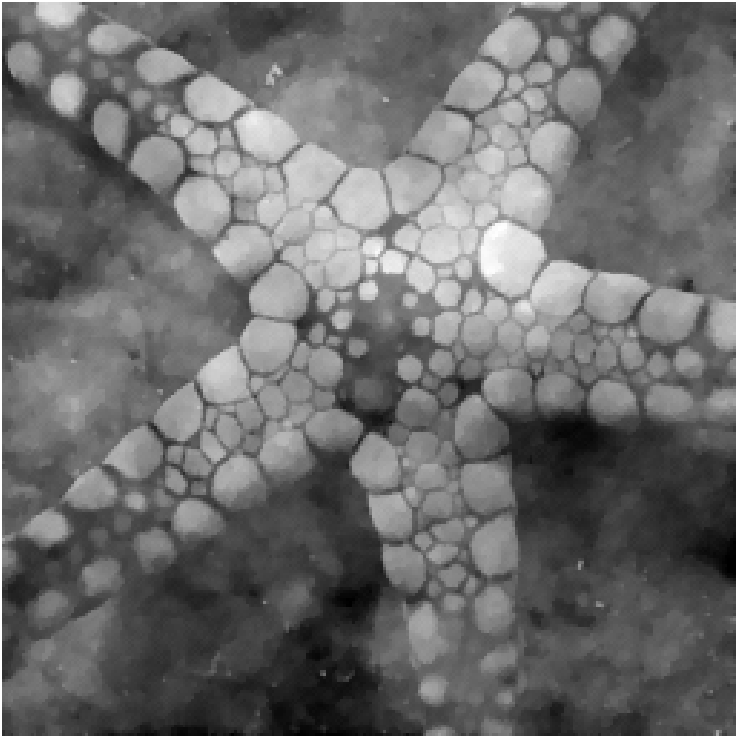}&
  \includegraphics[width=2.5cm]{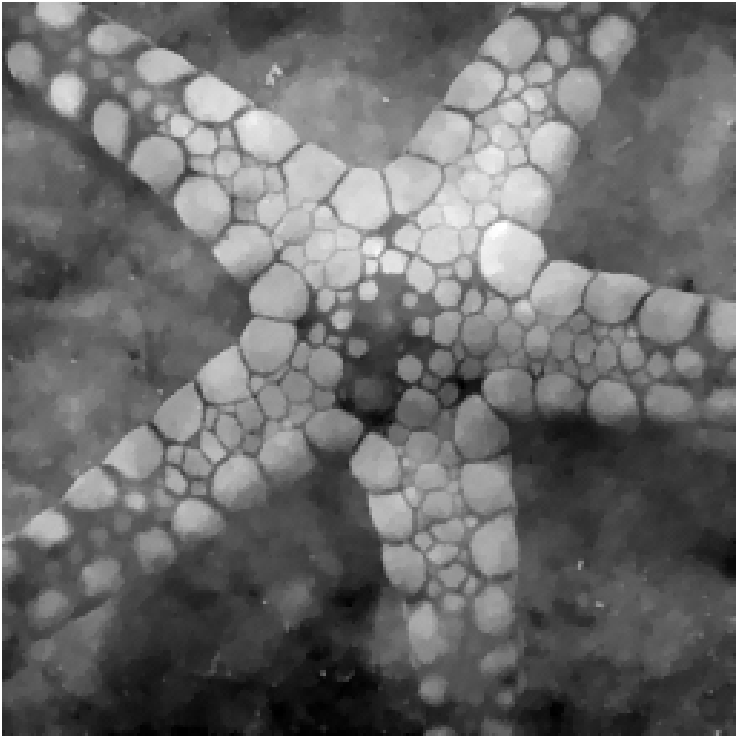}&
  \includegraphics[width=2.5cm]{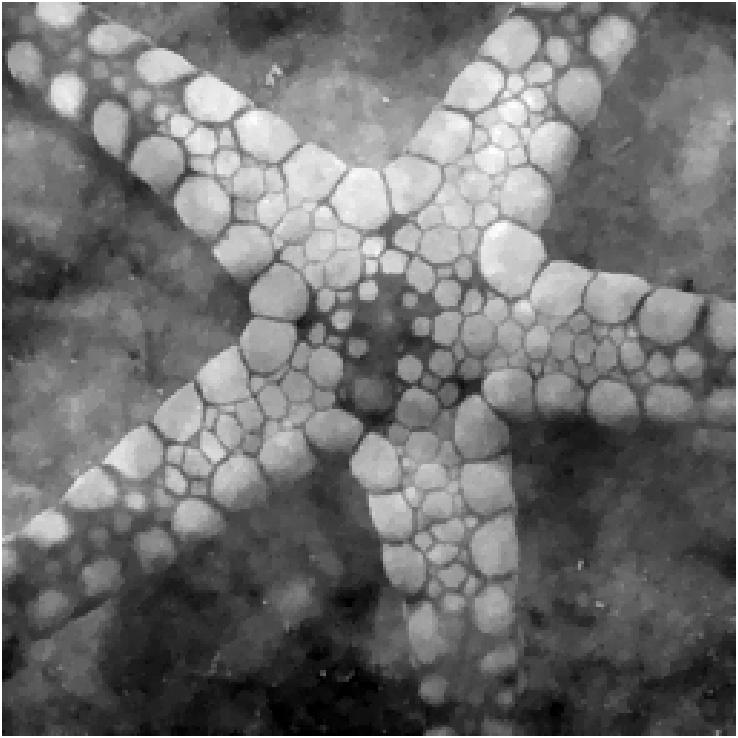}&
  \includegraphics[width=2.5cm]{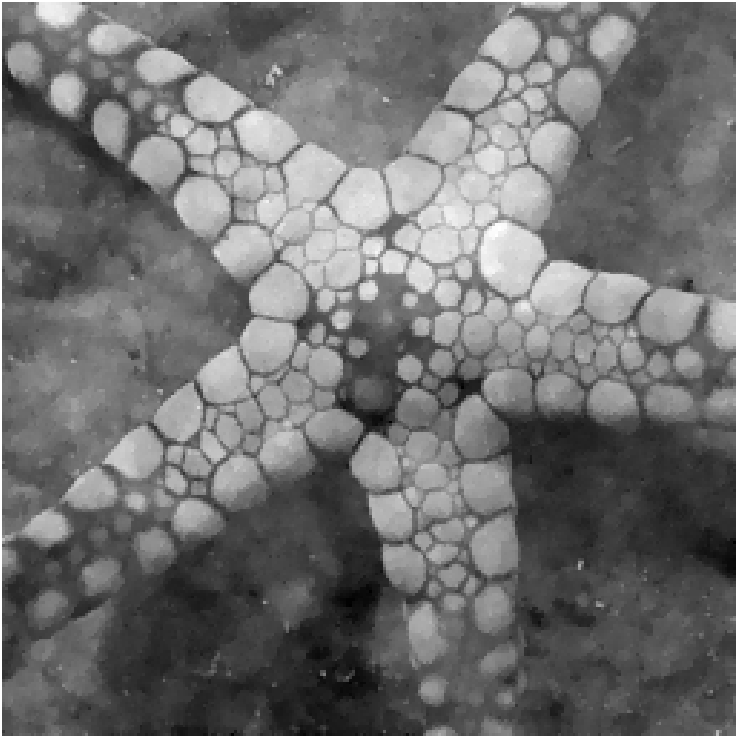}\\
  \includegraphics[width=2.5cm]{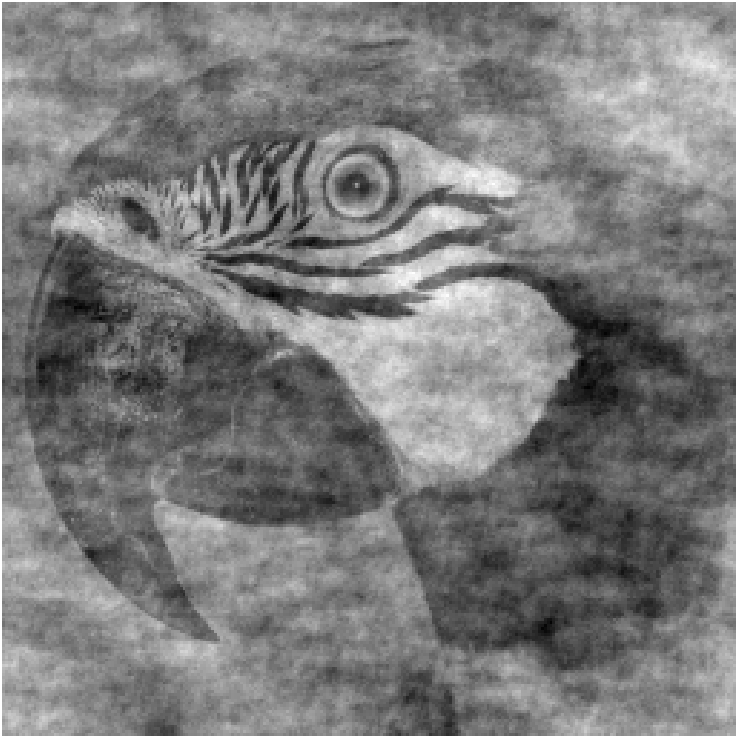}&
  \includegraphics[width=2.5cm]{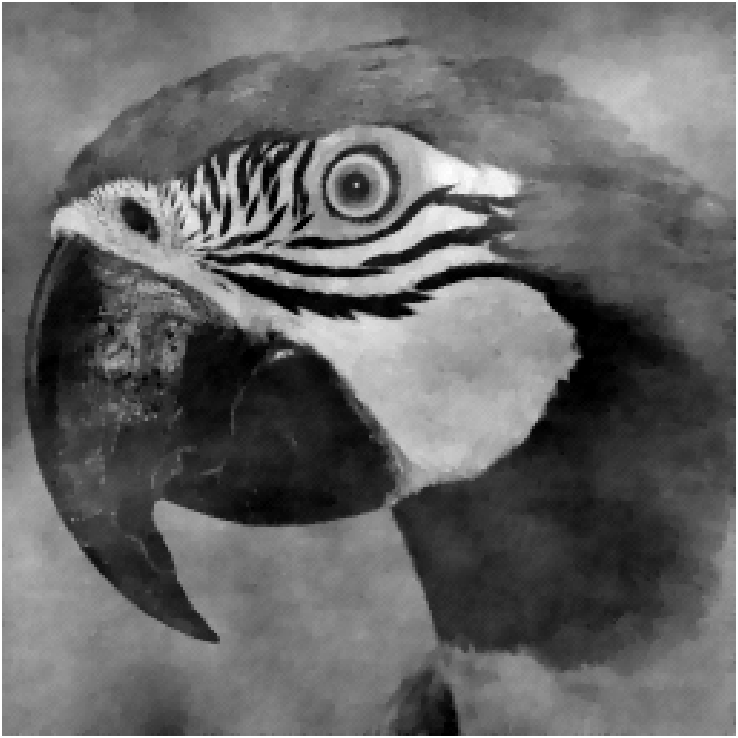}&
  \includegraphics[width=2.5cm]{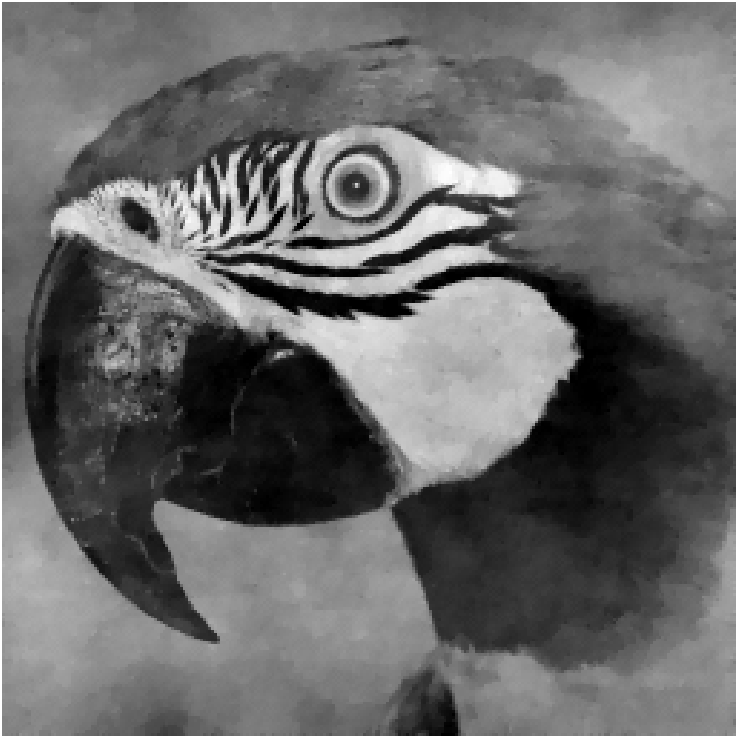}&
  \includegraphics[width=2.5cm]{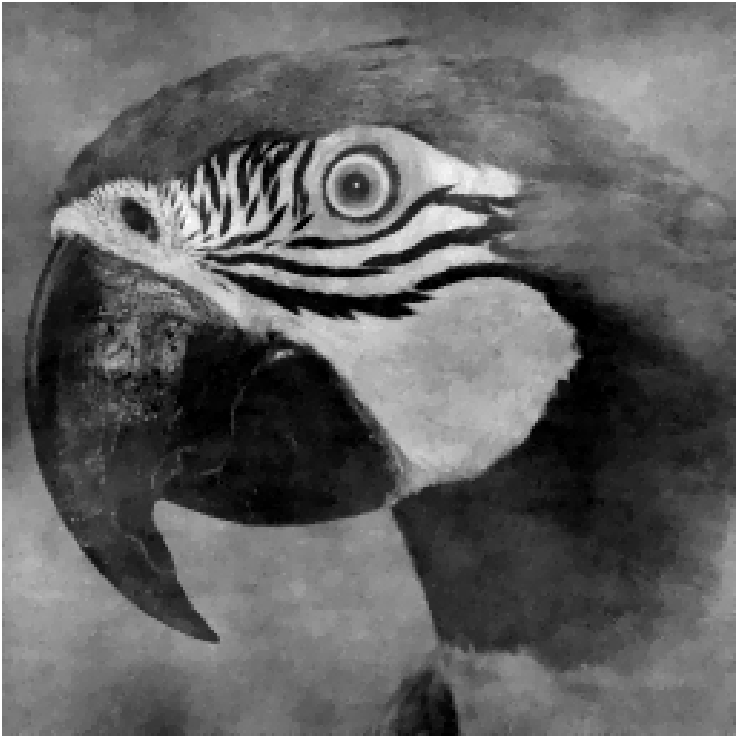}&
  \includegraphics[width=2.5cm]{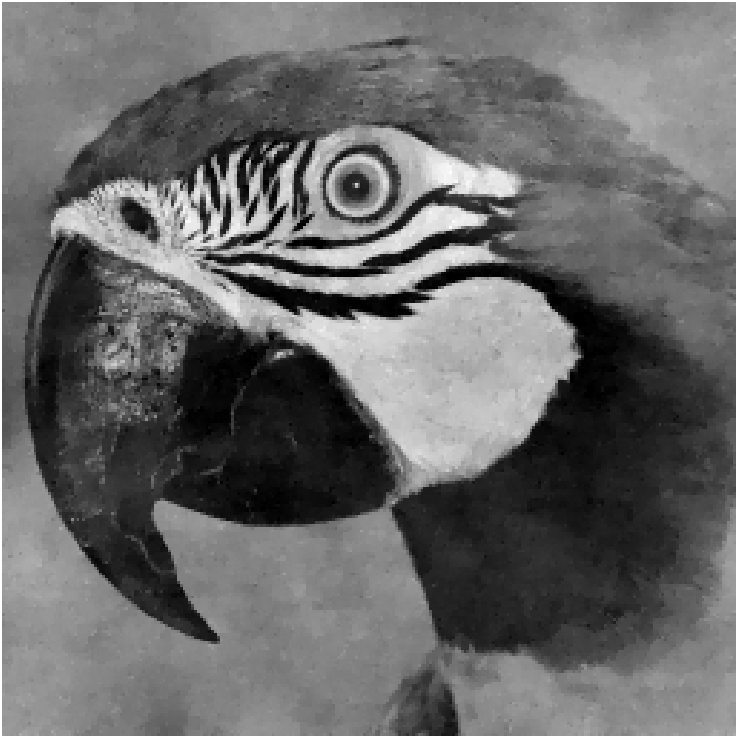}\\
  ZP& TV &$L_1-\alpha L_2$&TTV&PSV$_{a,p}$
  \end{tabular}
  \caption{Reconstructed images from data with Gaussian
  noise under 40\% sampling rates via  different methods.}\label{fig-Gnoisy40}
\end{figure}

In addition, by comparing Tables \ref{table-a} and \ref{table-p},
we find that, with $p=1$, varying $a$ yields limited changes​ in reconstruction performance.
In particular, when $a$ is already large, slight adjustments to $a$
barely affect the performance.
In contrast​, an appropriate selection of $p$ yields notable performance gains.
As shown in Table \ref{table-p}, with $(a,p)=(1,0.6)$,
the PSNR of Cameraman​ under the sampling rate of 60\% is increased to 39.5400.
These findings demonstrate that the PSV$_{a,p}$ model is more sensitive to $p$ than $a$,
which are consistent with the aforementioned theoretical analysis.
It also suggests a \emph{parameter tuning scheme}
that one should  primarily tune $p$ to achieve
the best overall reconstruction performance based on the  image content,
followed by fine-tuning of $a$ to achieve targeted optimization of specific metrics.

Next, to evaluate the overall performance and robustness of the proposed PSV$_{a,p}$ model,
taking the 40\% sampling rate as a case study, we conduct comparative experiments
between the PSV$_{a,p}$ against the ZP, the TV, the $L_1-\alpha L_2$, and the TTV
in both noise-free and noisy  scenarios.
Here, and thereafter, the \emph{(relative) noise level} is
defined as the relative error in the transformed domain
$\frac{\|y_{\mathrm{noisy}}-y_{\mathrm{orig}}\|_2}{\|y_{\mathrm{orig}}\|_2}$.
We  consider two types of noises: Gaussian noise and Gaussian--Poisson mixed noise.
In these experiments,  the Gaussian noise level is about $5\%$ while the mixed noise level is about
$6\%$ for Cameraman and Starfish and about $7\%$ for Bird.
To achieve favorable  performance,  the parameter for the $L_1-\alpha L_2$ is set to
$\alpha=0.5$ while the  parameter for the TTV  is set to $a=1$.
Notably, the parameter $p$ of the PSV$_{a,p}$ is manually tuned
while the other $a$ is uniformly set to 1 and no longer tuned.​
Quantitatively, the values of the PSNR, the SSIM, and the GMSD
of each reconstructed image are collected in Table \ref{table-NIR40}.
Visually, the reconstructed images are presented in Figures \ref{fig-noiseless40}, \ref{fig-Gnoisy40}, and \ref{fig-GBnoisy40}.
Temporally, the total runtime per method is reported in Table \ref{table-time1}.

\begin{table}[H]
\centering
\resizebox{0.95\textwidth}{!}{
\renewcommand{\arraystretch}{1.2}
\begin{tabular}{c|c|ccc|ccc|ccc}
\hline
\multirow{2}{*}{Method}& Noise &\multicolumn{3}{c|}{Noiseless} & \multicolumn{3}{c|}{Gaussian noise}
& \multicolumn{3}{c}{Gaussian--Poisson noise} \\
\cline{2-11}
& Image & Cameraman & Starfish & Bird & Cameraman & Starfish & Bird & Cameraman & Starfish & Bird\\
\hline
\multirow{3}{*}{ZP}
& PSNR & 15.7844 &15.4803 &15.0305 &15.5774 & 15.4116& 15.1774& 15.8203&15.4856 & 15.0232 \\
& SSIM & 0.4461 & 0.5251 & 0.4353 & 0.4297 &0.5128 & 0.4257& 0.4462& 0.5252 & 0.4352\\
& GMSD & 0.2816 & 0.2344 & 0.2825 & 0.2841&0.2363& 0.2848& 0.2815& 0.2346 & 0.2822\\
\hline
\multirow{4}{*}{TV} & $\lambda$ & 6.0e-3 & 7.0e-3&  6.0e-3&  1.0e-2&
9.0e-3& 9.0e-3& 7.0e-3 & 7.0e-3& 9.0e-3 \\
& PSNR & 24.2748 & 17.6811 &18.6855 &23.7082 &17.4788 & 18.3922& 24.2432&17.7260 & 18.5081\\
& SSIM & 0.8139 & 0.6970&0.7701 & 0.7446& 0.6727 & 0.7283 & 0.7964& 0.6977&0.7308 \\
& GMSD  & 0.1744  & 0.2118&0.1930& 0.2116& 0.2145 & 0.2085 & 0.1852& 0.2116 &0.2053\\
\hline
\multirow{5}{*}{$L_1-\alpha L_2$}
& $\lambda$ & 2.0e-3 & 3.0e-3 & 3.0e-3 &5.0e-3 &
6.0e-3& 6.0e-3 & 2.0e-3 &3.0e-3& 3.0e-3\\
& PSNR &27.9457 &18.2336& 22.1634 &25.2320&17.5634&20.1035&27.9247&18.2493&22.0846\\
& SSIM &0.8920&0.7280&0.8485 &0.8104&0.6797&0.7750&0.8912&0.7278&0.8474\\
& GMSD  &0.1408 &0.2178&0.1766&\textbf{0.1873}&0.2218&0.2025&0.1422&0.2176&0.1771\\
\hline
\multirow{4}{*}{TTV}
& $\lambda$ & 8.0e-4 & 9.0e-4 & 1.0e-3& 3.0e-3
& 2.0e-3 & 2.0e-3& 1.0e-3& 8.0e-4& 1.0e-3\\
& PSNR & 28.1226& 18.7618& 19.5340& 25.3004&18.5846&18.0636& 27.9243 & 18.6976& 19.4161\\
& SSIM & 0.8778&0.7381&0.8126&0.7806&0.6935&0.7252&0.8751&0.7371&0.8104\\
& GMSD&0.1499&0.2048&0.1836&0.2054&0.2141& 0.2168&0.1511&0.2040&0.1843\\
\hline
\multirow{5}{*}{PSV$_{a,p}$} & $(a,p)$ & (1,0.7) & (1,0.5) & (1,0.7) & (1,0.7) & (1,0.5) & (1,0.7) & (1,0.9) & (1,0.5) & (1,0.7)\\
& $\lambda$ & 1.2e-4 & 7.0e-5 & 1.8e-4 & 3.4e-4 & 2.0e-4 &
3.8e-4 & 4.0e-4 & 7.0e-5 & 1.8e-4\\
& PSNR & \textbf{30.6976}& \textbf{19.2616} & \textbf{22.9646} & \textbf{28.1043} & \textbf{19.1856}
& \textbf{21.8626} & \textbf{28.3286} & \textbf{19.2554} & \textbf{22.9513} \\
& SSIM &\textbf{0.9174}&\textbf{0.7752}&\textbf{0.8874}& \textbf{0.8242}&\textbf{0.7286}
&\textbf{0.8100}&\textbf{0.8923}&\textbf{0.7744}&\textbf{0.8857}\\
& GMSD &\textbf{0.1230}& \textbf{0.2014}&\textbf{0.1516}&0.1910&\textbf{0.2091}&
\textbf{0.1931}&\textbf{0.1361}&\textbf{0.2007}&\textbf{0.1520}\\
\hline
\end{tabular}}
\caption{Comparison of quality metrics for reconstructed images
under 40\% sampling rates across multiple noise scenarios via different methods.}
\label{table-NIR40}
\end{table}

As shown in Table \ref{table-NIR40}, with properly tuned parameters,
the proposed model  achieves
the highest PSNR and SSIM  and the lowest GMSD values in the vast majority of cases
across all test images and noise conditions, demonstrating its broad applicability.
From Figures \ref{fig-noiseless40}, \ref{fig-Gnoisy40}, and \ref{fig-GBnoisy40}, the images reconstructed by the ZP suffer from
severe noise contamination and staircasing artifacts
while other methods effectively suppress such effects.
In particular, the zoom-in view of the regions marked by the red box of Bird are shown in Figure \ref{fig-zoomin}.
In comparison, the proposed method further mitigates the staircase effect and preserves sharp edges.
Moreover, combined with Table \ref{table-time1}, we find that
(i) while the ZP reconstruction offers the fastest speed, it lacks reconstruction quality,
(ii) while the TV method incurs low computational cost,
its reconstruction quality still requires further improvement,
(iii) although the $L_1-\alpha L_2$ and the TTV method notably improve the reconstruction quality,
it takes approximately 9 to 20 times longer than the TV method,
(iii) the proposed PSV$_{a,p}$ requires computational cost
comparable to those of both the $L_1-\alpha L_2$ and the TTV,
while achieving a further improvement across multiple metrics under various testing scenarios.
In this sense, the proposed PSV$_{a,p}$ model demonstrates a favorable trade-off
between reconstruction quality and computational cost.

\begin{figure}[t]
  \centering
  \begin{tabular}{ccccc}
  \includegraphics[width=2.5cm]{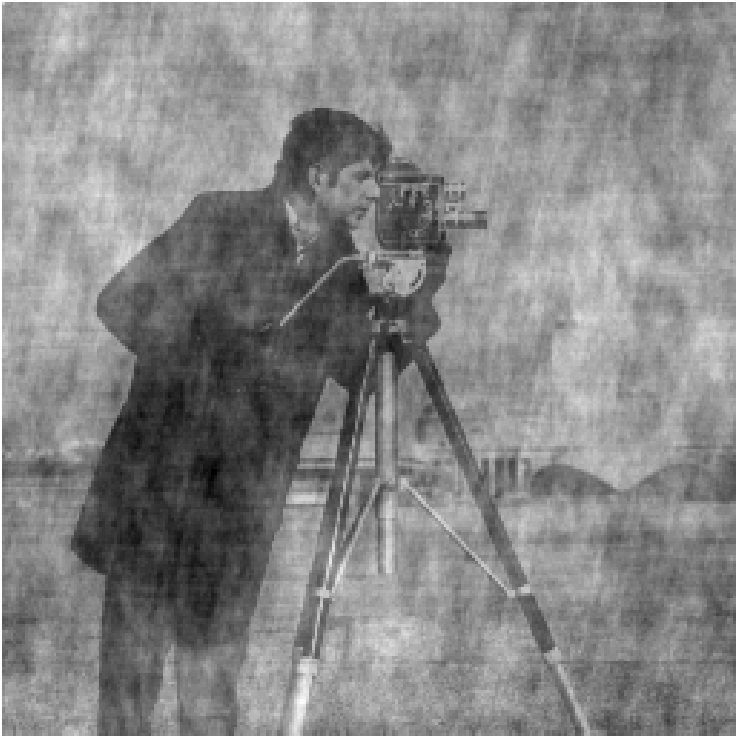}&
  \includegraphics[width=2.5cm]{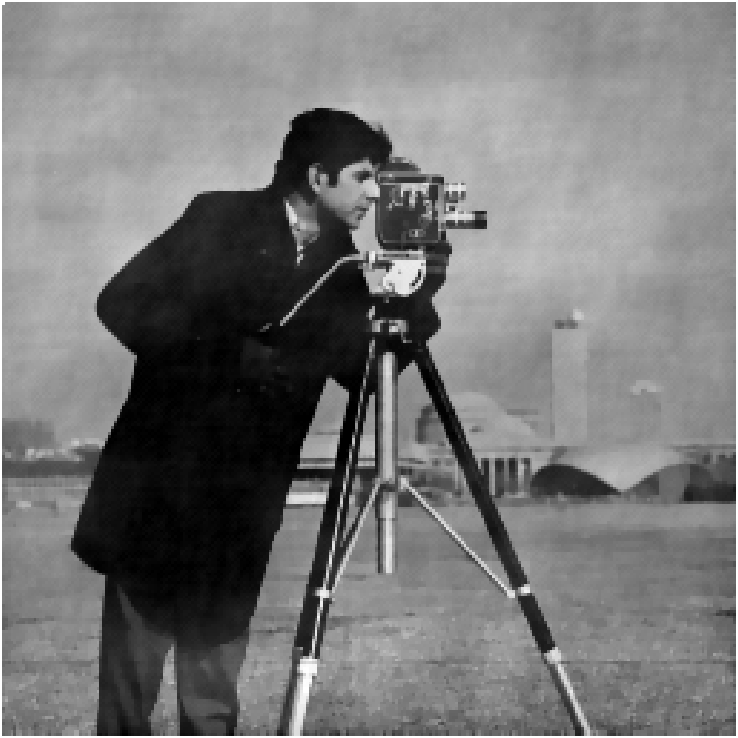}&
  \includegraphics[width=2.5cm]{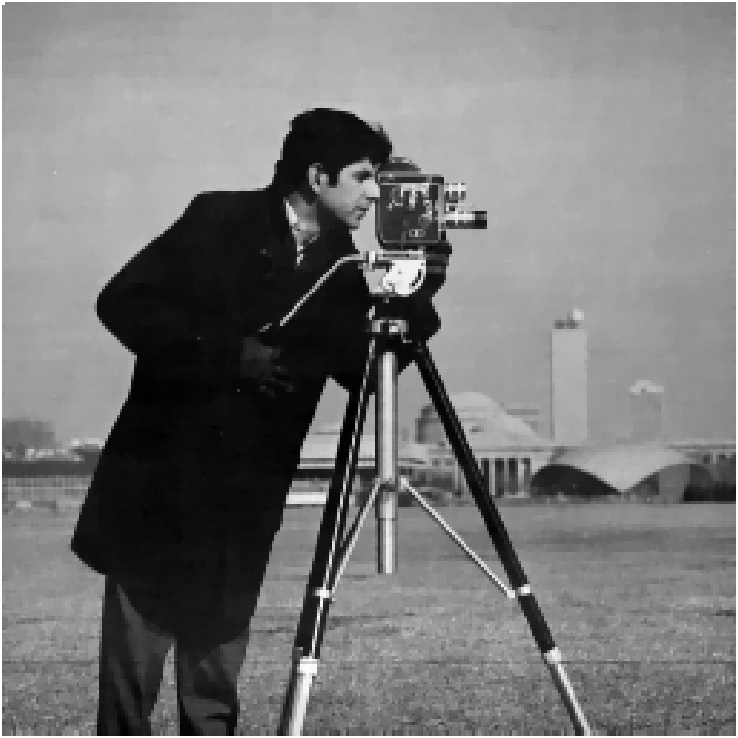}&
  \includegraphics[width=2.5cm]{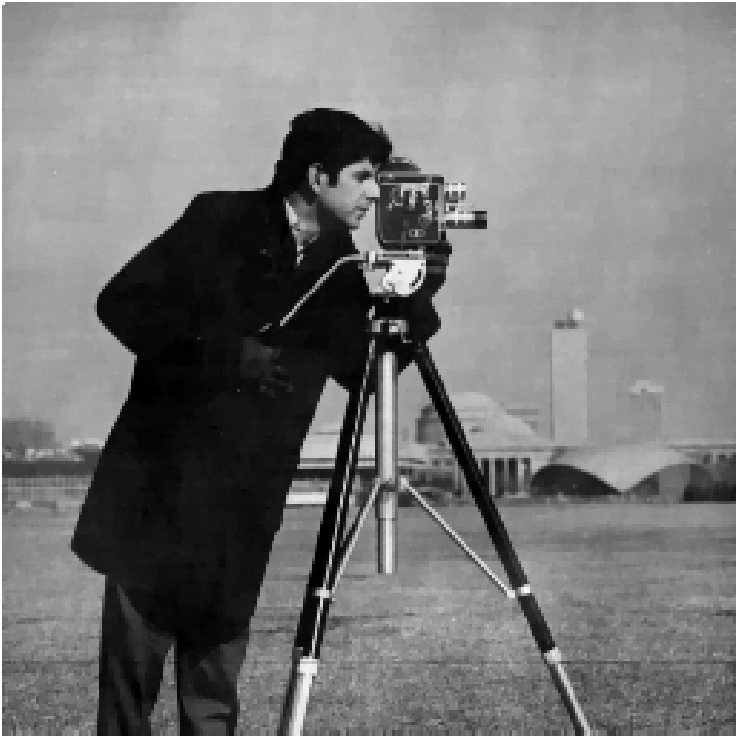}&
  \includegraphics[width=2.5cm]{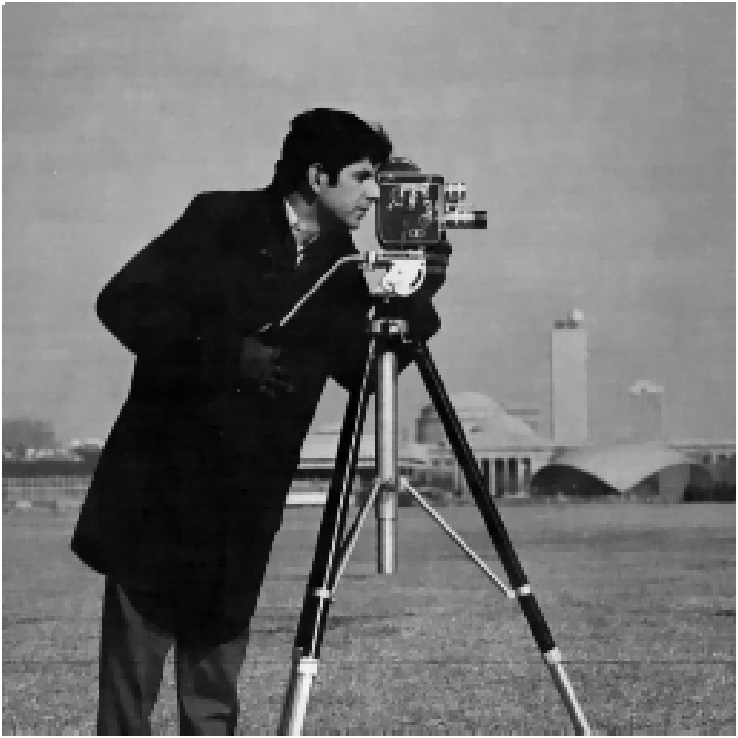}\\
  \includegraphics[width=2.5cm]{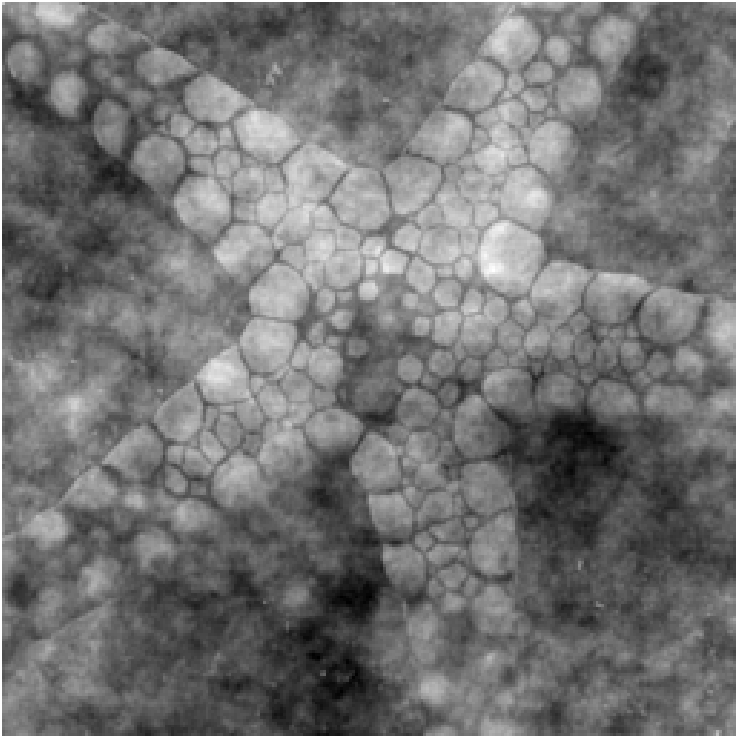}&
  \includegraphics[width=2.5cm]{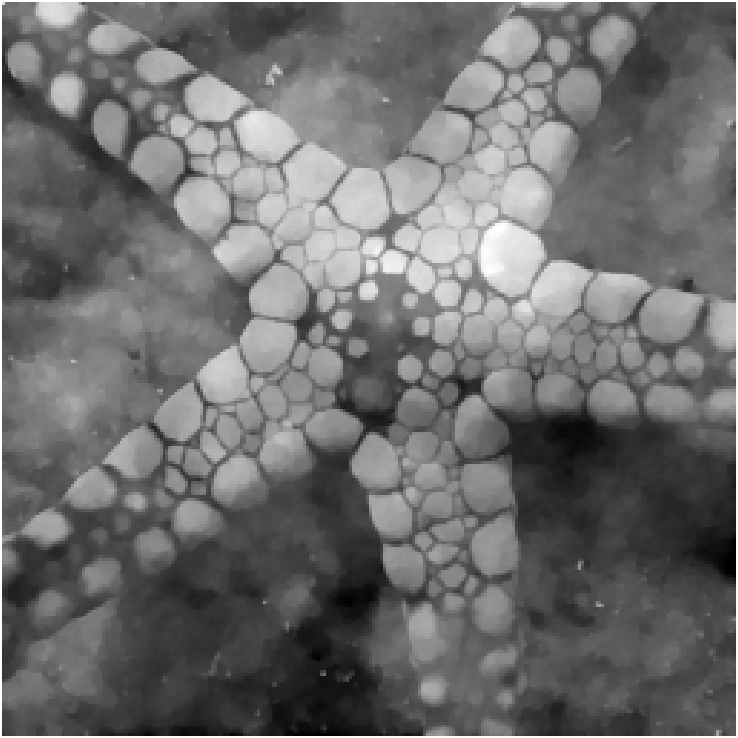}&
  \includegraphics[width=2.5cm]{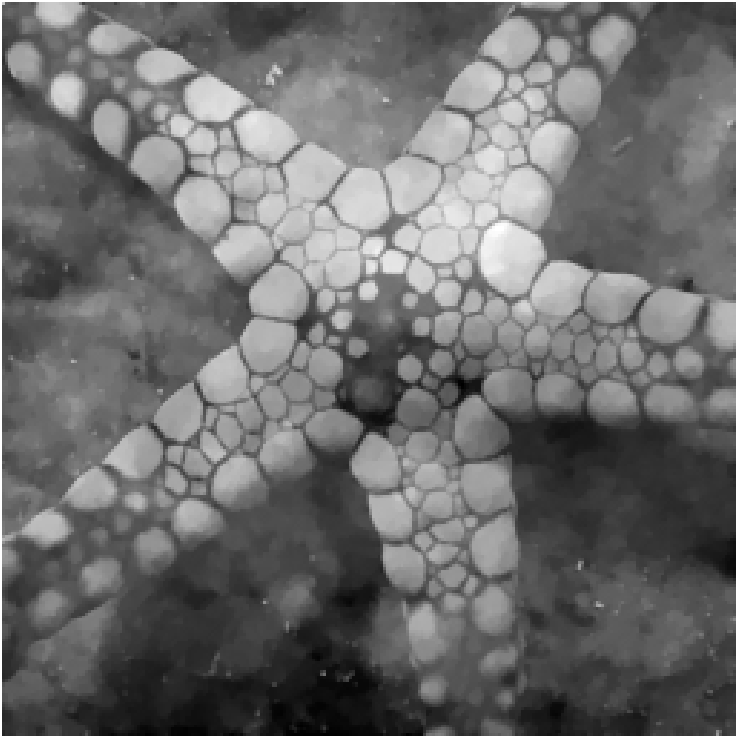}&
  \includegraphics[width=2.5cm]{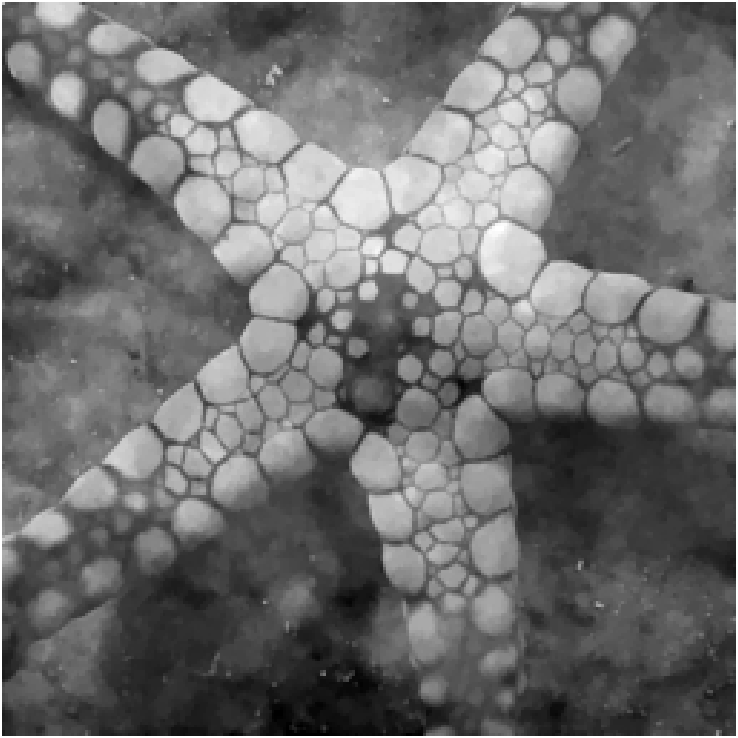}&
  \includegraphics[width=2.5cm]{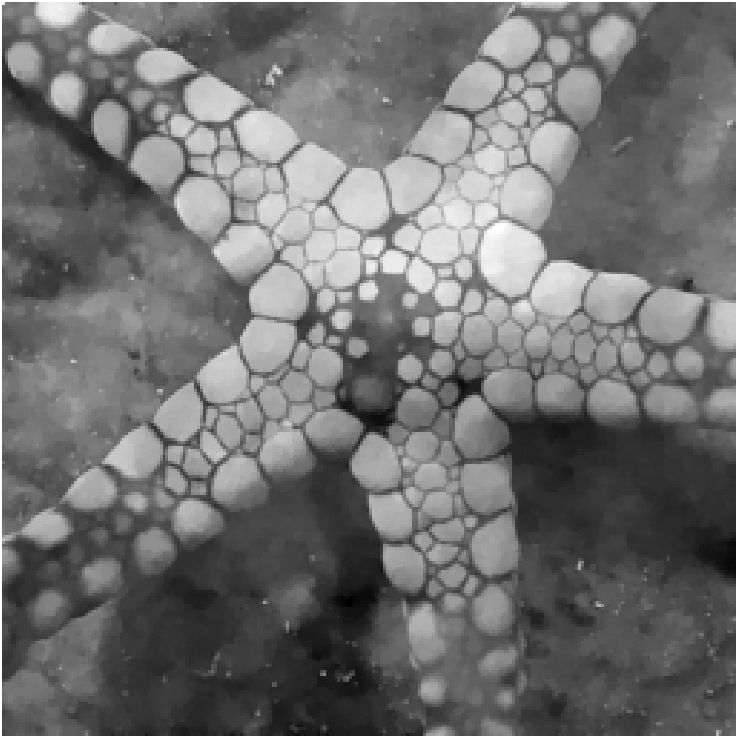}\\
  \includegraphics[width=2.5cm]{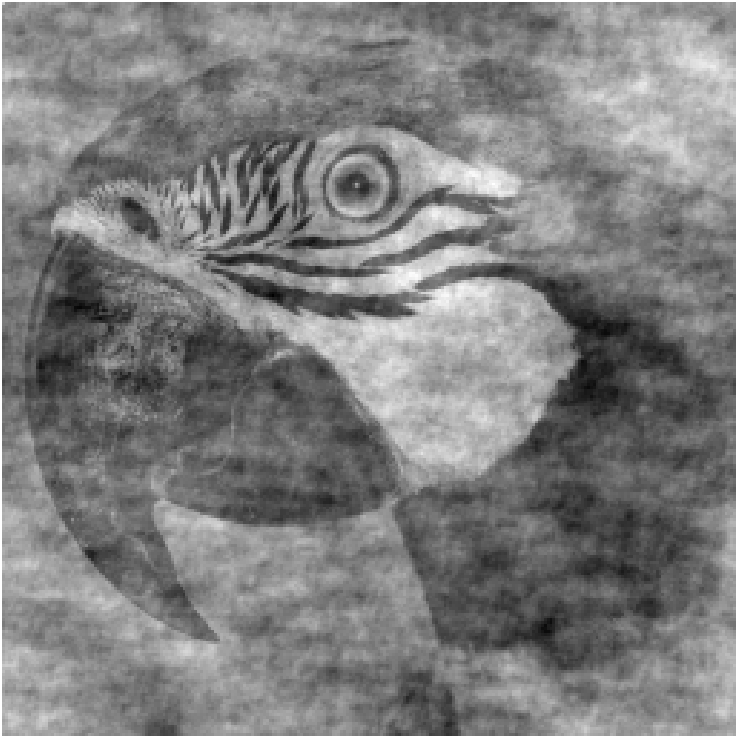}&
  \includegraphics[width=2.5cm]{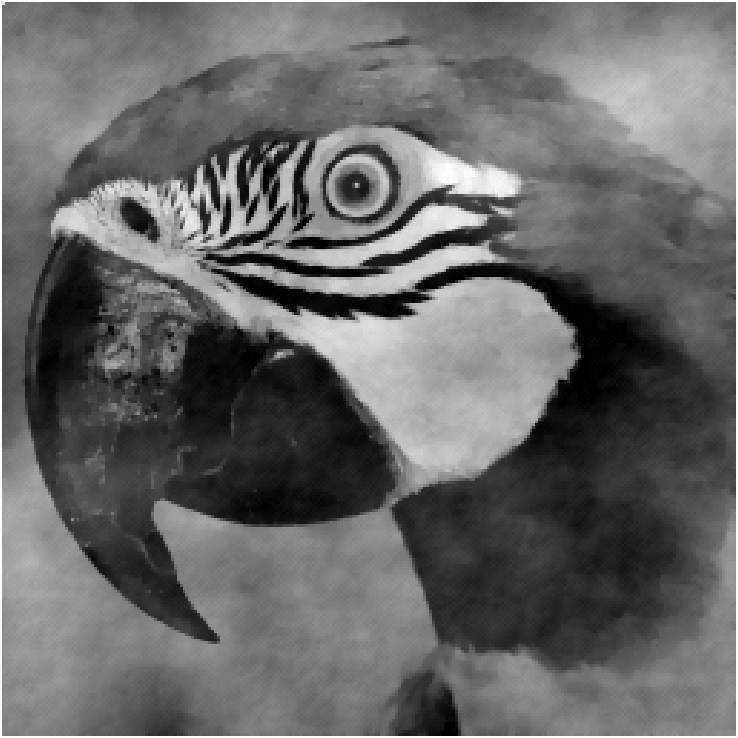}&
  \includegraphics[width=2.5cm]{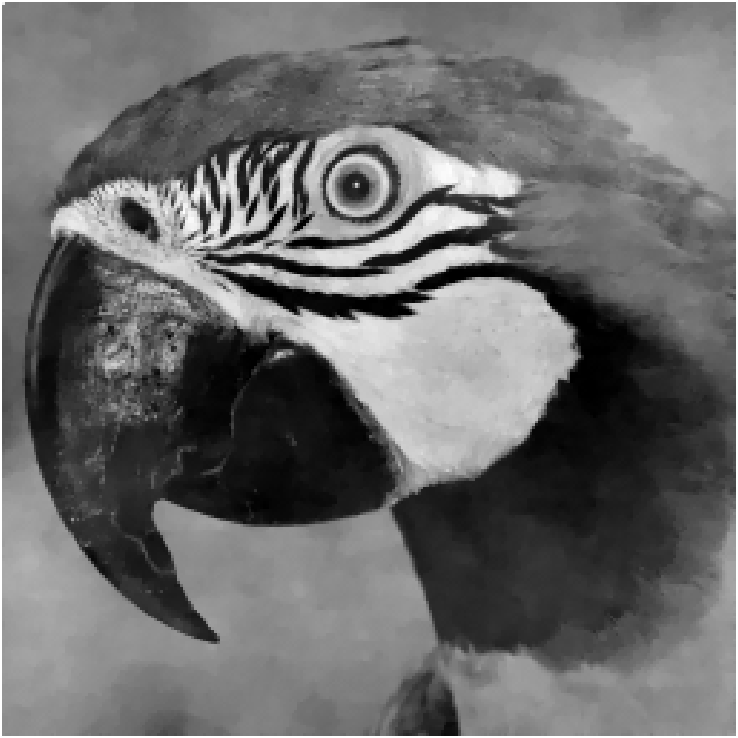}&
  \includegraphics[width=2.5cm]{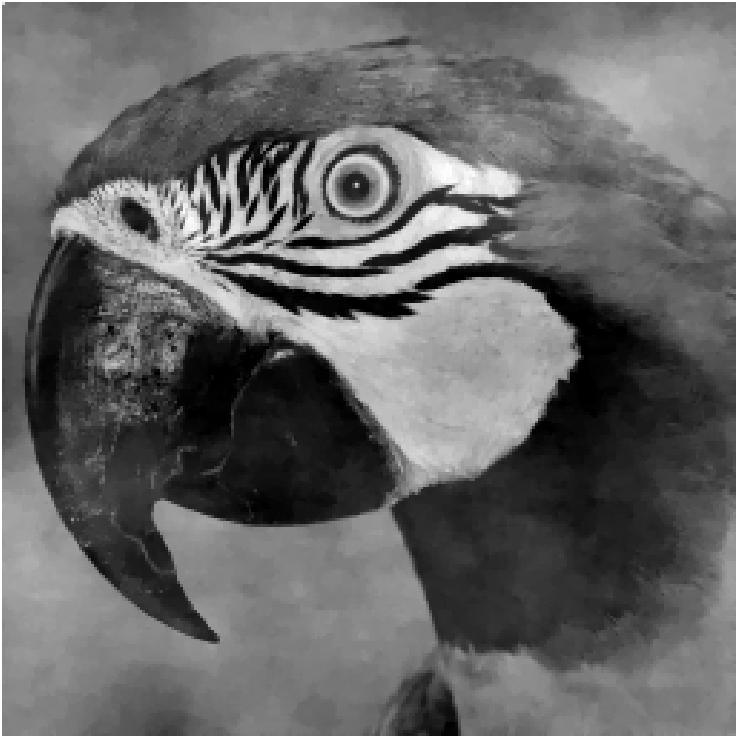}&
  \includegraphics[width=2.5cm]{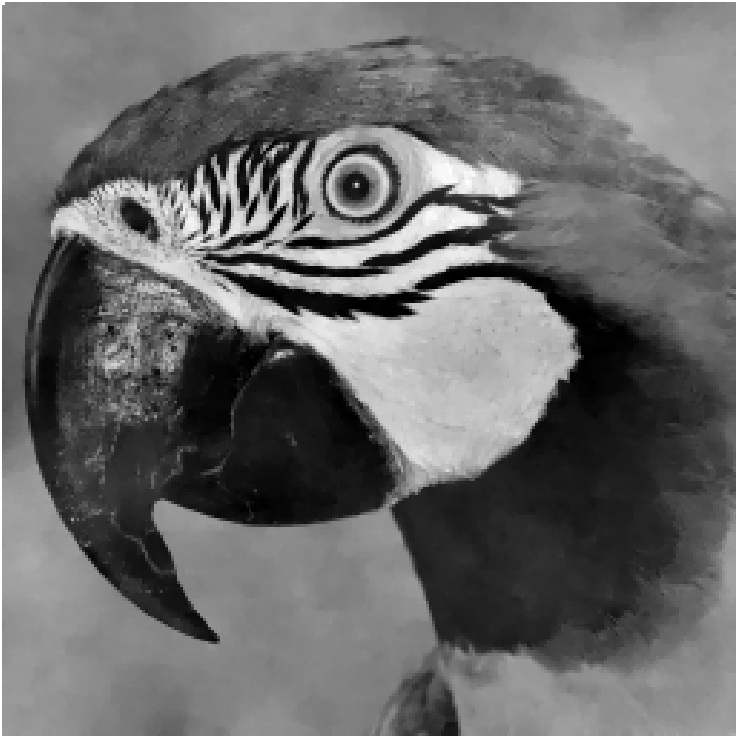}\\
  ZP& TV &$L_1-\alpha L_2$&TTV&PSV$_{a,p}$
  \end{tabular}
  \caption{Reconstructed images from data with mixed Gaussian--Poisson noise
  under 40\% sampling rates via different methods. }\label{fig-GBnoisy40}
\end{figure}

\begin{figure}[H]
  \centering
  \begin{tabular}{ccccc}
  \includegraphics[width=2.5cm]{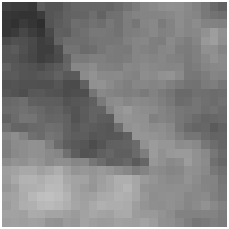}&
  \includegraphics[width=2.5cm]{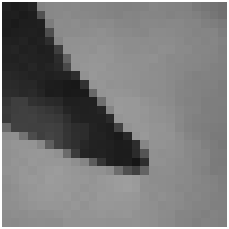}&
  \includegraphics[width=2.5cm]{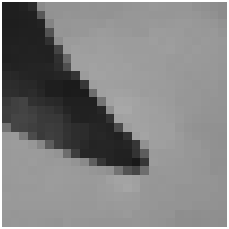}&
  \includegraphics[width=2.5cm]{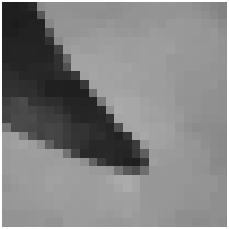}&
  \includegraphics[width=2.5cm]{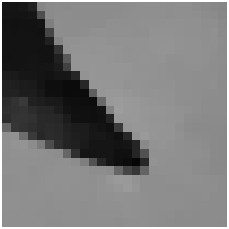}\\
  ZP& TV &$L_1-\alpha L_2$&TTV&PSV$_{a,p}$
  \end{tabular}
  \caption{Zoom-in view of the regions marked by the red box in Figure \ref{fig-noiseless40}.}\label{fig-zoomin}
\end{figure}

\begin{table}[H]
\centering
\renewcommand{\arraystretch}{1.2}
  \begin{tabular}{cccccc}
  \hline
  Method&ZP&TV&$L_1-\alpha L_2$&TTV&PSV$_{a,p}$\\
  \hline
  Cameraman &0.0037 & 4.0210&78.6043 & 69.4403 & 75.3755\\
    \hline
  Starfish & 0.0032 & 3.7917& 78.5817 & 68.4756 & 72.1294\\
    \hline
  Bird & 0.0025 & 3.8094 & 78.8976&27.8411 & 76.2701\\
  \hline
  \end{tabular}
  \caption{Total runtime (in seconds) per method for all experiments listed in Table \ref{table-NIR40}.}
  \label{table-time1}
 \end{table}

\subsection{MRI Reconstruction}\label{subsec-MRI}

In this subsection, we apply the proposed PSV$_{a,p}$
model to the magnetic resonance imaging (MRI) reconstruction.
Comparative experiments are conducted to exhibit the superiority
in both global structure preservation (high SSIM) 
and local detail preservation (low GMSD)
of the proposed PSV$_{a,p}$ model for the MRI reconstruction.

We sample the measurement along 24, 32, and 40
radial lines  in the frequency domain
after taking the Fourier transform of three brain images
from the BrainWeb Simulated Brain Database\footnote{Available at https://brainweb.bic.mni.mcgill.ca/brainweb/selection\_ normal.html}
(see Figure \ref{fig-linemask}).
To be precise, the operator $\Phi$ in \eqref{eq-unproblem} represents a composed operator $\Phi:=\widetilde{S}F$,
where $F$ denotes the Fourier matrix and $\widetilde{S}$
the sampling matrix along  radial lines.

\begin{figure}[H]
  \centering
  \begin{tabular}{cccccc}
  \includegraphics[width=2cm]{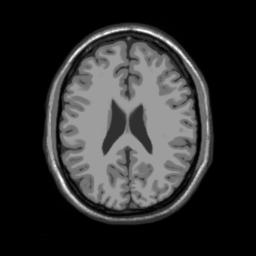}&
  \includegraphics[width=2cm]{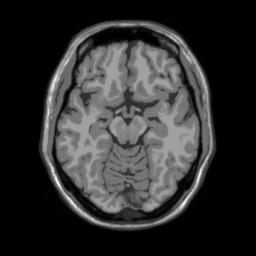}&
  \includegraphics[width=2cm]{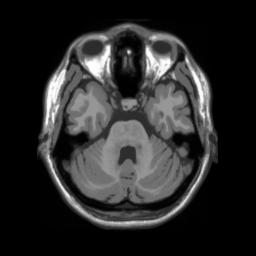}&
  \includegraphics[width=2cm]{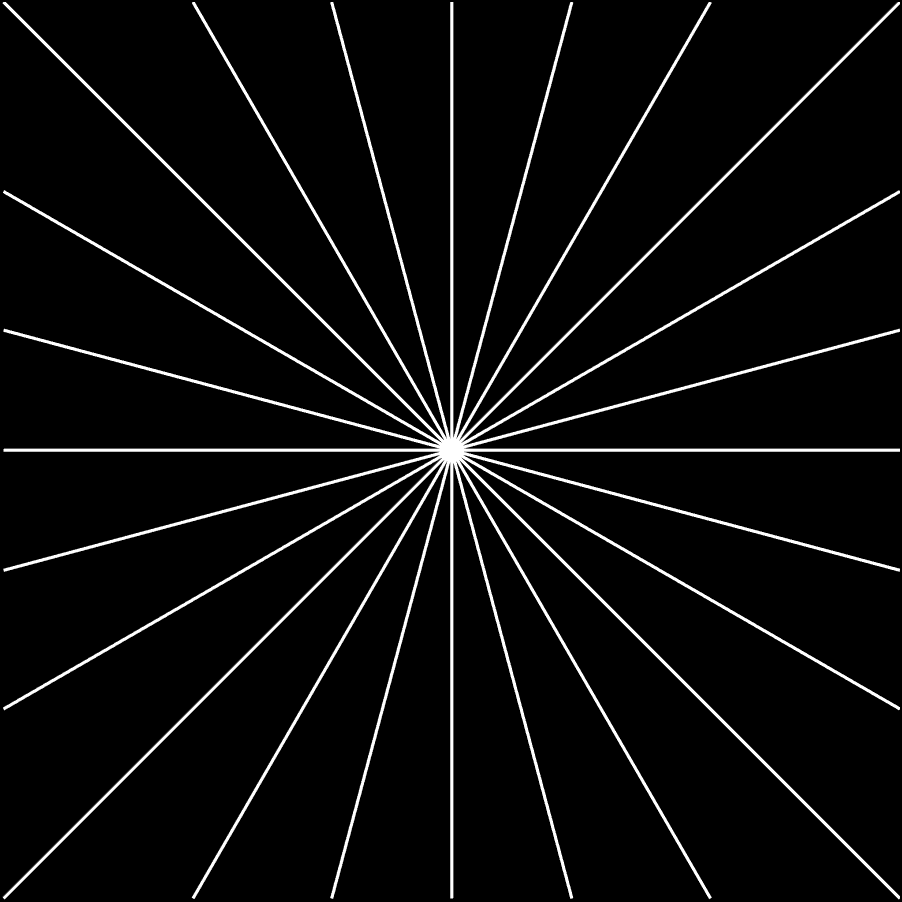}&
  \includegraphics[width=2cm]{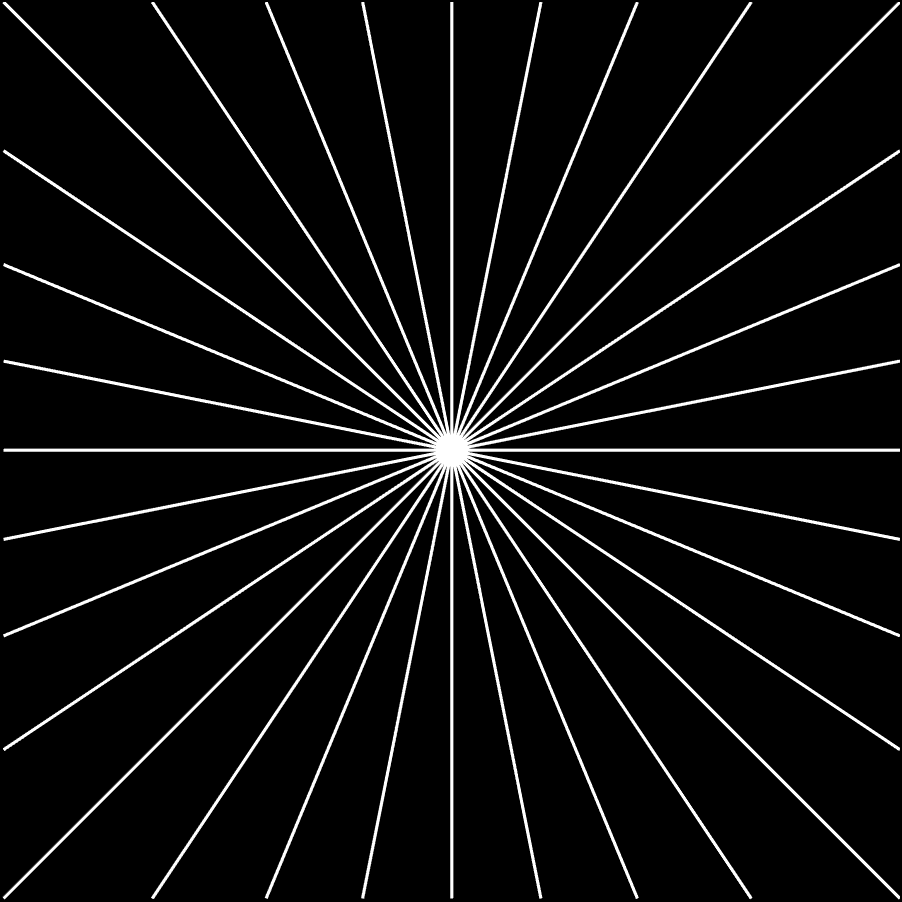}&
  \includegraphics[width=2cm]{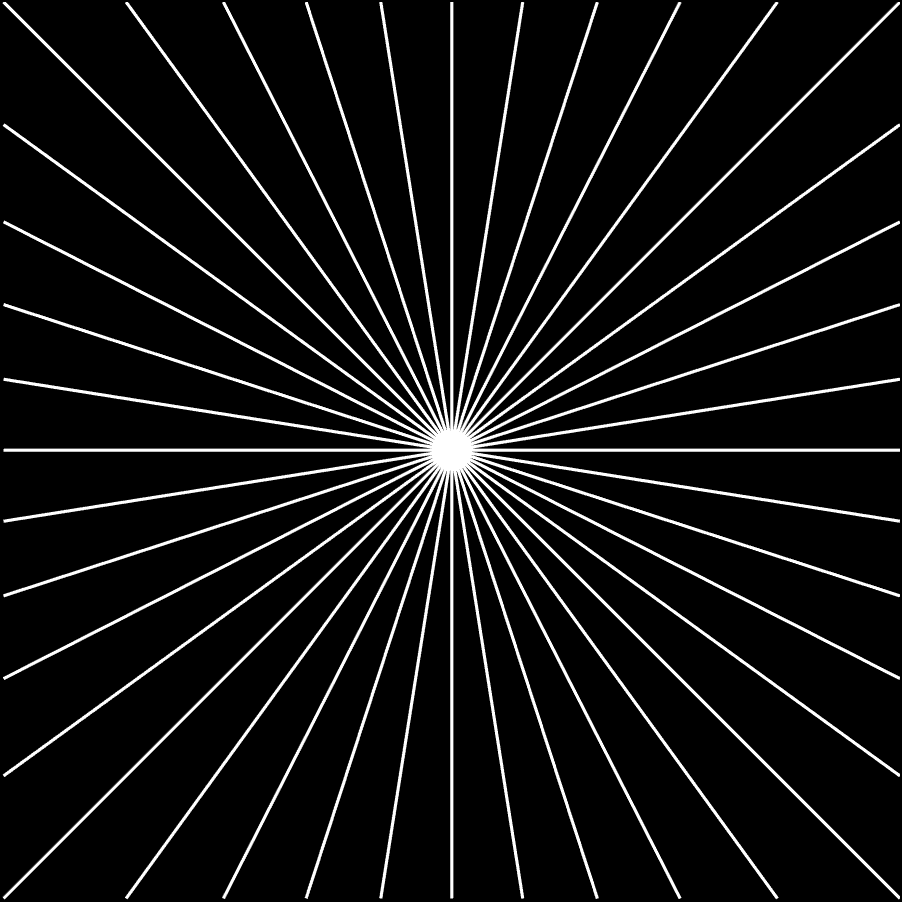}\\
  Brain A & Brain B& Brain C&24  lines & 32   lines & 40  lines
  \end{tabular}
  \caption{Three original brain images (A, B, and C) and
  three radial sampling patterns (24, 32, and 40 lines) without the k-space center.}\label{fig-linemask}
\end{figure}

The stopping criteria of the IRLSPSV algorithm
are $k_{\mathrm{out}}=200$ for the outer loop,
$k_{\mathrm{mid}}=10$ and $\varepsilon_{\mathrm{mid}}=5\times 10^{-4}$ for the mid loop,
and $k_{\mathrm{inn}}=5$ and  $\varepsilon_{\mathrm{inn}}=10^{-4}$ for the inner loop.

Guided by the sensitivity study in the last subsection,
we employ the same parameter tuning scheme: primarily tuning  $p$,
followed by fine-tuning of $a$. As a preliminary attempt,
we use 24 sampling lines to assess the reconstruction quality of Brains A, B, and C,
under different settings of $p$ with $a=1$ (see Table \ref{table-choicep}).
Based on Table \ref{table-choicep}, in the following experiments,
to achieve a better trade-off among quality metrics,
we fix $p=0.5$ and  tune $a$
according to the specific image and testing scenario.

\begin{table}[H]
\centering
\resizebox{1\textwidth}{!}{
\renewcommand{\arraystretch}{1.19}
\begin{tabular}{c|cccc|cccc|cccc}
\hline
Image & \multicolumn{4}{c|}{Brain A} & \multicolumn{4}{c|}{Brain B}& \multicolumn{4}{c}{Brain C}\\
\hline
$p$& 1 &0.7 &0.5&0.3 & 1 &0.7 &0.5&0.3& 1 &0.7 &0.5&0.3  \\
\hline
$\lambda$& 1.0e-5 &1.5e-5&2.4e-5&4.0e-5 & 1.0e-5 
&2.5e-5&8.0e-5&9.0e-5& 1.0e-5 &2.5e-5&6.5e-5&7.0e-5 \\
\hline
PSNR& 24.3947 &24.9772&\textbf{26.6741}&24.2788
& 25.3709 &25.1265&\textbf{25.4786}&25.4731
&26.8444&\textbf{27.3765}&27.0467&26.6152\\
\hline
SSIM& 0.7327 & 0.7344&0.7509&\textbf{0.8021}
& 0.7648 & 0.7411&0.7972&\textbf{0.8208}
&0.7471&0.7332&0.7811&\textbf{0.7978}\\
\hline
GMSD& 0.2737&0.2520&\textbf{0.2464}&0.2776
& 0.2283&0.2269&\textbf{0.2035}&0.2267
&0.2382&0.2344&\textbf{0.2159}&0.2401 \\
\hline
\end{tabular}}
\caption{Quality metrics for reconstructed images 
from noise-free data of Brains A, B, and C
under 24  radial sampling lines with $a=1$ and  
$p\in\{1,0.7,0.5,0.3\}$.}\label{table-choicep}
\end{table}

\begin{table}[H]
\centering
\renewcommand{\arraystretch}{1.2}
\begin{tabular}{c|cccc|cccc}
\hline
 \multirow{2}{*}{$a$} &\multicolumn{4}{c|}{Noise-free} & \multicolumn{4}{c}{Gaussian Noise}\\
 \cline{2-9}
 &$\lambda$& PSNR &SSIM &GMSD &$\lambda$& PSNR &SSIM&GMSD\\
 \hline
0.1 &1.0e-5& 29.6122 & 0.9157 & 0.2264 &2.0e-5& 29.3934& \textbf{0.8550} & 0.2623\\
\hline
0.5 &1.0e-5& 32.9195 & \textbf{0.9288} & \textbf{0.1524} &2.0e-5& \textbf{31.5531}& 0.8449 & \textbf{0.1818}\\
\hline
5&1.0e-5& \textbf{33.1998}  & 0.9105 &0.1678 &2.0e-5& 31.3307 & 0.7953 &  0.2191 \\
\hline
\end{tabular}
  \caption{Quality metrics for reconstructed images from noise-free and Gaussian-noise-corrupted measurements
  of Brain A under 40 radial sampling lines with $p=0.5$ and  $a\in\{0.1,0.5,5\}$.} \label{table-various_a}
\end{table}

To further  investigate the effect of varying $a$ on the reconstruction quality metrics,
taking Brain A with 40 radial lines as a case study,
we  evaluate the reconstruction performance
under both noise-free and Gaussian-noise-corrupted (about $5\%$ noise level)
conditions by varying the parameter $a$.
The values of the PSNR, the SSIM, and the GMSD recorded in Table \ref{table-various_a}
 demonstrate​ that, by adjusting $a$, the proposed model successfully achieves
targeted optimization for the enhancement of specific evaluation metrics.
This indicates that the proposed model offers the flexibility to trade
off certain metrics against others.
For instance, it is clinically desirable that
one can tolerate a slight drop in the PSNR to achieve a substantial gain
in structural preservation (SSIM) and visual perception (GMSD).

\begin{figure}[t]
  \centering
  \begin{tabular}{ccccc}
  \includegraphics[width=2.5cm]{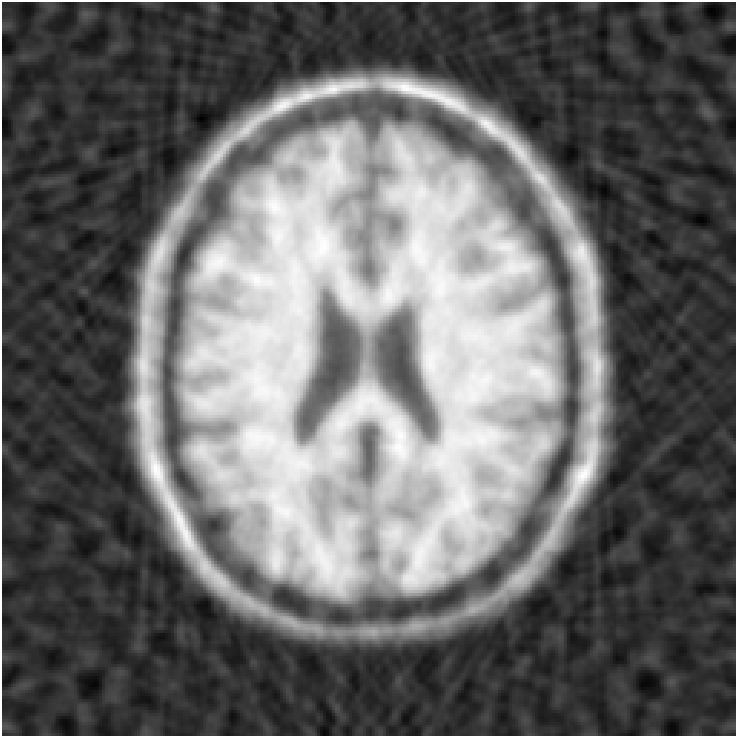}&
  \includegraphics[width=2.5cm]{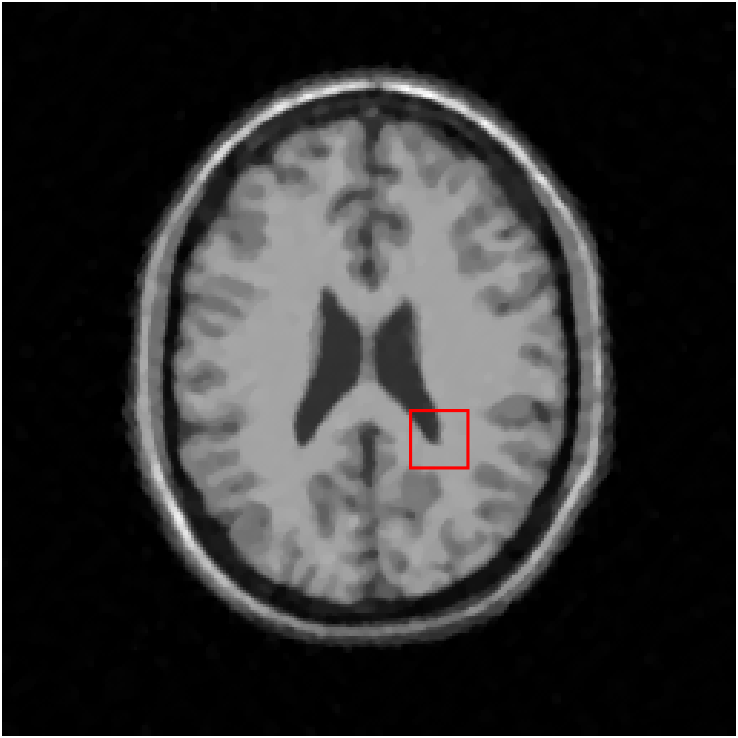}&
  \includegraphics[width=2.5cm]{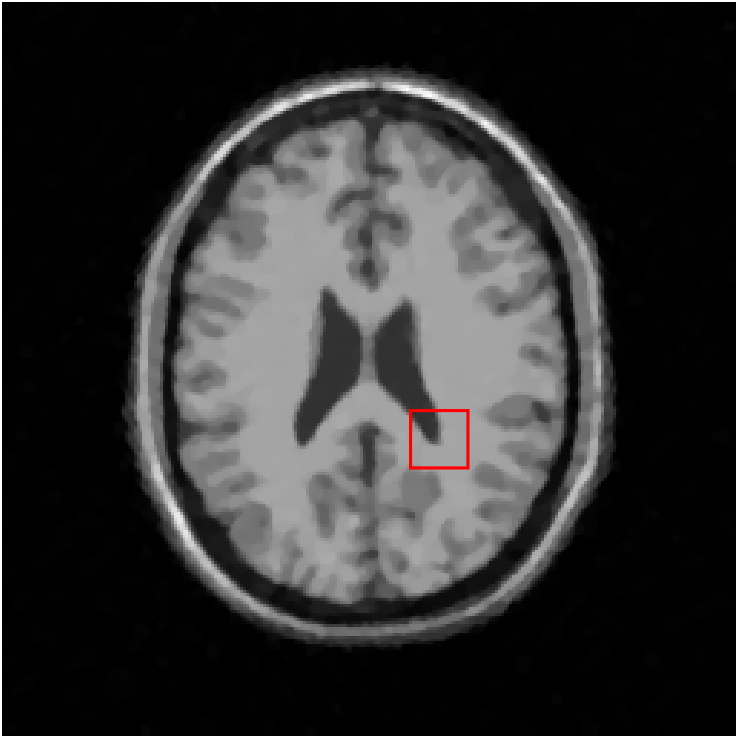}&
  \includegraphics[width=2.5cm]{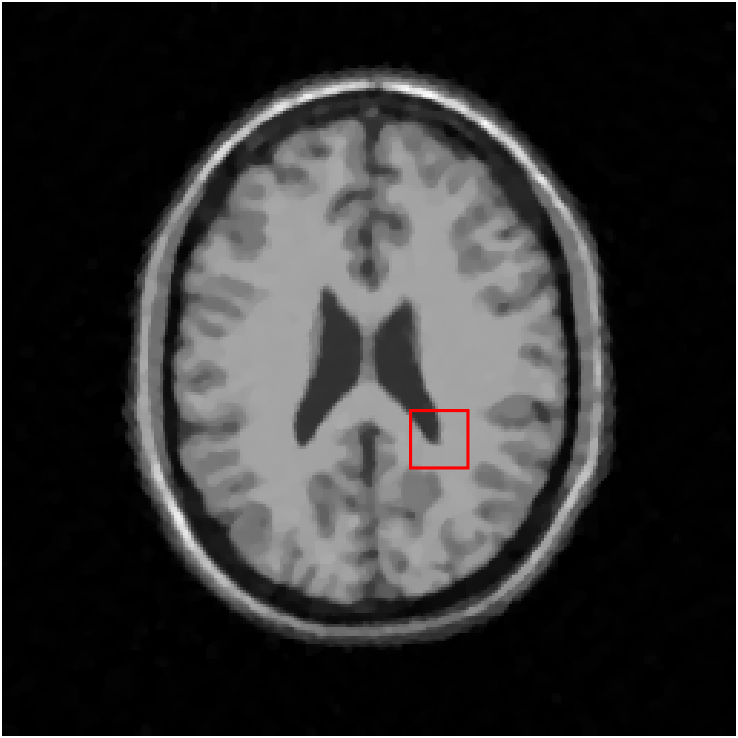}&
  \includegraphics[width=2.5cm]{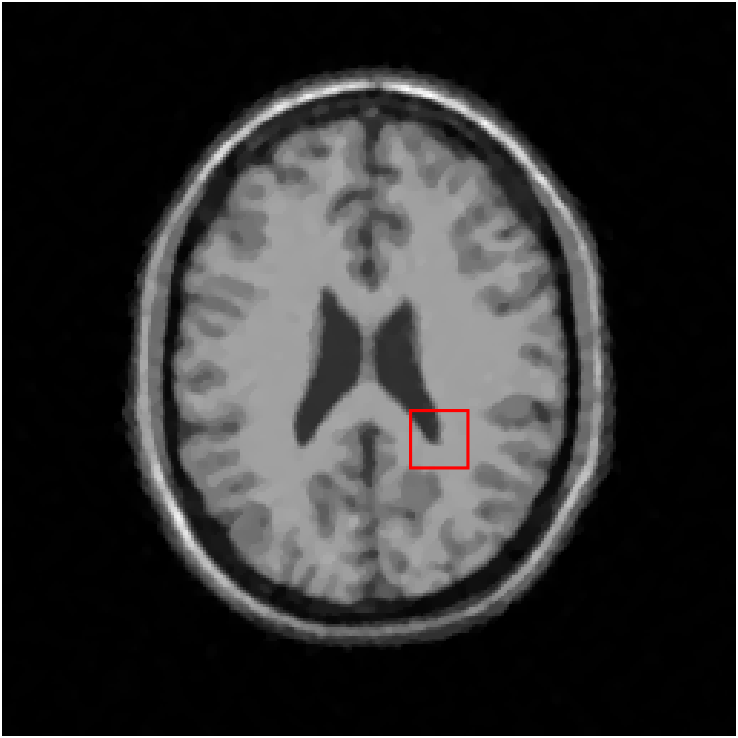}\\
  \includegraphics[width=2.5cm]{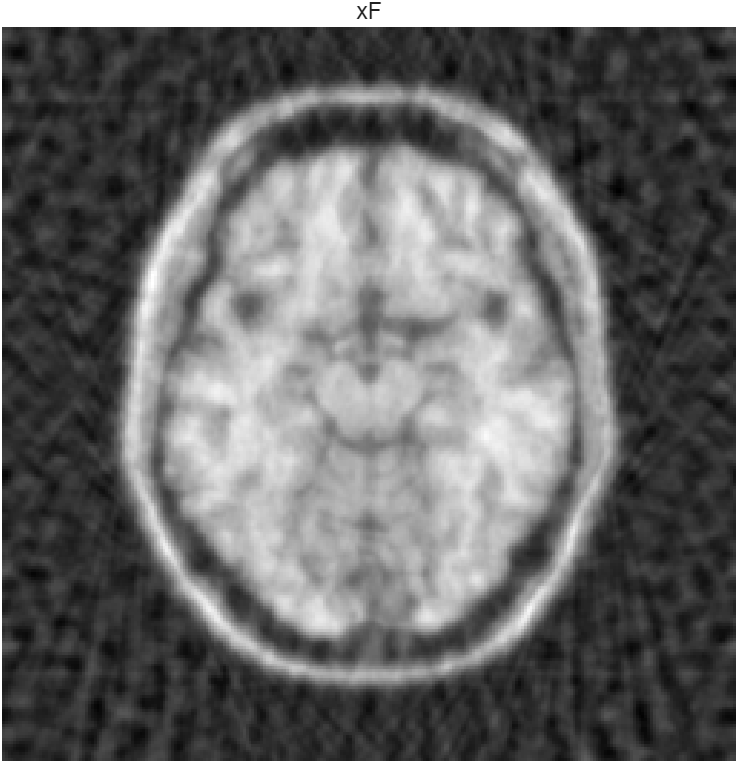}&
  \includegraphics[width=2.5cm]{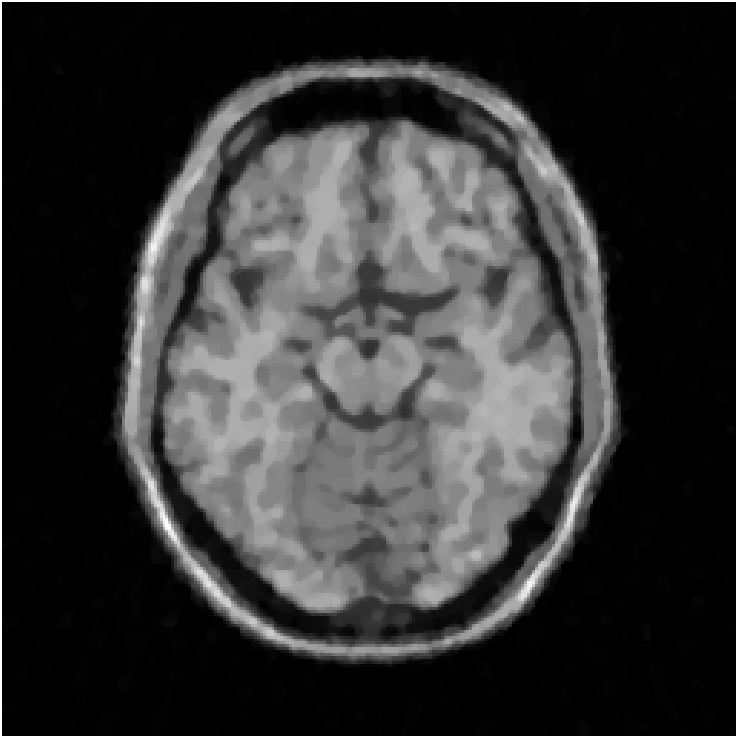}&
  \includegraphics[width=2.5cm]{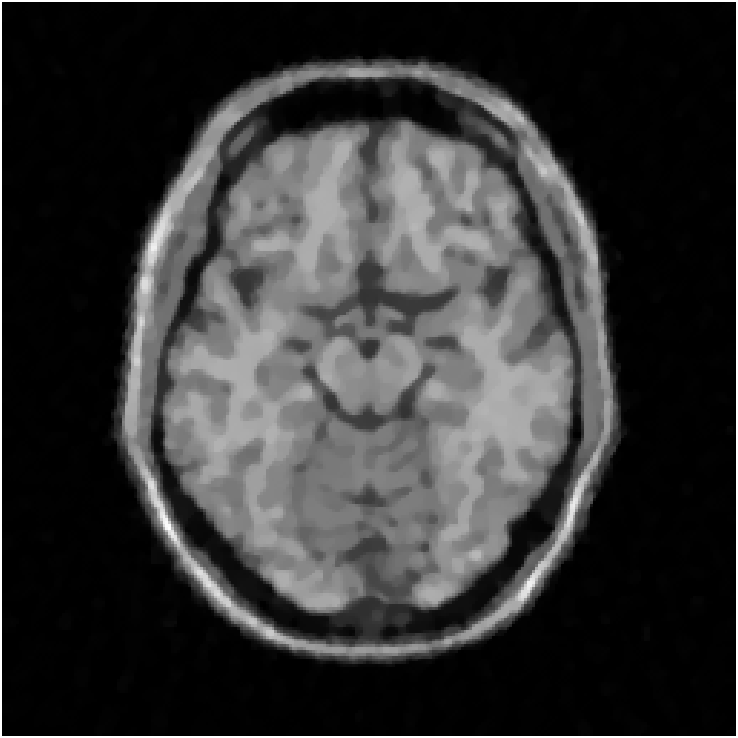}&
  \includegraphics[width=2.5cm]{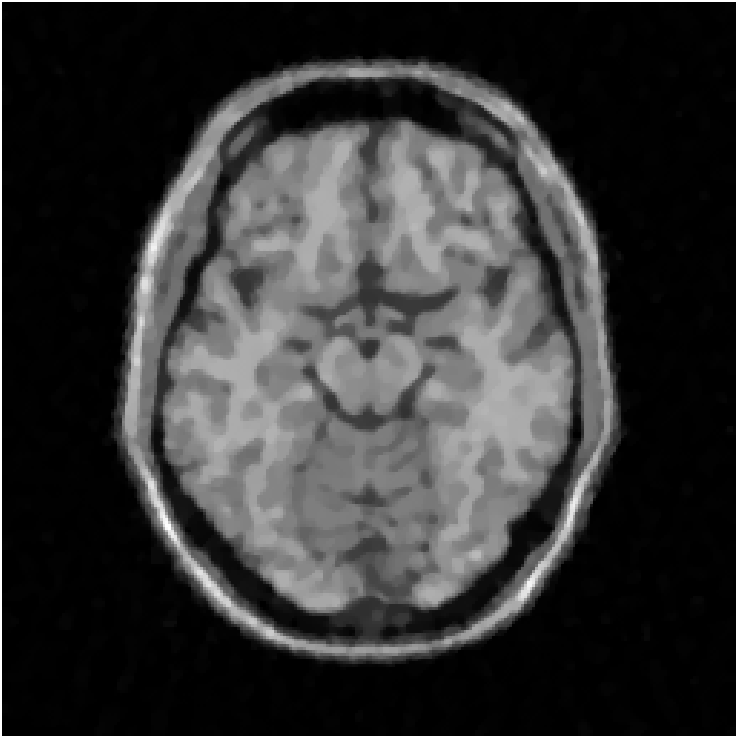}&
  \includegraphics[width=2.5cm]{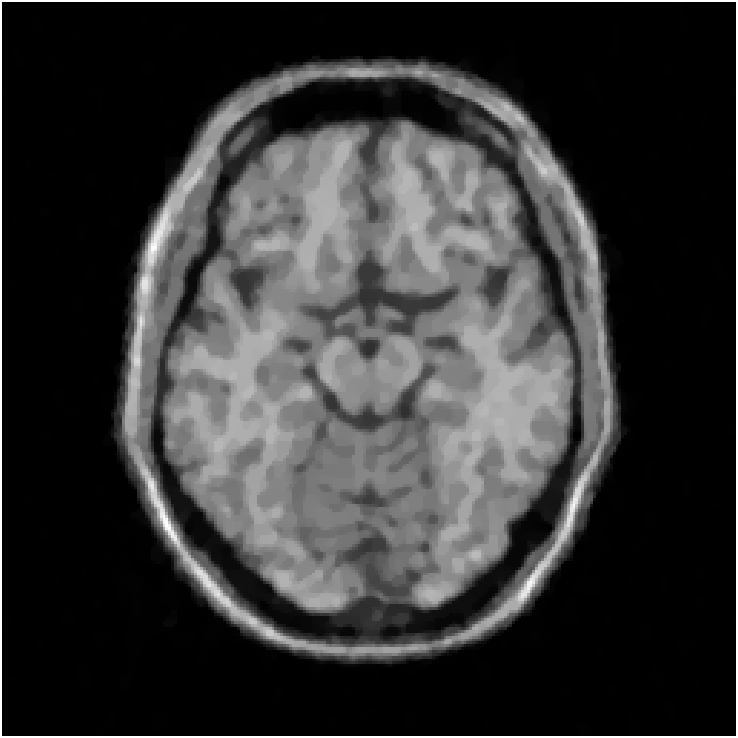}\\
  \includegraphics[width=2.5cm]{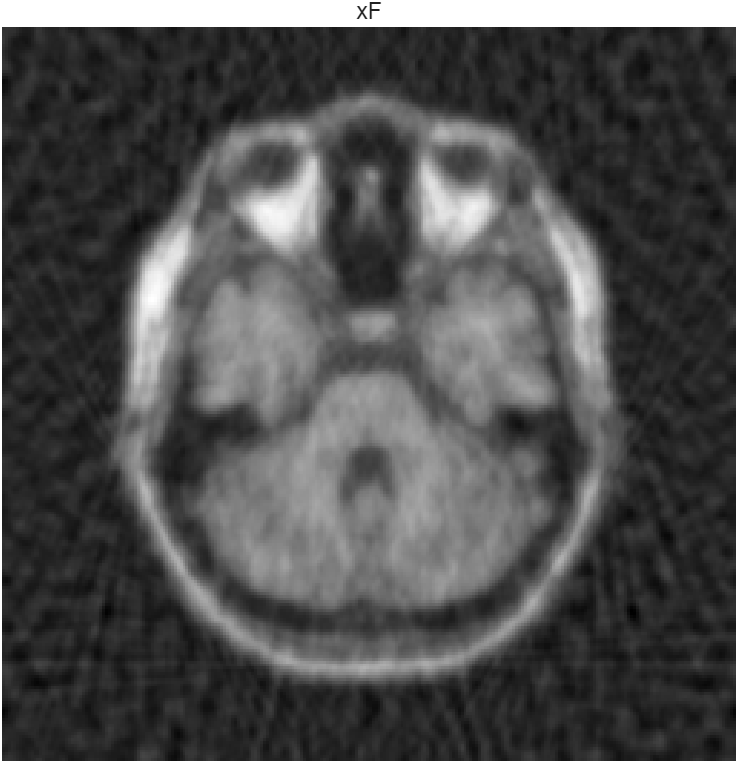}&
  \includegraphics[width=2.5cm]{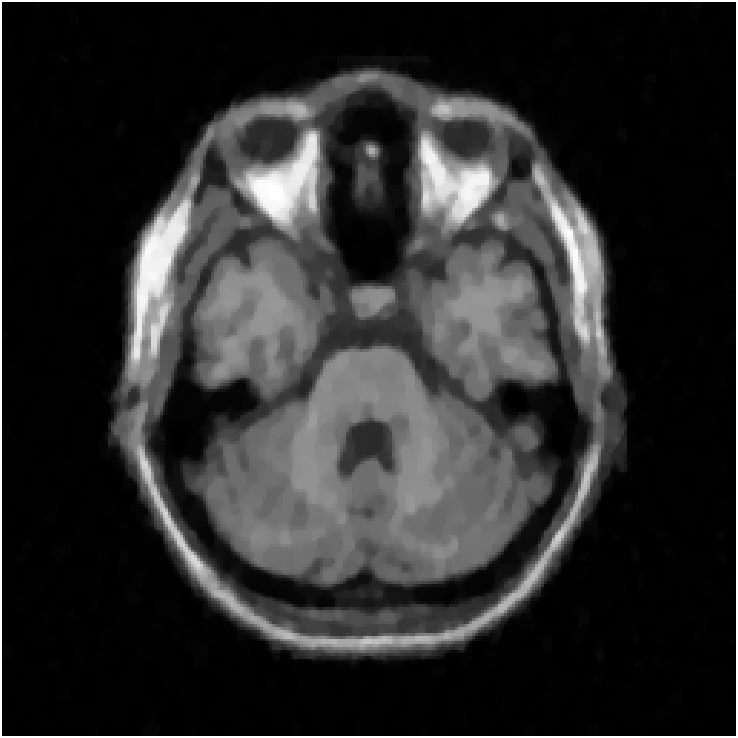}&
  \includegraphics[width=2.5cm]{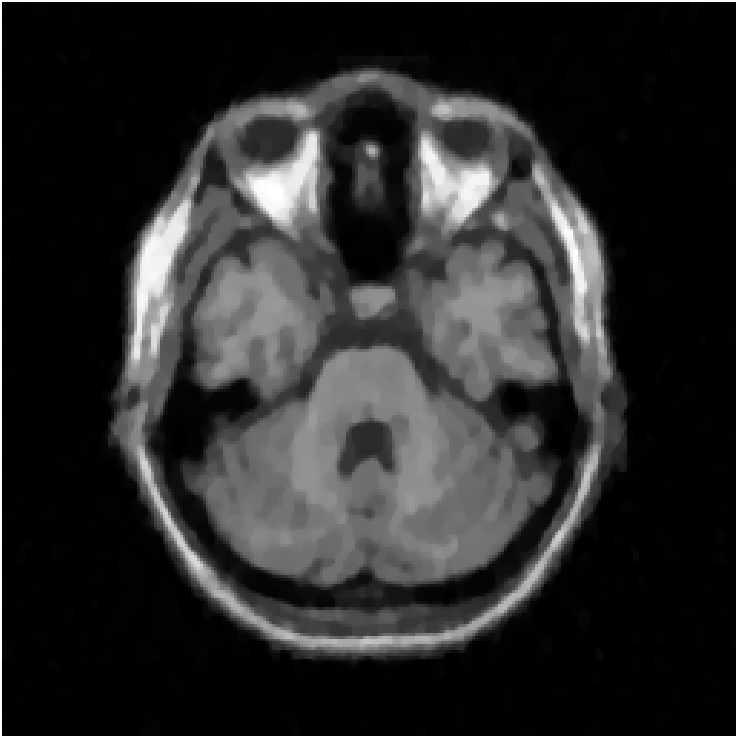}&
  \includegraphics[width=2.5cm]{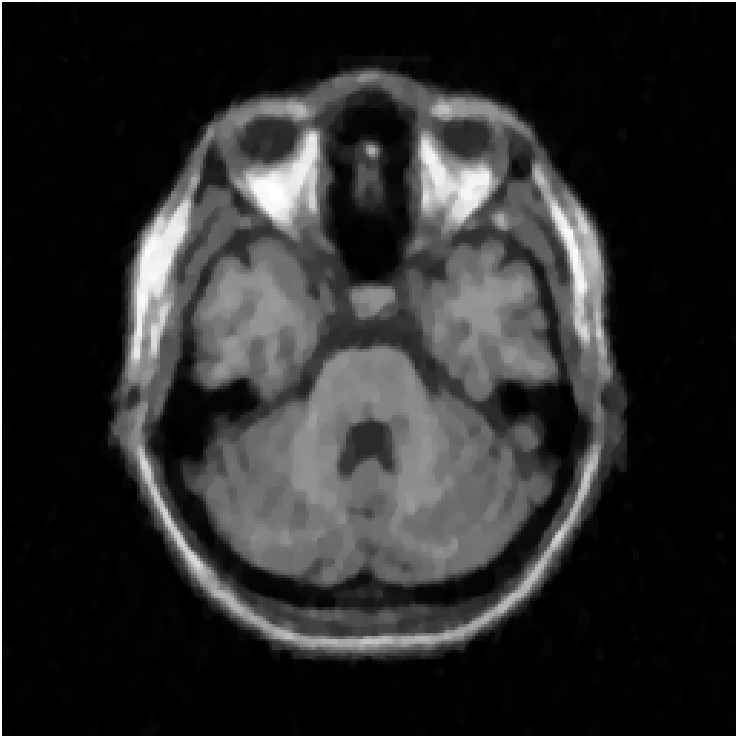}&
  \includegraphics[width=2.5cm]{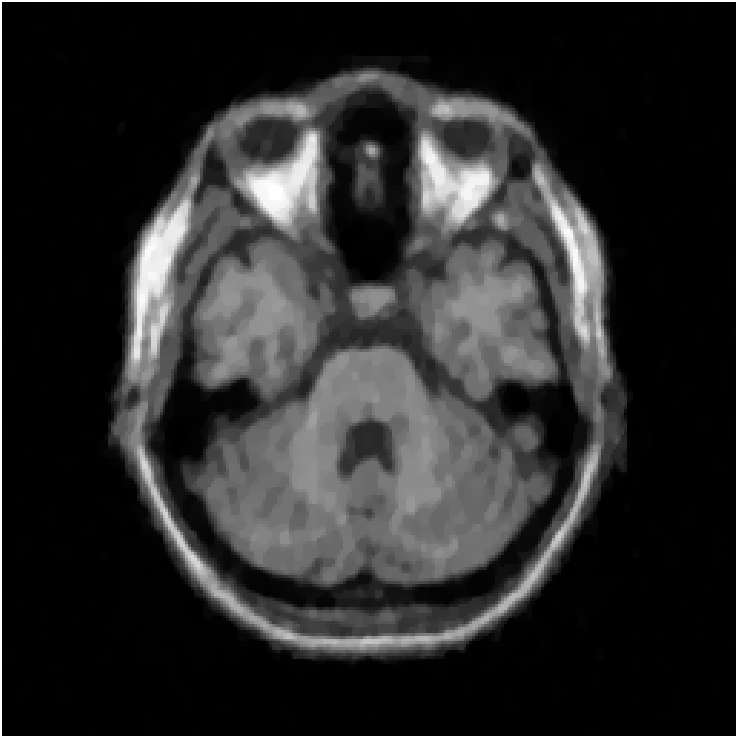}\\
  ZP& TV &$L_1-\alpha L_2$&TTV&PSV$_{a,p}$
  \end{tabular}
  \caption{Reconstructed images from noise-free 
  measurements of Brains A, B, and C via different methods under 40 radial lines.}\label{fig-noiselessbrain}
\end{figure}

\begin{table}[H]
\centering
\resizebox{0.95\textwidth}{!}{
\renewcommand{\arraystretch}{1.2}
\begin{tabular}{c|c|ccc|ccc|ccc}
\hline
\multirow{2}{*}{Measure}& Image &\multicolumn{3}{c|}{Brain A} & \multicolumn{3}{c|}{Brain B}
& \multicolumn{3}{c}{Brain C} \\
\cline{2-5} \cline{6-8} \cline{9-11}
& Sampling lines & 24 &  32 &  40&  24 &  32 &  40 &  24 &  32 &  40 \\
\hline
\multirow{5}{*}{PSNR}
& ZP & 12.5285 & 12.8657 &12.9260 &12.2420 & 12.3029& 13.6014& 13.8908&17.4919 & 17.3458 \\
& TV & 25.0992 & 30.7268 &32.4462 &25.1420&23.7661& 29.2643& 27.2735&30.0685 & 32.4772\\
& $L_1-\alpha L_2$ &25.1339&30.8424&32.3885&24.9523 &23.7980&29.5247&27.3430&30.1455&32.5819\\
& TTV & 25.1230 & 30.7466& 32.3775& 25.1598  & 24.7709 & 29.5134& 27.3468 & 30.1675& 32.5448\\
& PSV$_{a,p}$&25.4421& 30.7839 &33.3895&25.1496& 24.0318
& 29.1987& 27.3716& 29.9246& 32.1305\\
\hline
\multirow{5}{*}{SSIM}
& ZP& 0.2616, & 0.2783 & 0.3046& 0.2640 &0.2983 & 0.3319& 0.2493& 0.3080 & 0.3279\\
& TV & 0.7003 & 0.7801&0.8628& 0.7104& 0.8032& 0.8682 & 0.7109& 0.8262&0.8760\\
& $L_1-\alpha L_2$ &0.7085&0.7903&0.8810&0.7132&0.7822&0.8742&0.7176&0.8365&0.8908\\
& TTV & 0.6732&0.7560&0.8454&0.7237&0.7287&0.8528&0.7009&0.8197&0.8714\\
& PSV$_{a,p}$&\textbf{0.7603}&\textbf{0.8434}&\textbf{0.9272}&\textbf{0.8015}&\textbf{0.8505}
&\textbf{0.9283}&\textbf{0.7628}&\textbf{0.8839}&\textbf{0.9235}\\
\hline
\multirow{5}{*}{GMSD}
& ZP & 0.3421 & 0.3555& 0.3690 & 0.3503&0.3626&0.3746& 0.3283& 0.3347 & 0.3377\\
& TV& 0.2551 & 0.2119&0.1598 & 0.2388&0.2081& 0.1730 & 0.2430& 0.1980&0.1712\\
& $L_1-\alpha L_2$  &0.2500&0.2029&0.1536&0.2403&0.2099&0.1756&0.2416& 0.1986&0.1707\\
& TTV&0.2548&0.2092&0.1587&0.2412&0.2205&0.1807&0.2492&0.2036&0.1760\\
& PSV$_{a,p}$& \textbf{0.2466}&\textbf{0.2017}&0.1543&\textbf{0.2062}&\textbf{0.2014}&
\textbf{0.1675}&\textbf{0.2333}&\textbf{0.1911}&\textbf{0.1682}\\
\hline
\end{tabular}}
\caption{Comparison of quality metrics for reconstructed
images from noise-free measurements of Brains A, B, and C
via different methods under 24, 32, and 40 radial
sampling lines.}\label{table-noiselessbrain}
\end{table}

We also conduct comparative experiments on Brains A, B, and C,
respectively, using noise-free and
Gaussian-noise-corrupted (about $5\%$ noise level) measurements.
For satisfactory performance, we use the following parameter settings:
$\alpha=0.5$ for the $L_1-\alpha L_2$, $a=5$ for the TTV, and $p=0.5$ for the proposed PSV$_{a,p}$.
The values of the PSNR, the SSIM, and the GMSD for each reconstructed image are compared
in Tables \ref{table-noiselessbrain} and \ref{table-Gaussbrain},
while only the reconstructed images from noise-free and Gaussian-noise-corrupted measurements
with 40 radial lines  are shown in Figures \ref{fig-noiselessbrain} and \ref{fig-Gaussbrain}.
The total runtime per method and the corresponding hyperparameter settings
for all experiments recorded in Tables \ref{table-noiselessbrain} and \ref{table-Gaussbrain}
are respectively summarized  in Tables \ref{table-time2} and \ref{table-hyperparameter}.
As shown in Tables \ref{table-noiselessbrain},  \ref{table-Gaussbrain},   and \ref{table-time2},
although the ZP reconstruction is near-instantaneous,
the quality of the ZP-reconstructed images is clearly inferior.
Compared with  other methods, the proposed
PSV$_{a,p}$ maintains a competitive PSNR,
simultaneously improves the SSIM value from $3.6\%$ to $12.5\%$,
and, except one, yields the lowest GMSD.
Moreover,  PSV$_{a,p}$ achieves a desirable trade-off
between quality and runtime because it achieves competitive PSNR
and GMSD and a higher SSIM
while requiring no more than twice the runtime of the TTV and being faster than
the $L_1-\alpha L_2$. The zoom-in view of the regions marked by the red box in Figure \ref{fig-noiselessbrain}
is presented in Figure \ref{fig-noiselessbrainred}.

\begin{table}[H]
\centering
\resizebox{0.9\textwidth}{!}{\renewcommand{\arraystretch}{1.2}
\begin{tabular}{c|c|ccc|ccc|ccc}
\hline
\multirow{2}{*}{Measure}& Image &\multicolumn{3}{c|}{Brain A} & \multicolumn{3}{c|}{Brain B}
& \multicolumn{3}{c}{Brain C} \\
\cline{2-5} \cline{6-8} \cline{9-11}
& Sampling lines & 24 &  32 &  40&  24 &  32 &  40 &  24 &  32 &  40 \\
\hline
\multirow{5}{*}{PSNR}
& ZP & 12.3972 & 13.0221 &12.8583 &12.3369& 12.0827& 13.3159& 13.8524&17.7575& 16.8885\\
& TV & 24.4161& 26.9293 &31.3912&25.6834&26.0189 & 28.0949& 26.8649&29.1342& 30.6481\\
& $L_1-\alpha L_2$ &24.6632&27.0344&31.6954&25.7352&25.9938&28.3102&26.9528&29.3274&30.9075\\
& TTV &24.3133& 27.3917&31.2165&25.5070&26.0310&27.5372& 26.8200&29.0367&30.5254\\
& PSV$_{a,p}$& 25.1538&28.3978&30.0065& 25.5975& 26.5164
& 28.5281 & 26.9674 &29.0334&30.6600\\
\hline
\multirow{5}{*}{SSIM}
& ZP & 0.2576, &0.2757 & 0.2996& 0.2613 &0.2935& 0.3247& 0.2463& 0.3060 & 0.3210\\
& TV & 0.6654 & 0.7278&0.7476& 0.6859& 0.7212& 0.7550 & 0.6883& 0.7591&0.7712\\
& $L_1-\alpha L_2$ &0.6732&0.7295&0.7686&0.7077&0.7327&0.7655&0.6918&0.7665&0.7836\\
& TTV & 0.6308&0.7144&0.7004&0.6954&0.6937&0.7166& 0.6733& 0.7343&0.7384\\
& PSV$_{a,p}$& \textbf{0.7252} &\textbf{0.8082} &\textbf{0.8101} & \textbf{0.7620} & \textbf{0.7826}
&\textbf{0.8613}&\textbf{0.7383} & \textbf{0.8310}& \textbf{0.8480}\\
\hline
\multirow{5}{*}{GMSD}
& ZP & 0.3414&0.3561& 0.3691 & 0.3489&0.3639&0.3748 & 0.3289& 0.3356 & 0.3395 \\
& TV& 0.2771&0.2457&0.2361 & 0.2584&0.2469& 0.2361  &0.2597& 0.2314&0.2276\\
& $L_1-\alpha L_2$  &0.2673&0.2233&0.1968&0.2540&0.2398&0.2216&0.2544&0.2232&0.2132\\
& TTV&0.2823 &0.2596&0.2425 &0.2605&0.2597&0.2523&0.2682& 0.2428&0.2417\\
& PSV$_{a,p}$&\textbf{0.2475}& \textbf{0.2143}& \textbf{0.1941}&\textbf{0.2136}& \textbf{0.2189}&\textbf{0.1652}
 &\textbf{0.2405} &\textbf{0.1923}&\textbf{0.1821}\\
\hline
\end{tabular}}
\caption{Comparison of quality metrics for reconstructed images from Gaussian-noise-corrupted measurements of Brains A, B, and C
 via different methods, under 24, 32, and 40 radial sampling lines.}\label{table-Gaussbrain}
\end{table}

\begin{figure}[t]
  \centering
  \begin{tabular}{ccccc}
  \includegraphics[width=2.5cm]{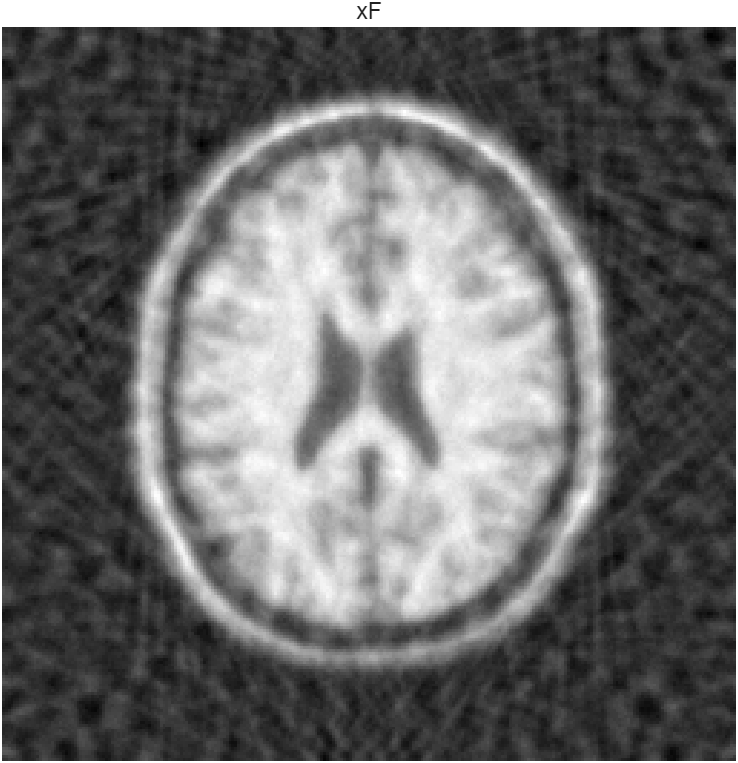}&
  \includegraphics[width=2.5cm]{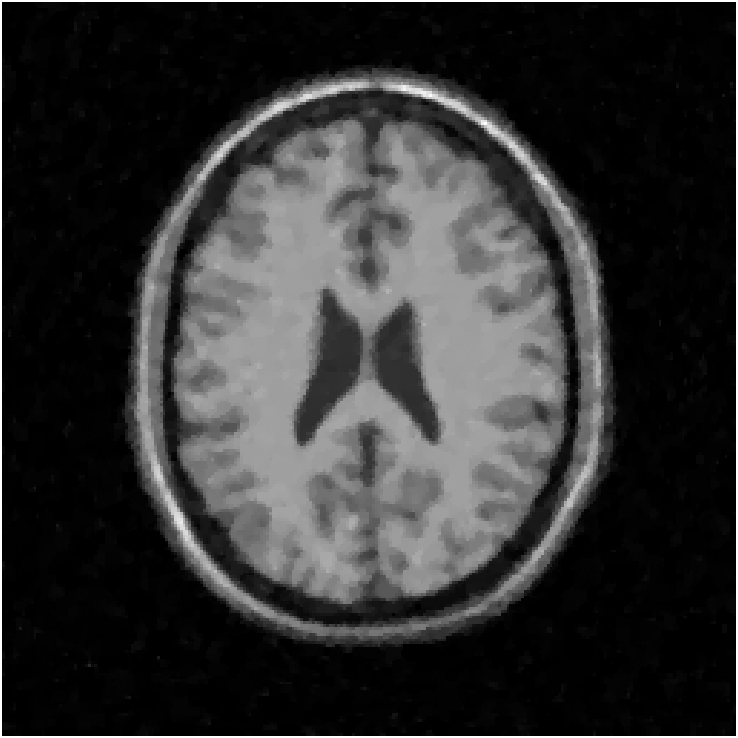}&
  \includegraphics[width=2.5cm]{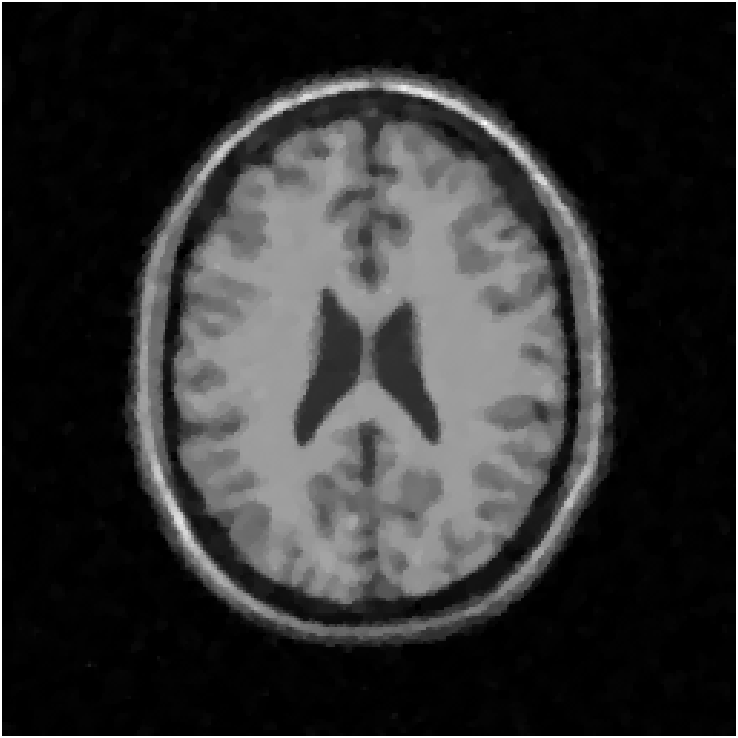}&
  \includegraphics[width=2.5cm]{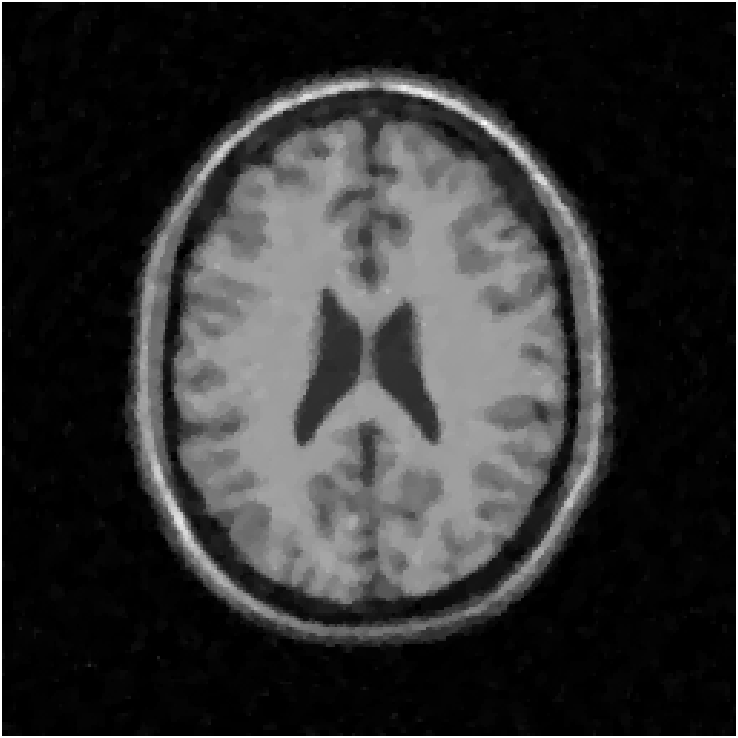}&
  \includegraphics[width=2.5cm]{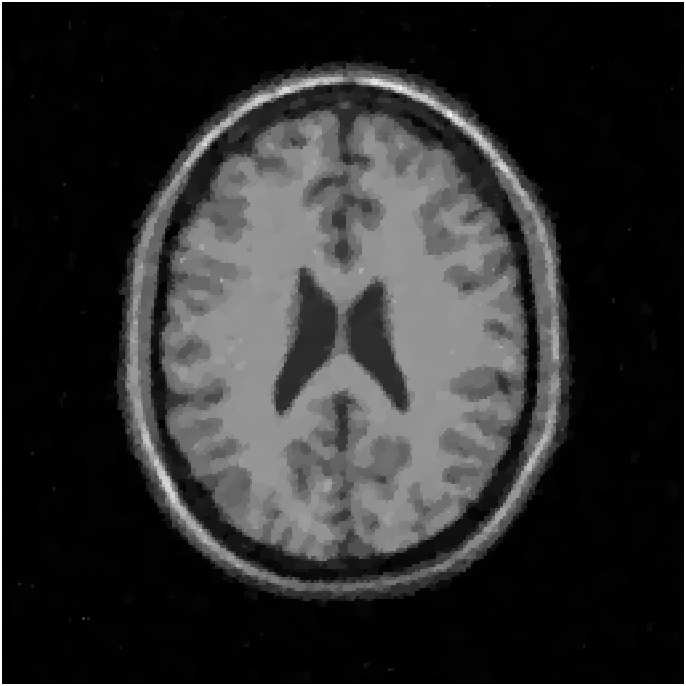}\\
  \includegraphics[width=2.5cm]{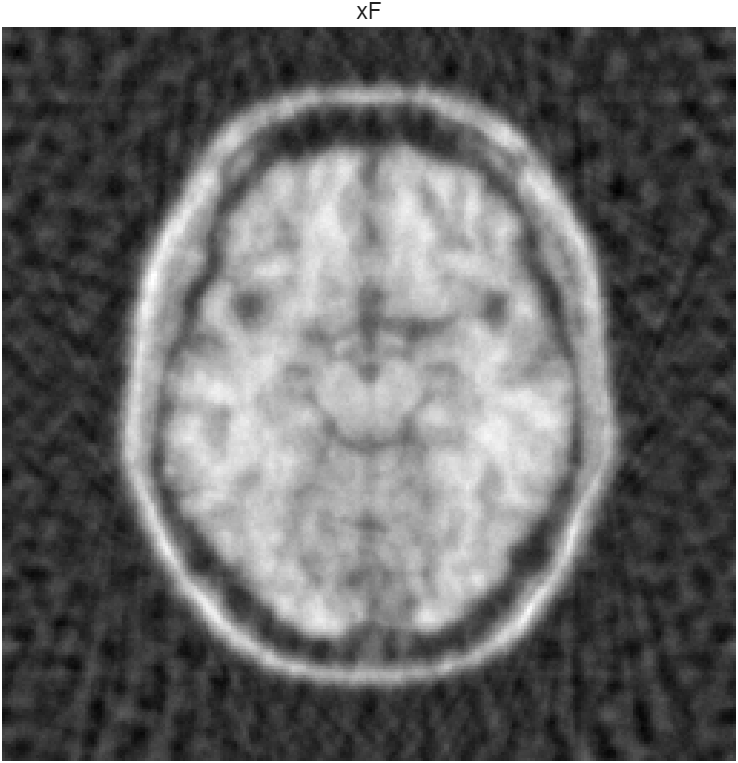}&
  \includegraphics[width=2.5cm]{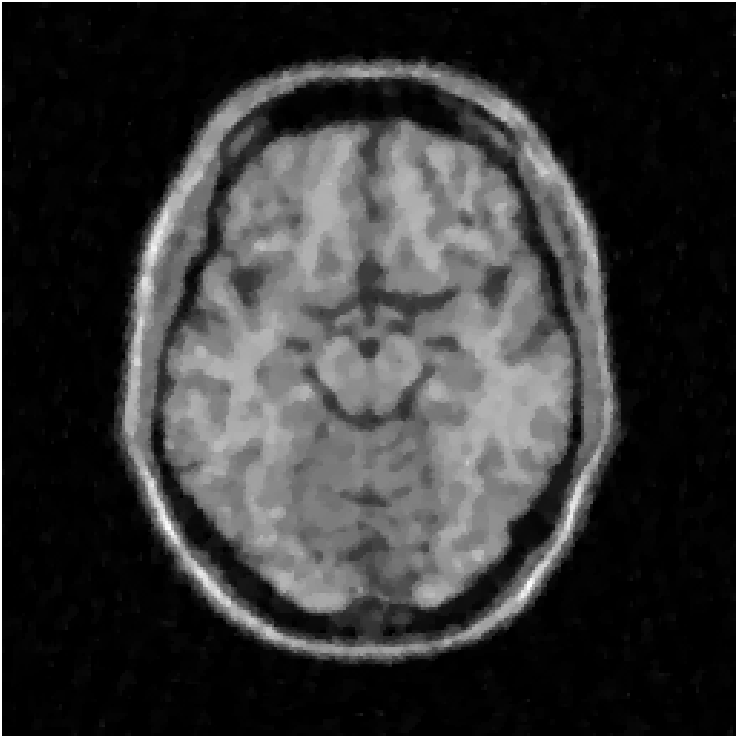}&
  \includegraphics[width=2.5cm]{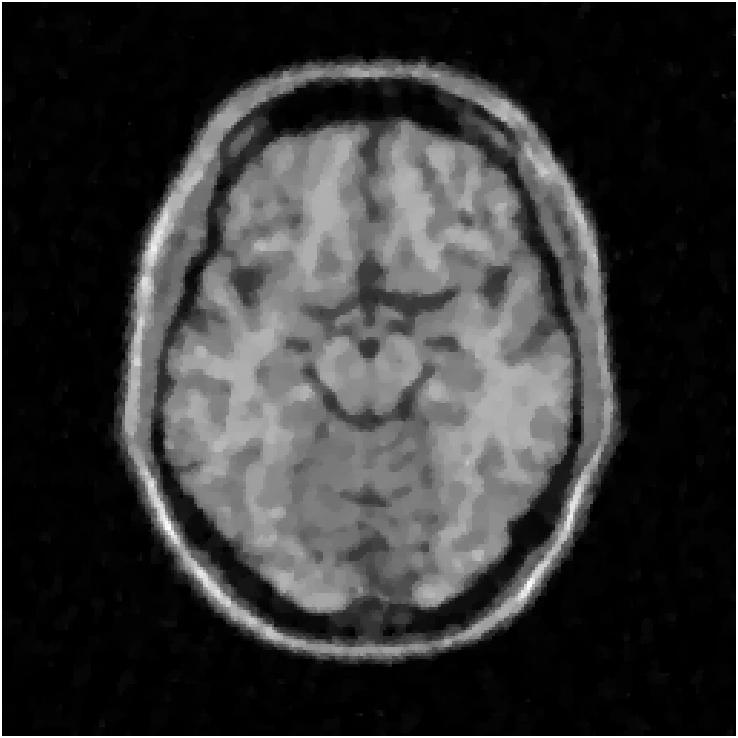}&
  \includegraphics[width=2.5cm]{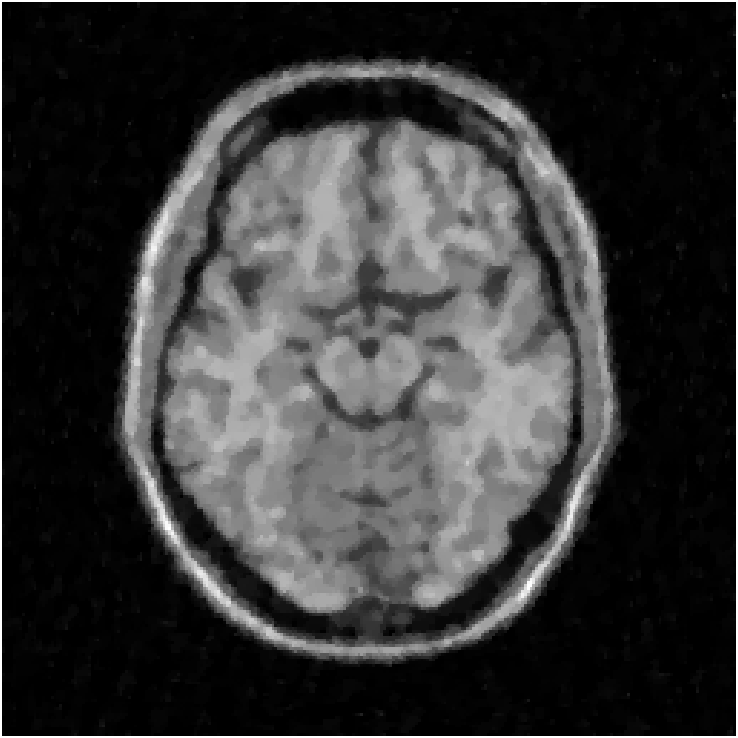}&
  \includegraphics[width=2.5cm]{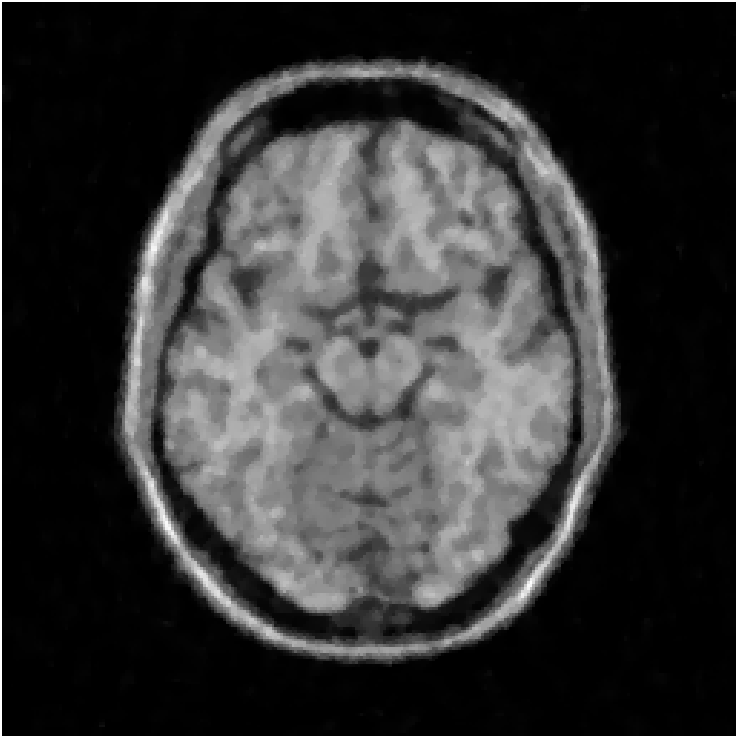}\\
  \includegraphics[width=2.5cm]{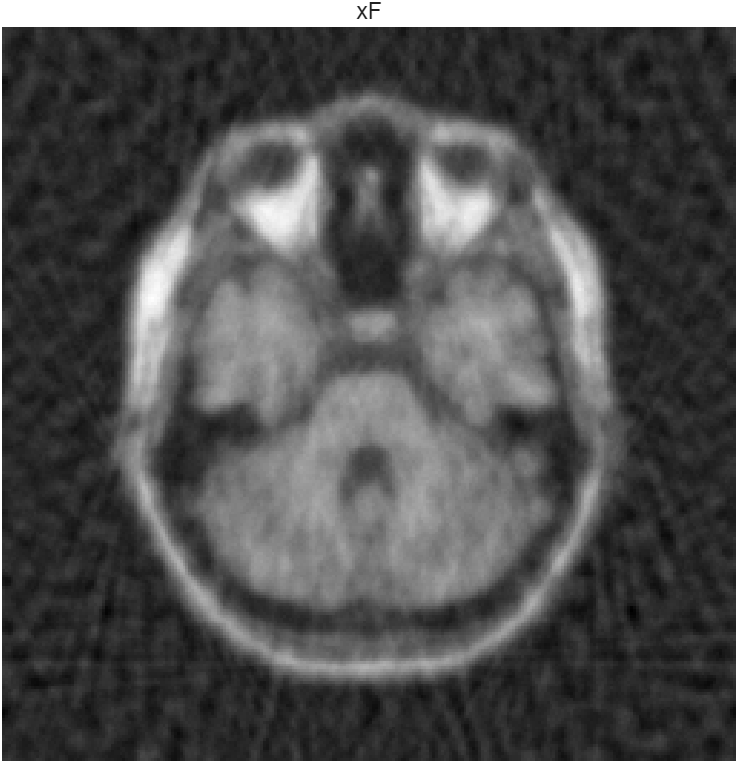}&
  \includegraphics[width=2.5cm]{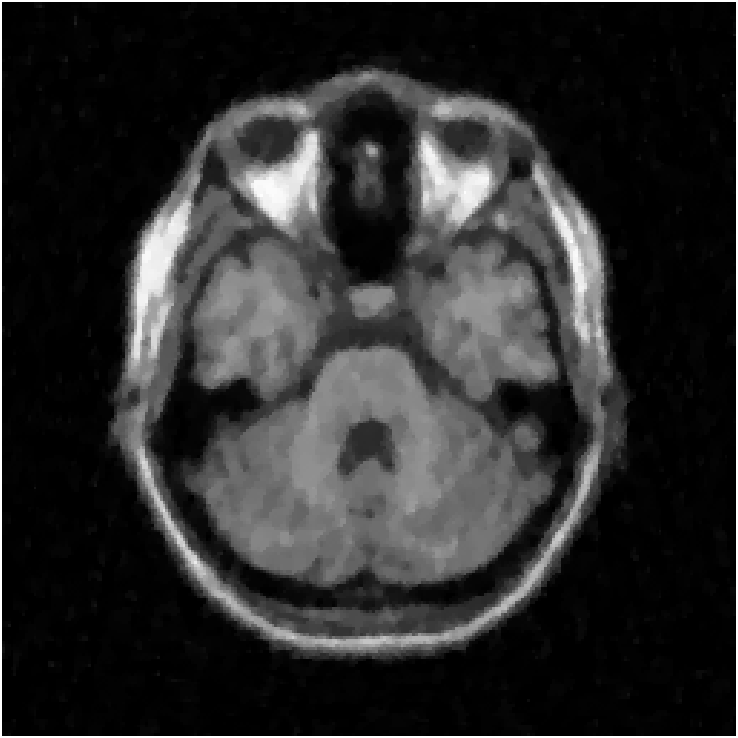}&
  \includegraphics[width=2.5cm]{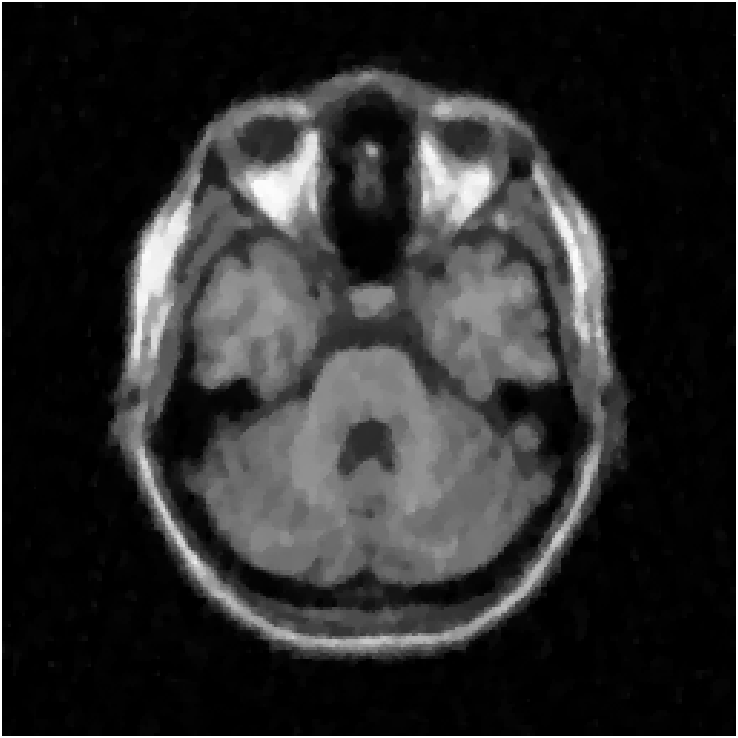}&
  \includegraphics[width=2.5cm]{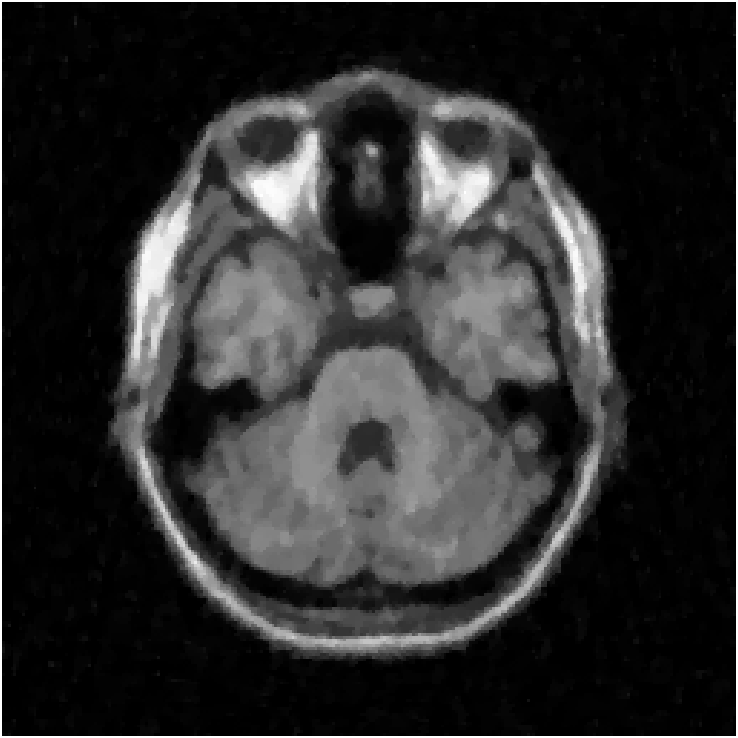}&
  \includegraphics[width=2.5cm]{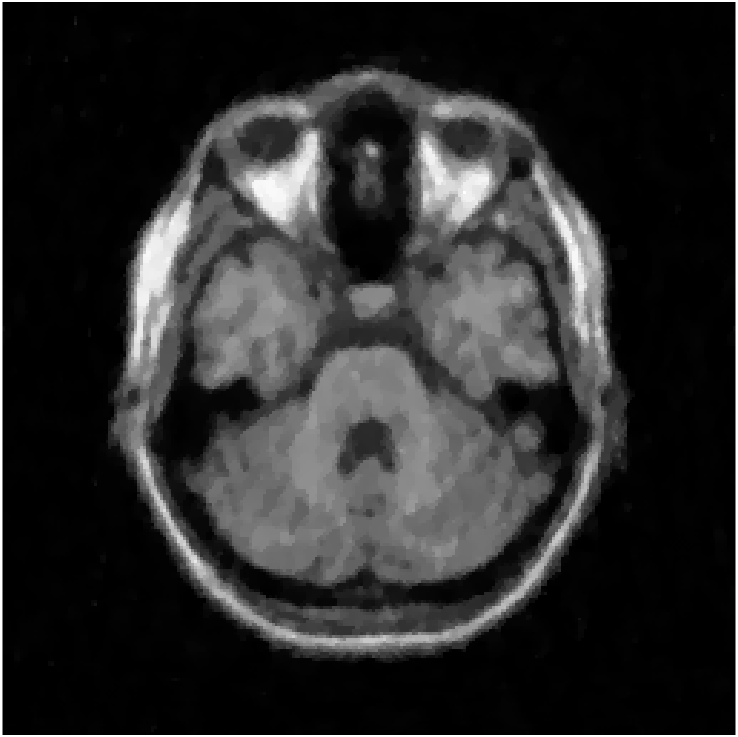}\\
  ZP & TV &$L_1-\alpha L_2$&TTV& PSV$_{a,p}$
  \end{tabular}
  \caption{Reconstructed images from Gaussian-noise-corrupted
  measurements of Brains A, B, and C via different methods,
  under 40 radial lines.}\label{fig-Gaussbrain}
\end{figure}

\begin{table}[H]
\centering
\renewcommand{\arraystretch}{1.2}
  \begin{tabular}{cccccc}
  \hline
  Method&ZP&TV&$L_1-\alpha L_2$&TTV& PSV$_{a,p}$\\
  \hline
  Noiseless& 0.0077&3.7358& 42.8243& 6.9970& 13.4143\\
  \hline
  Gaussian noise& 0.0097& 4.0066& 43.2080& 7.0346 & 10.9737 \\
  \hline
  \end{tabular}
  \caption{Total runtime  (in seconds) per method for all experiments listed in
  Tables \ref{table-noiselessbrain} and \ref{table-Gaussbrain}.}\label{table-time2}
 \end{table}

\begin{figure}[H]
  \centering
  \begin{tabular}{cccc}
  \includegraphics[width=3.2cm]{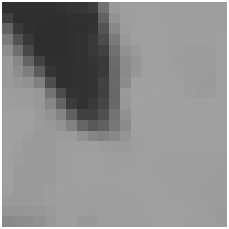}&
  \includegraphics[width=3.2cm]{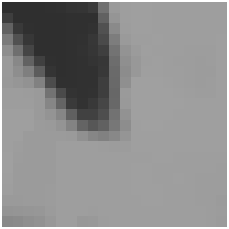}&
  \includegraphics[width=3.2cm]{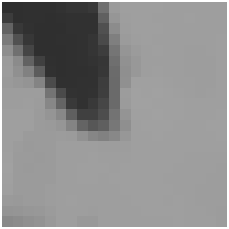}&
  \includegraphics[width=3.2cm]{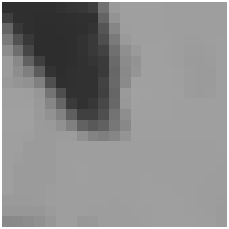}\\
TV &$L_1-\alpha L_2$&TTV& PSV$_{a,p}$
  \end{tabular}
  \caption{Zoom-in view of the regions marked by the 
  red box in Figure \ref{fig-noiselessbrain}.}
  \label{fig-noiselessbrainred}
\end{figure}

\begin{table}[H]
\centering
\resizebox{1\textwidth}{!}{
\renewcommand{\arraystretch}{1.2}
\begin{tabular}{c|c|c|ccc|ccc|ccc}
\hline
 \multicolumn{3}{c|}{Image} &\multicolumn{3}{c|}{Brain A} & \multicolumn{3}{c|}{Brain B}
& \multicolumn{3}{c}{Brain C} \\
\hline
 \multicolumn{3}{c|}{Sampling lines} & 24 &  32 &  40&  24 &  32 &  40 &  24 &  32 &  40 \\
\hline
 \multirow{5}{*}{Noise-free}&TV& $\lambda$ & 1.0e-3& 1.0e-3& 1.0e-3
 &1.0e-3& 4.0e-4& 1.0e-3
 & 1.0e-3& 1.0e-3& 1.0e-3\\
\cline{2-12}
&$L_1-\alpha L_2$& $\lambda$ & 1.0e-3& 1.0e-3& 1.0e-3
 &1.0e-4& 1.0e-4& 1.0e-3
 & 1.0e-3& 1.0e-3& 9.0e-4\\
\cline{2-12}
& TTV& $\lambda$ & 8.0e-4 &  8.0e-4 &  8.0e-4
 &3.0e-4 &1.0e-3& 6.0e-4
 &6.0e-4 &5.0e-4 & 4.0e-4\\
\cline{2-12}
&\multirow{2}{*}{ PSV$_{a,p}$}& $a$& 0.4& 0.7& 1&  0.4&  0.5&  0.4&  0.4 & 0.4& 0.4\\
& & $\lambda$& 1.0e-5& 2.0e-5 &  2.0e-5&  2.5e-5
& 9.0e-6 &  1.0e-5& 1.5e-5& 1.0e-5& 9.0e-6\\
 \hline
\multirow{5}{*}{Gaussian noise}&TV& $\lambda$ & 1.0e-3& 1.0e-3& 1.0e-3
 &1.0e-3& 1.0e-3& 1.0e-3
 & 1.0e-3& 1.0e-3& 1.0e-3\\
\cline{2-12}
& $L_1-\alpha L_2$& $\lambda$ & 1.0e-3& 1.5e-3&1.5e-3
 &5.0e-4 & 8.0e-4& 1.0e-3& 1.0e-3& 1.0e-3& 1.0e-3\\
\cline{2-12}
 &TTV& $\lambda$ & 8.0e-4 & 3.0e-4 &8.0e-4
 &2.0e-4&5.0e-4&5.0e-4&5.0e-4 &5.0e-4&5.0e-4\\
\cline{2-12}
&\multirow{2}{*}{ PSV$_{a,p}$}& $a$& 0.6& 0.7& 1&  0.4&  0.5&  0.4&  0.4 & 0.4& 0.4\\
&& $\lambda$& 2.0e-5& 3.0e-5 & 2.0e-5
&  2.0e-5& 1.5e-5&  2.0e-5 & 1.5e-5 &  1.5e-5& 1.5e-5 \\
\hline
\end{tabular}}
\caption{Hyperparameter settings for all experiments listed in
Tables \ref{table-noiselessbrain} 
and \ref{table-Gaussbrain}.}\label{table-hyperparameter}
\end{table}

\subsection{CT Reconstruction}\label{subsec-CT}

In this subsection, we apply the proposed PSV$_{a,p}$ model to
the X-ray computed tomography (CT) reconstruction.
In the CT reconstruction, the operator $\Phi$ in \eqref{eq-unproblem}
represents the system matrix.
For numerical experiments,
the projection data are simulated from the standard Shepp-Logan phantom with
256 equally spaced parallel X-ray beams over a 180$^\circ$ scan with
$1^{\circ}$ angular increment; see Figure \ref{fig-SLp}.
In particular, we consider an undersampled  pattern, characterized by
a full scan with an angular gap between $60^{\circ}$ and $90^{\circ}$, which
results in a sampling rate of $82.78\%$.
Moreover, as shown in \cite{XYMZHW12}, the measurements obey a Poisson distribution;
therefore, we also consider the case of data corrupted by Poisson noise.
To be precise, we simulate a noise intensity with
the relative projection Poisson noise level of $0.5\%$ and an intensity of $1\%$.

\begin{figure}[H]
  \centering
  \begin{tabular}{ccccc}
  \includegraphics[width=3.25cm]{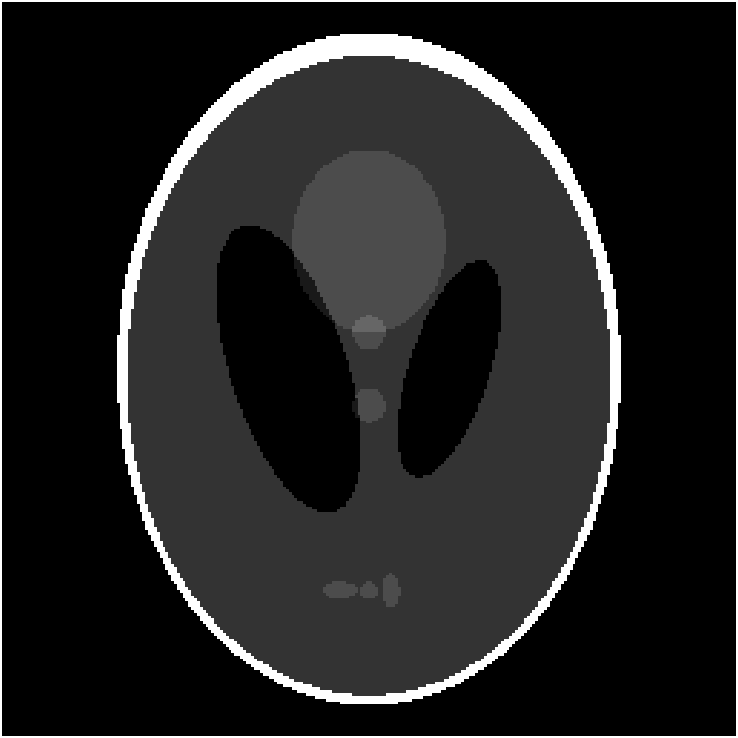}&
  \includegraphics[width=2.3cm]{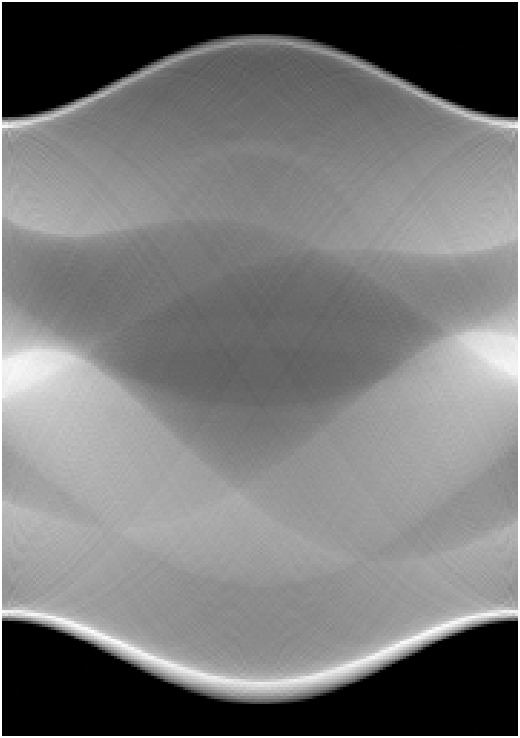}&
  \includegraphics[width=2.3cm]{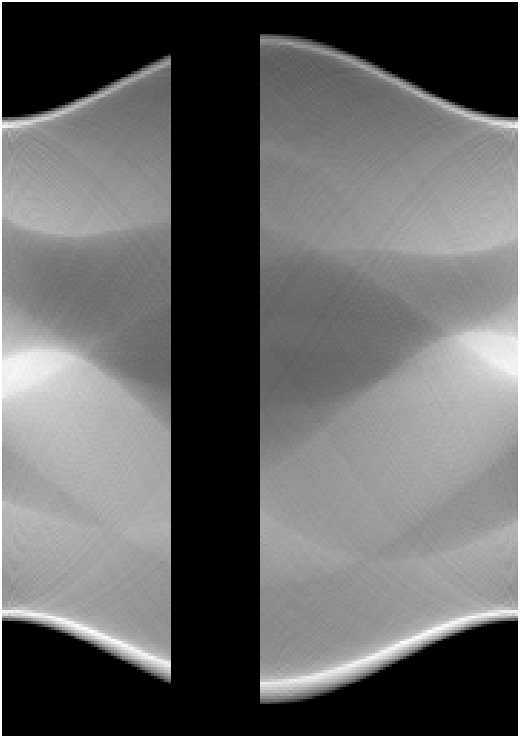}&
  \includegraphics[width=2.3cm]{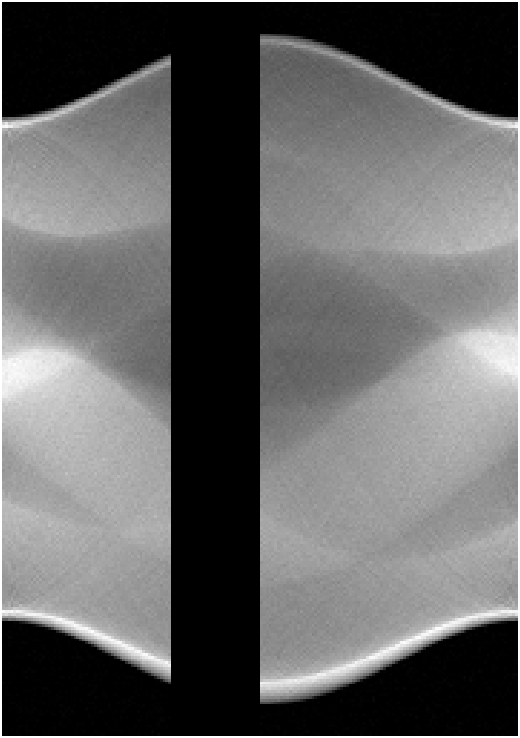}&
  \includegraphics[width=2.3cm]{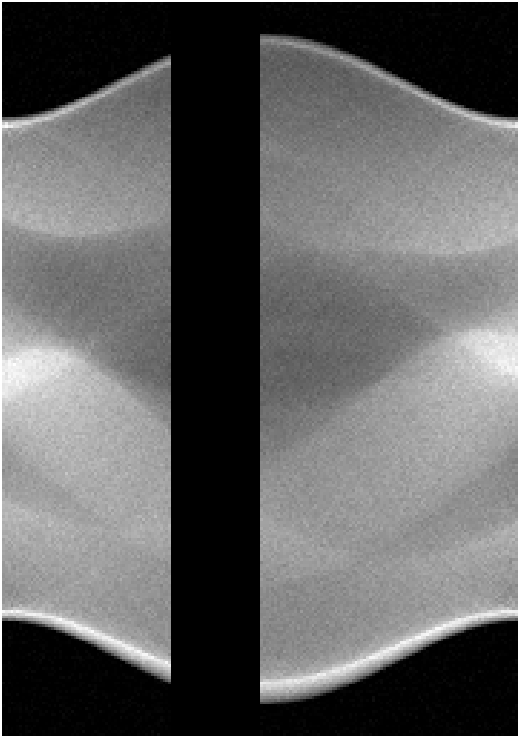}\\
(a) & (b) & (c) & (d) & (e)
  \end{tabular}
  \caption{(a) Shepp-Logan (SL) phantom; (b) Radon transform of the SL;
  (c) incomplete Radon transform of the SL with an angular gap $[60^\circ,90^\circ]$;
  (d) incomplete Radon transform of the SL with an angular gap
  $[60^\circ,90^\circ]$, corrupted by 0.5\% Poisson noise;
  (e) incomplete Radon transform of the SL with an angular
  gap $[60^\circ,90^\circ]$, corrupted by 1\% Poisson noise.}\label{fig-SLp}
\end{figure}

The stopping criteria of the IRLSPSV algorithm
are $k_{\mathrm{out}}=15$ and $\varepsilon_{\mathrm{out}}=10^{-7}$ for the outer loop,
$k_{\mathrm{mid}}=5$ and $\varepsilon_{\mathrm{mid}}=10^{-4}$ for the mid loop,
and $k_{\mathrm{inn}}=30$ and  $\varepsilon_{\mathrm{inn}}=5.0\times10^{-5}$ for the inner loop.
Furthermore, it is particularly noteworthy that the parameter $s$ in Algorithm \ref{alg1} is set to $\lfloor0.1N^2+0.5\rfloor$ here.

Taking the filtered back-projection (FBP) reconstruction as a baseline,
we conduct comparative experiments to evaluate the performance
between the PSV$_{a,p}$  against the TV, the $L_1-\alpha L_2$, and the TTV methods.
The quality metrics (the PSNR, the SSIM, and the GMSD) of each reconstructed image
and the total runtime  per method are compared in
Table \ref{table-CT}. Visually, reconstructed images are presented in Figure \ref{fig-CT}.

\begin{table}[H]
\centering
\resizebox{1\textwidth}{!}{
\renewcommand{\arraystretch}{1.3}
\begin{tabular}{c|cccc|cccc|cccc|c}
\hline
\multirow{2}{*}{Methods}
& \multicolumn{4}{c|}{Noise-free} & \multicolumn{4}{c|}{0.5\% Poisson noise}&
\multicolumn{4}{c|}{1\% Poisson noise}&\multirow{2}{*}{Time (s)}\\
\cline{2-13}
& $\lambda$& PSNR &SSIM &GMSD& $\lambda$& PSNR &SSIM &GMSD &
$\lambda$& PSNR &SSIM &GMSD& \\
\hline
FBP & -& 17.9964 &0.5721&0.3323& -& 17.6981&0.4651& 0.3323
& -& 17.3575&0.3089&0.3100&0.0768\\
\hline
TV & 3.0e-4& 31.5012 &0.9826 &0.1412& 7.0e-4& 28.9558&0.9666&0.1691
&1.2e-3& 27.1840 &0.9345&0.2053&77.5378\\
\hline
$L_1-\alpha L_2$ & 2.0e-4& 31.4075 &0.9825 &0.1424& 5.0e-4& 28.9405 &0.9667&0.1713
&8.0e-4& 27.1850&0.9341 &0.2117&91.1185\\
\hline
TTV& 1.0e-4& 31.2146&0.9817 &0.1461& 3.0e-4& 29.1654&0.9657 &0.1792
& 6.0e-4& 27.3072 &0.9328 &0.2113&76.9359 \\
\hline
PSV$_{a,p}$ & 1.2e-5& 30.5128&\textbf{0.9838} &\textbf{0.1146}
& 3.0e-5& 29.6134&\textbf{0.9739} &\textbf{0.1171}
& 5.5e-5& 25.7103&\textbf{0.9384} &\textbf{0.1556}& 80.1648\\
\hline
\end{tabular}}
\caption{Comparison of quality metrics for reconstructed
images across multiple scenarios  by different methods.}\label{table-CT}
\end{table}

\begin{figure}[H]
  \centering
  \begin{tabular}{ccccc}
  \includegraphics[width=2.5cm]{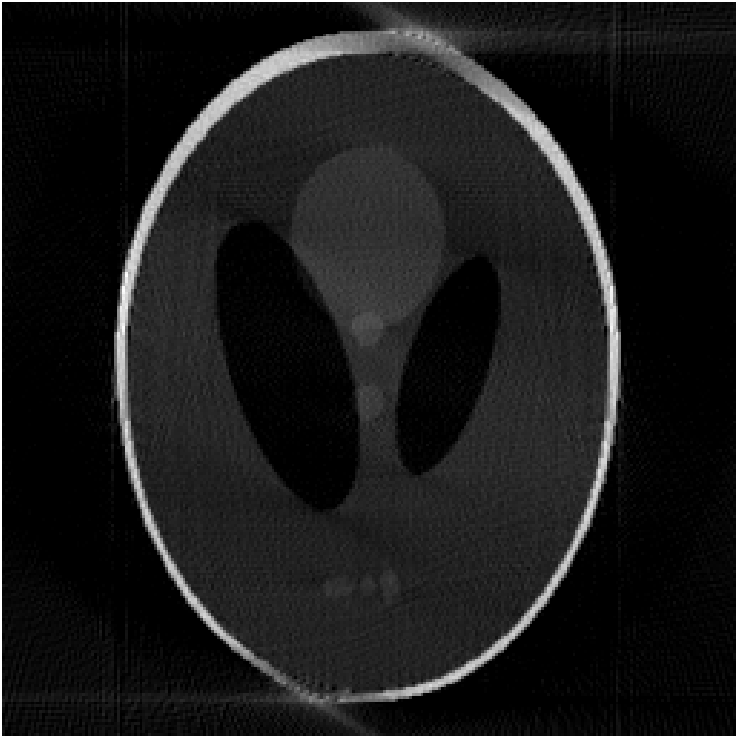}&
  \includegraphics[width=2.5cm]{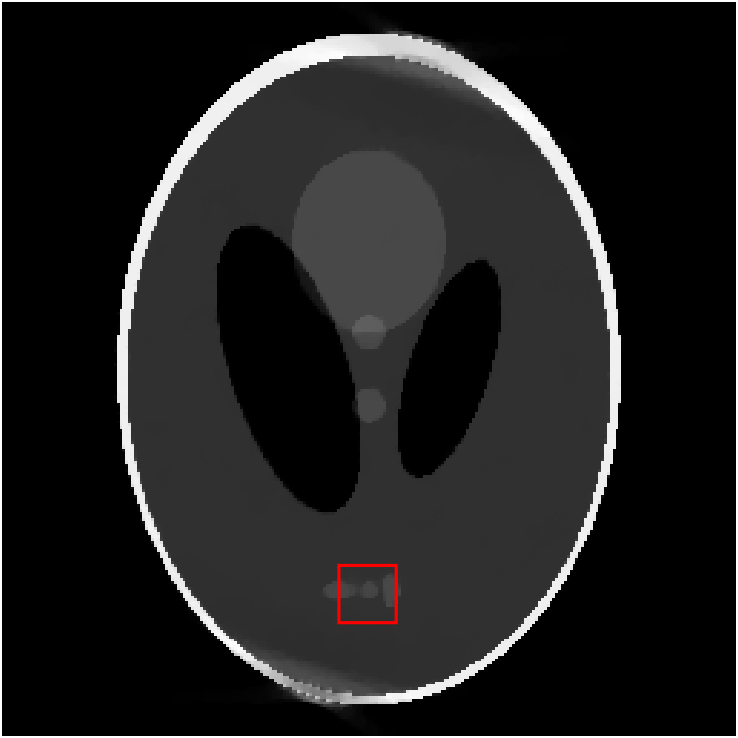}&
  \includegraphics[width=2.5cm]{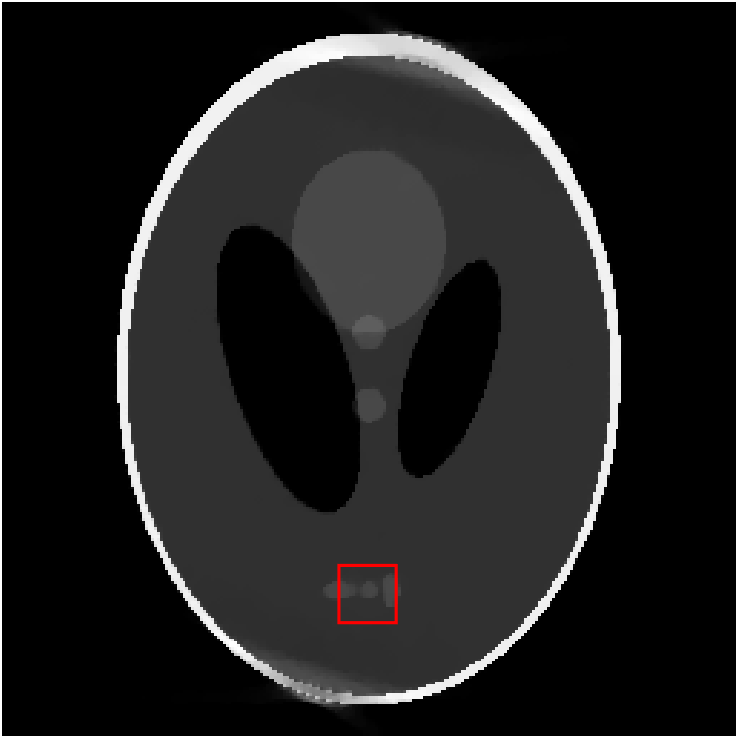}&
  \includegraphics[width=2.5cm]{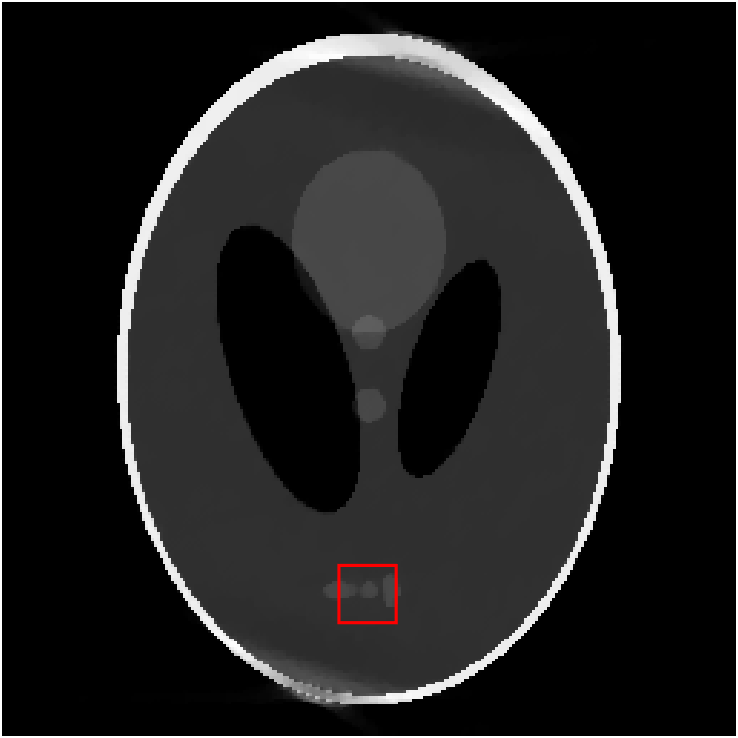}&
  \includegraphics[width=2.5cm]{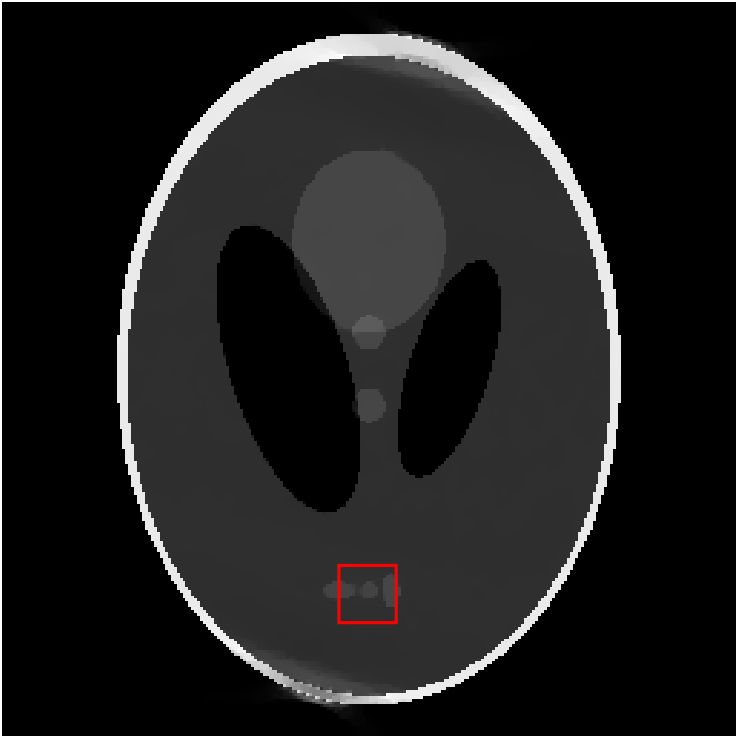}\\
  \includegraphics[width=2.5cm]{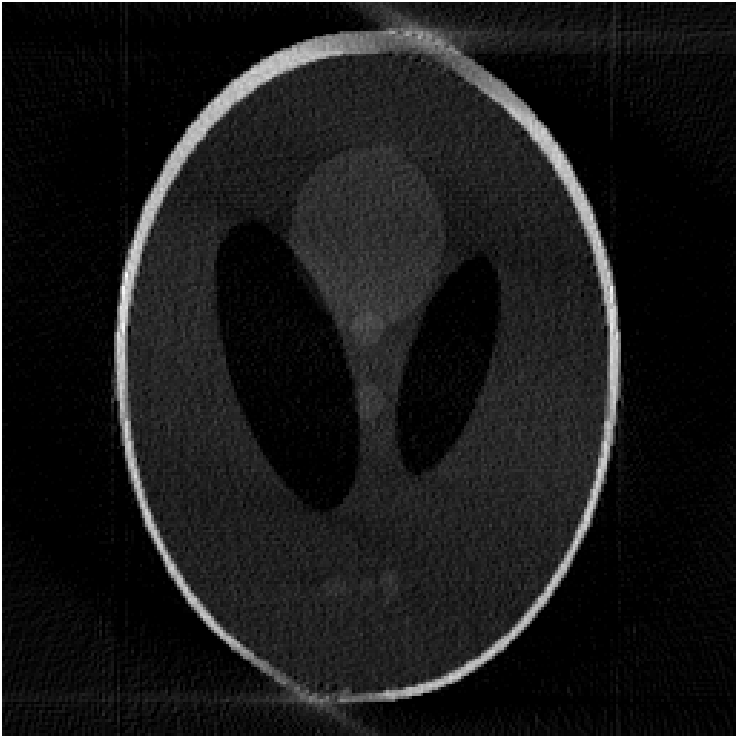}&
  \includegraphics[width=2.5cm]{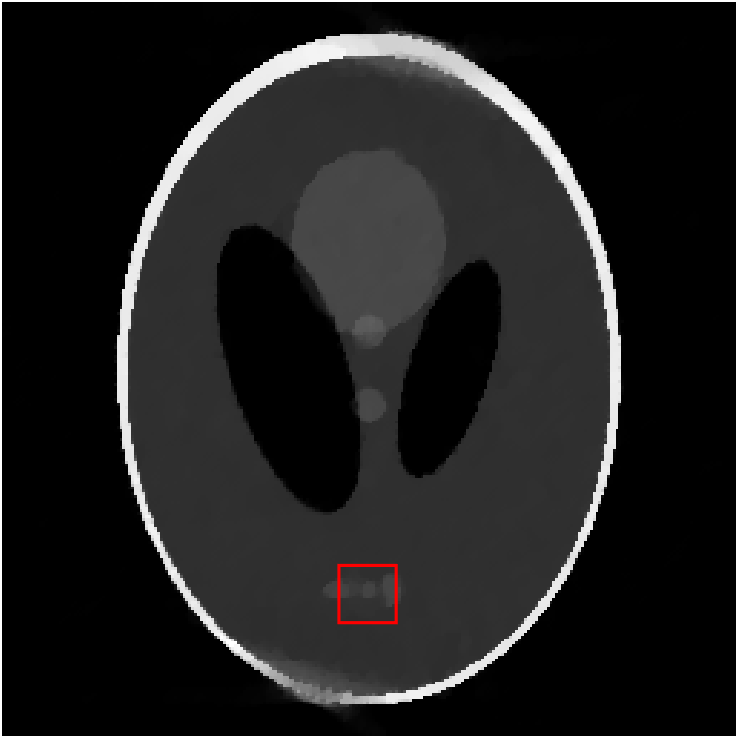}&
  \includegraphics[width=2.5cm]{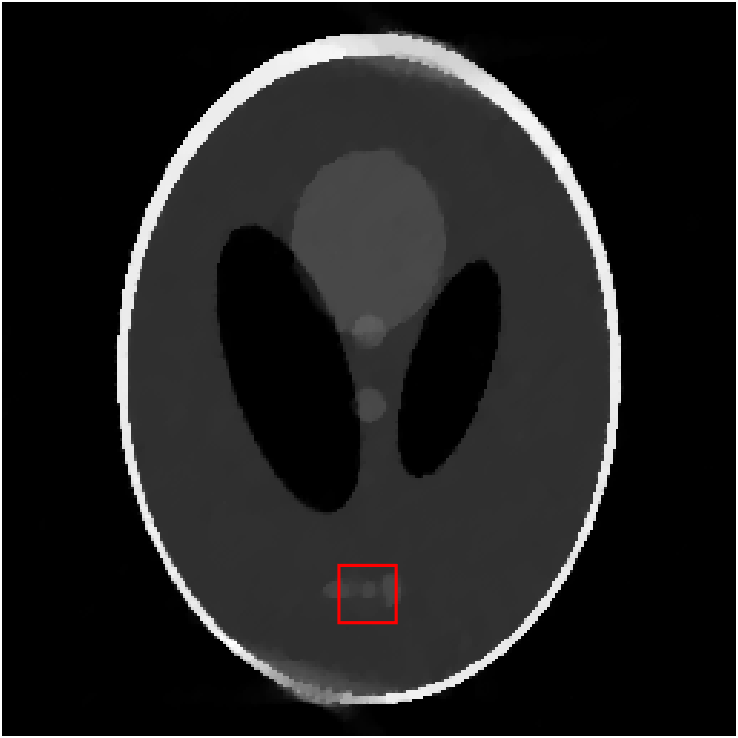}&
  \includegraphics[width=2.5cm]{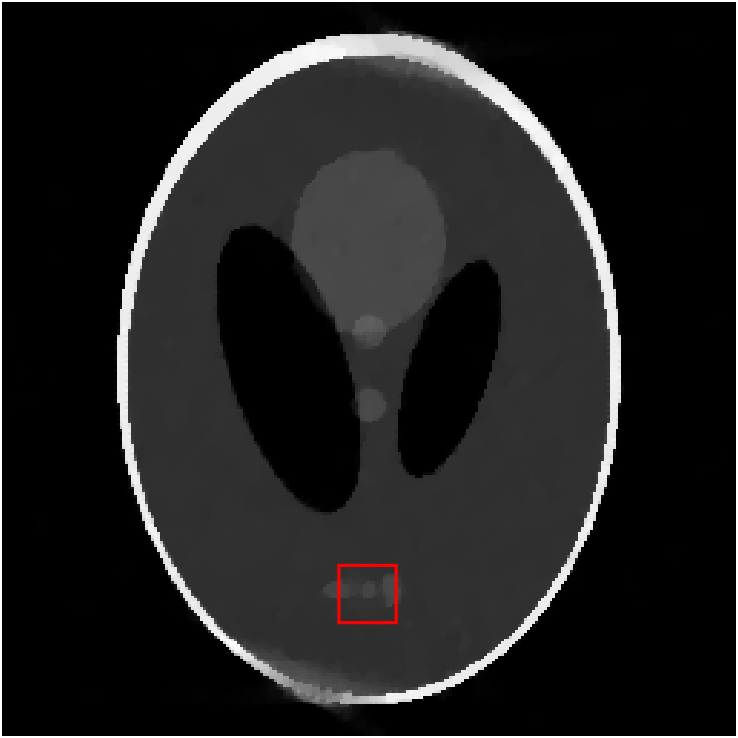}&
  \includegraphics[width=2.5cm]{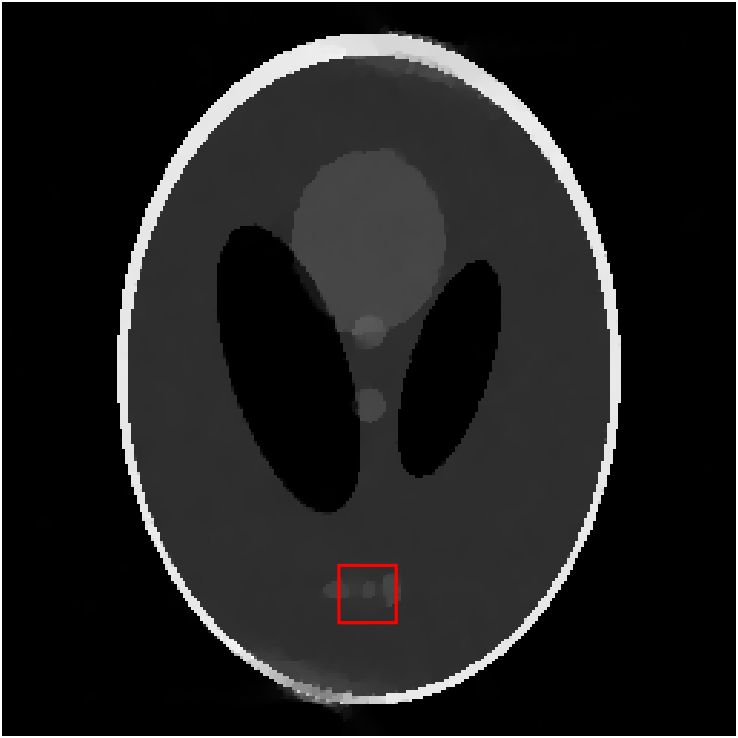}\\
  \includegraphics[width=2.5cm]{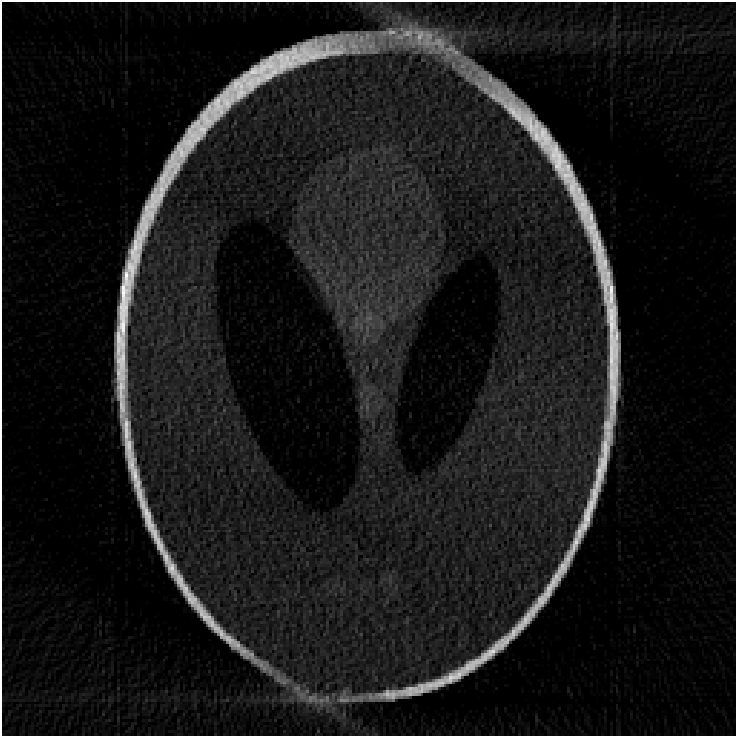}&
  \includegraphics[width=2.5cm]{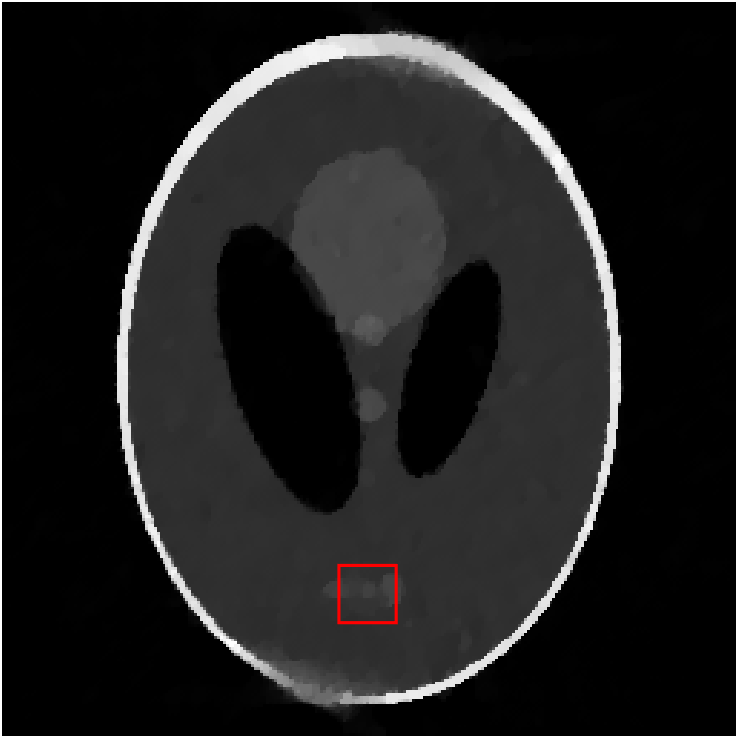}&
  \includegraphics[width=2.5cm]{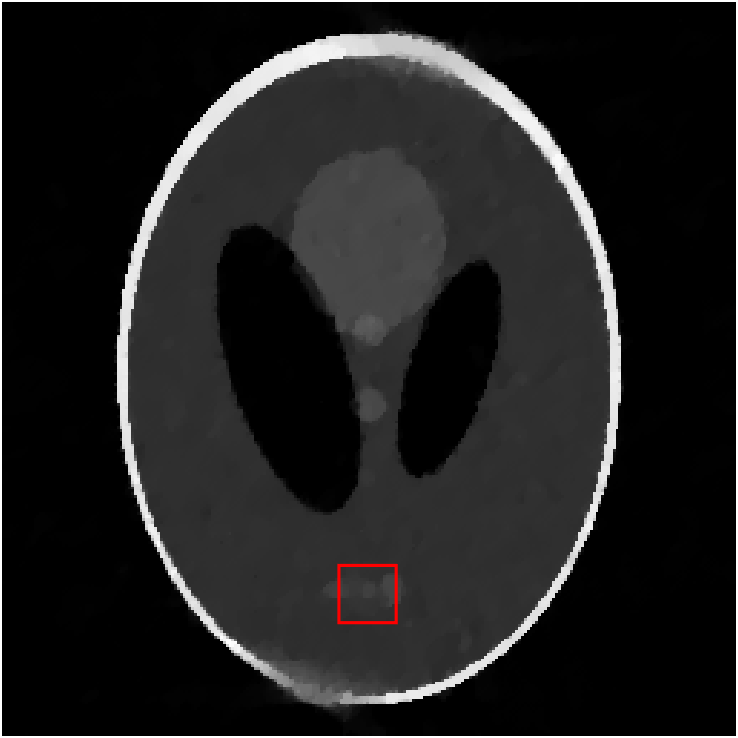}&
  \includegraphics[width=2.5cm]{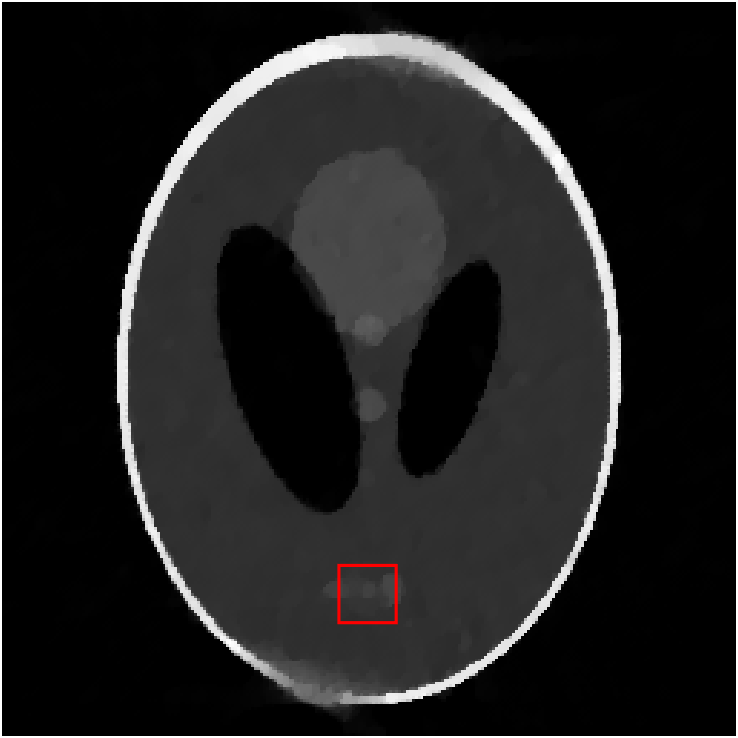}&
 \includegraphics[width=2.5cm]{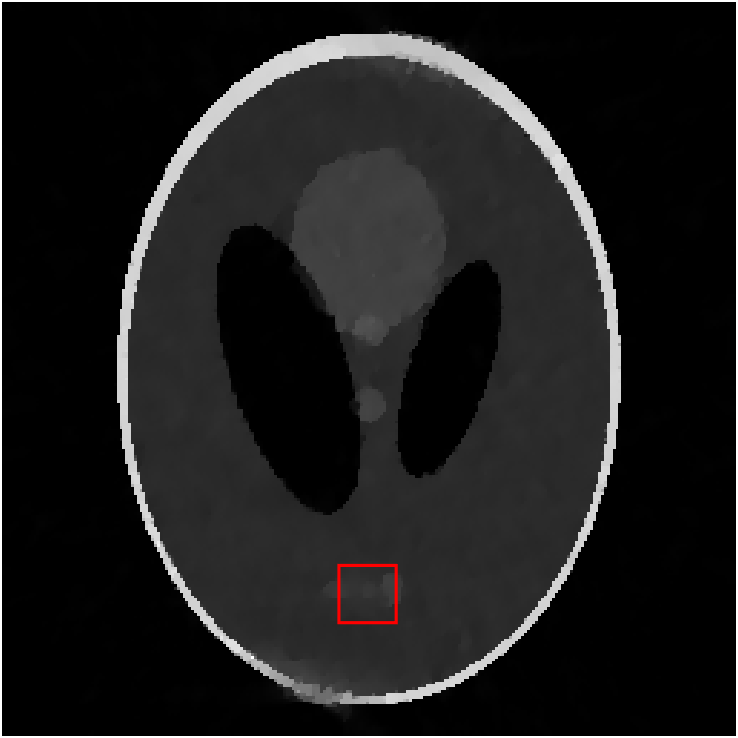}\\
FBP & TV &$L_1-\alpha L_2$&TTV& PSV$_{a,p}$
  \end{tabular}
  \caption{Reconstructed images across multiple scenarios  by different methods.
  From  top to bottom, images are reconstructed in noise-free,  $0.5\%$ Poisson noise, and $1\%$ Poisson noise scenarios.}\label{fig-CT}
\end{figure}

As shown in Table \ref{table-CT} and Figure \ref{fig-CT},
the FBP method fails to provide satisfactory  reconstructions
while all the other methods successfully​ enhance the quality metrics and
suppress the staircase artifacts. In particular, Table \ref{table-CT}
indicates that the proposed PSV$_{a,p}$
achieves the lowest GMSD among all comparative methods despite
a slight drop in the PSNR for certain experiments.
Furthermore, the zoom-in view of the red-box regions
of each reconstructed image by these comparative methods
is shown in Figure \ref{fig-CTred}.  As   illustrated,
the PSV$_{a,p}$ method further effectively suppresses artifacts in flat regions
while still remaining highly competitive in edge preservation.
Here, the parameters are set to $\alpha=0.5$ for the $L_1-\alpha L_2$,
$a=1$ for the TTV, and $a=1,\ p=0.5$ for the PSV$_{a,p}$.

\section{Conclusions\label{s5}}

In this article,
we introduce a power-scale variation (PSV$_{a,p}$),
parameterized by the sparsity-inducing exponent
$p\in(0,1]$ and the  scaling factor $a\in(0,\infty)$.
Using the PSV$_{a,p}$ minimization, we establish
the stable recovery in both the gradient and the  image domains
under the RIP  framework.
We also design the IRLSPSV algorithm to solve the unconstrained PSV$_{a,p}$ minimization.
Numerical experiments are conducted in three scenarios
including the natural image reconstruction, the MRI  reconstruction,
and the CT reconstruction, which demonstrate the superior performance
and broad applicability of the proposed  PSV$_{a,p}$ model.
The main novelties include that
(i) the PSV$_{a,p}$ minimization  promotes the great flexibility and
the wide applicability compared to single-parameter formulation
due to its two parameters $a$ and $p$,
(ii) in the limiting case $a\to\infty$, the proposed PSV$_{a,p}$ minimization reduces to
the $p$-th power total variation (TV$_p$) minimization,
which, to our knowledge, has not been previously investigated  under the RIP framework,
(iii) applying the sparse convex combination technique, we obtain an asymptotically optimal (in $a$) RIP upper bound
$\overline{\delta}$ for gradient recovery.
To be precise, as $a\to\infty$, $\overline{\delta}$ reduces to $\delta_p$
which is proved to be the tightest upper bound for gradient recovery
via the TV$_p$  minimization,
(iv) sensitivity analysis corroborates the theoretical
finding that $a$ and $p$ play distinct roles
in the PSV$_{a,p}$ minimization, which motivates
a parameter tuning scheme to enhance the practical utility of this model.

\begin{figure}[t]
  \centering
  \begin{tabular}{cccc}
  \includegraphics[width=3.1cm]{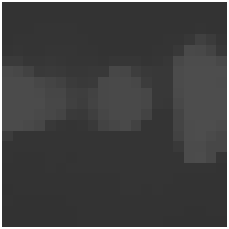}&
   \includegraphics[width=3.1cm]{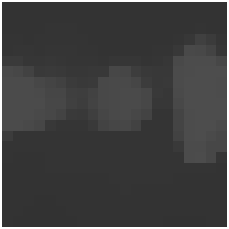}&
  \includegraphics[width=3.1cm]{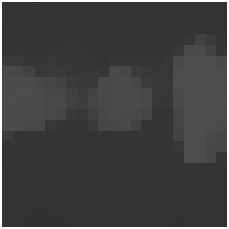}&
  \includegraphics[width=3.1cm]{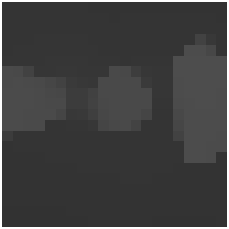}\\
    \includegraphics[width=3.1cm]{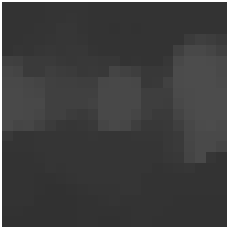}&
     \includegraphics[width=3.1cm]{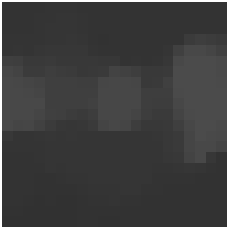}&
  \includegraphics[width=3.1cm]{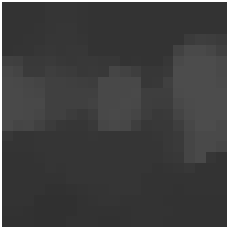}&
  \includegraphics[width=3.1cm]{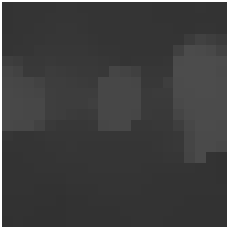}\\
    \includegraphics[width=3.1cm]{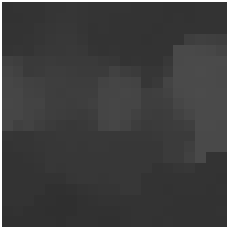}&
     \includegraphics[width=3.1cm]{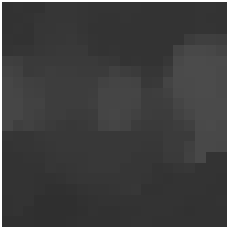}&
  \includegraphics[width=3.1cm]{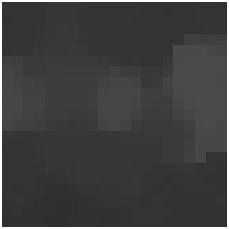}&
  \includegraphics[width=3.1cm]{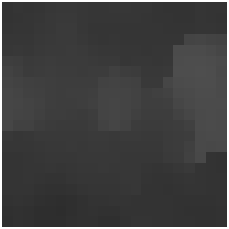}\\
TV &$L_1-\alpha L_2$&TTV& PSV$_{a,p}$
  \end{tabular}
  \caption{Zoom-in view of the regions marked by the red box in Figure \ref{fig-CT}.}\label{fig-CTred}
\end{figure}

\bigskip

\noindent Ziwei Li

\medskip

\noindent Institute of Applied Physics and Computational Mathematics,
Beijing 100088, The People's Republic of China

\smallskip

\noindent{\it E-mail:} \texttt{zwli@buct.edu.cn}

\bigskip

\noindent Wengu Chen

\medskip

\noindent Institute of Applied Physics and Computational Mathematics,
National Key Laboratory of Computational Physics,
Beijing 100088, The People's Republic of China

\smallskip

\noindent{\it E-mail:} \texttt{chenwg@iapcm.ac.cn}

\bigskip

\noindent Huanmin Ge

\medskip

\noindent School of Sports Engineering, Beijing Sport University, Beijing, 100084, China
 The People's Republic of China

\smallskip

\noindent{\it E-mail:} \texttt{gehuanmin@bsu.edu.cn}

\bigskip

\noindent Limei Huo

\medskip

\noindent School of Mathematics and Statistics, Henan University
of Science and Technology, Luoyang, 471023,
The People's Republic of China

\smallskip

\noindent{\it E-mail:} \texttt{limeih@haust.edu.cn}

\bigskip

\noindent Dachun Yang (Corresponding author)

\medskip

\noindent Laboratory of Mathematics and Complex Systems (Ministry of Education of China),
School of Mathematical Sciences, Institute for Advanced Study,
Beijing Normal University, Beijing, 100875,
The People's Republic of China

\smallskip

\noindent{\it E-mail:} \texttt{dcyang@bnu.edu.cn}


\begin{thebibliography}{99}

\bibitem{asf17}
S. A. Abbas, Q. Sun and H. Foroosh,
An exact and fast computation of discrete
Fourier transform for
polar and spherical grid,
IEEE Trans. Signal Process. 65 (2017), 2033--2048.

\vspace{-0.3cm}

\bibitem{AM15} L. Adhikari and R. F. Marcia, $p$-th power
total variation regularization in photon-limited imaging via iterative reweighting,
European Signal Processing Conference 23 (2015),  1621--1625.

\vspace{-0.3cm}

\bibitem{bshyh2007}
N. Bi, Q. Sun,  D.  Huang,  Z. Yang and J. Huang,
Robust image watermarking based on multiband wavelets and
empirical mode decomposition,
IEEE Trans. Image Process. 16 (2007), 1956--1966.

\vspace{-0.3cm}

\bibitem{BKP10} K. Bredies, K. Kunisch and T. Pock,  Total generalized variation,
SIAM J. Imaging Sci. 3 (2010),  492--526.

\vspace{-0.3cm}

\bibitem{CZ13} T. Cai and A. Zhang, Sharp RIP bound for sparse signal and
low-rank matrix recovery,
 Appl. Comput. Harmon. Anal. 35 (2013), 74--93.

\vspace{-0.3cm}

\bibitem{CZ14} T. Cai and A. Zhang, Sparse representation of a polytope
and recovery in sparse signals and low-rank matrices,
 IEEE Trans. Inform. Theory 60 (2014), 122--132.

\vspace{-0.3cm}

\bibitem{CRT06a} E. J. Cand\`{e}s, J. K.  Romberg and T. Tao,
Stable signal recovery from incomplete and inaccurate measurements,
Comm. Pure Appl. Math. 59 (2006), 1207--1223.

\vspace{-0.3cm}

\bibitem{CRT06b} E. J. Cand\`{e}s, J. K.  Romberg and T. Tao,
Robust uncertainty principles: exact signal reconstruction from highly incomplete frequency information, IEEE Trans. Inform. Theory 52 (2006), 489--509.

\vspace{-0.3cm}

\bibitem{CT05} E. J. Cand\`{e}s and T. Tao,  Decoding by linear programming,
IEEE Trans. Inform. Theory 51 (2005), 4203--4215.

\vspace{-0.3cm}

\bibitem{CMM01} T. Chan,  A. Marquina  and P.  Mulet, High-order total variation-based image restoration,
SIAM J. Sci. Comput. 22 (2001),  503--516.

\vspace{-0.3cm}

\bibitem{C07} R. Chartrand,  Exact reconstruction of sparse signals via
nonconvex minimization, IEEE Signal Process. Lett. 14 (2007), 707--710.

\vspace{-0.3cm}

\bibitem{ccls} C. Cheng, Y. Chen, Y. J. Lee and Q. Sun, SVD-based graph
Fourier transforms on directed product graphs,
IEEE Trans. Signal Inform. Process. Netw. 9 (2023), 531--541.

\vspace{-0.3cm}

\bibitem{DDFG10} I. Daubechies, R. Devore, M. Fornasier and C. S. G\"{u}nt\"{u}rk,
   Iteratively reweighted least squares minimization for sparse recovery,
   Comm. Pure Appl. Math. 63 (2010), 1--38.

\vspace{-0.3cm}

\bibitem{DZYZN25} M. Ding, X. Zhao, J. Yang, Z. Zhou and M. K. Ng, Bilateral tensor low-rank representation for
insufficient observed samples in multidimensional image clustering and recovery, SIAM J. Imaging Sci. 18 (2025),  20--59.

\vspace{-0.3cm}

\bibitem{D06} D. L. Donoho, Compressed sensing, IEEE Trans. Inform. Theory 52 (2006),
1289--1306.

\vspace{-0.3cm}

\bibitem{FL09} S. Foucart and M. J. Lai,  Sparsest solutions of
underdetermined linear systems via $l_q$-minimization
for $0 <q\le1$, Appl. Comput. Harmon. Anal. 26 (2009),
395--407.

\vspace{-0.3cm}

\bibitem{FR13} S. Foucart and H. Rauhut,
A Mathematical Introduction to Compressed Sensing, Birkh\"{a}user, Boston, 2013.

\vspace{-0.3cm}

\bibitem{HL25} G. Huang and S. Li,
Low-rank Toeplitz matrix restoration: descent cone
analysis and structured random matrix,
IEEE Trans. Inform. Theory 71 (2025), 3950--3956.

\vspace{-0.3cm}

\bibitem{HCGN22} L. Huo, W. Chen, H. Ge and M. K. Ng, Stable
image reconstruction using transformed total variation minimization,
SIAM J. Imaging Sci. 15 (2022), 1104--1139.

\vspace{-0.3cm}

\bibitem{HCGN23} L. Huo, W. Chen, H. Ge and M. K. Ng, $L_1 - \beta L_q$
minimization for signal and image recovery,
SIAM J. Imaging Sci. 16 (2023),  1886--1928.

\vspace{-0.3cm}

\bibitem{LXY13} M. Lai, Y. Xu and W. Yin, Improved iteratively reweighted
least squares for unconstrained smoothed $\ell_q$ minimization,
SIAM J. Numer. Anal. 51 (2013),  927--957.

\vspace{-0.3cm}

\bibitem{LMS16} A. Lanza, S. Morigi and F. Sgallari, Constrained TV$_p$-$\ell_2$ model for image restoration,
J. Sci. Comput. 68 (2016),  64--91.

\vspace{-0.3cm}

\bibitem{lcgy} Z. Li, W. Chen, H. Ge, and D. Yang, A novel
two-parameter penalty: relaxation degree analysis and
sparse signal recovery, Submitted or arXiv:2603.09722.

\vspace{-0.3cm}

\bibitem{LYHX15}  Y. Lou, P. Yin, Q. He and J. Xin, Computing
sparse representation in a highly coherent dictionary based on difference of
$L_1 $ and $L_2$, J. Sci. Comput. 64 (2015),  178--196.

\vspace{-0.3cm}

\bibitem{LHLXJZ24} J. Lu, L. Huang, X. Liu, N. Xie, Q. Jiang and Y. Zou,
3D Poissonian image deblurring via patch-based tensor logarithmic Schatten-$p$ minimization,
Inverse Problems 40 (2024),  Paper No. 065010, 28 pp.

\vspace{-0.3cm}

\bibitem{lxhljml} J. Lu, C. Xu, Z. Hu, X. Liu, Q. Jiang, D. Meng and Z. Lin,
A new nonlocal low-rank regularization method with
applications to magnetic resonance image denoising,
Inverse Problems 38 (2022), Paper No. 065012, 24 pp.

\vspace{-0.3cm}

\bibitem{LLCZWY17} W. Lu, L. Li, A. Cai, H. Zhang, L. Wang and B. Yan,
A weighted difference of $L_1$ and $L_2$ on the gradient minimization based on alternating direction
method for circular computed tomography, J. X-Ray Sci.  Tech. 25 (2017), 813--829.

\vspace{-0.3cm}

\bibitem{LLCS17}
S. Luo, Q. Lv, H. Chen and J. Song, Second-order total variation and primal-dual algorithm for CT image reconstruction,
Internat. J. Numer. Anal. Model. 14 (2017), 76--87.

\vspace{-0.3cm}

\bibitem{MN25} Z. Ma and M. K. Ng, Multispectral image restoration by generalized opponent transformation total variation,
SIAM J. Imaging Sci. 18 (2025),  246--279.

\vspace{-0.3cm}

\bibitem{NW13} D. Needell and R. Ward, Stable image reconstruction using total variation minimization, SIAM J. Imaging Sci. 6 (2013),  1035--1058.

\vspace{-0.3cm}

\bibitem{NW13b} D. Needell and R. Ward, Near-optimal compressed sensing guarantees for total variation minimization,
IEEE Trans. Image Process. 22 (2013), 3941--3949.

\vspace{-0.3cm}

\bibitem{PMLC20} Z.-F. Pang, G. Meng, H. Li, and K. Chen,  Image restoration via the adaptive TV$_p$ regularization,
Comput. Math. Appl. 80 (2020), 569--587.

\vspace{-0.3cm}

\bibitem{PZLA20} Z.-F. Pang, H.-L. Zhang, S. Luo and T. Zeng,
Image denoising based on the adaptive weighted TV$_p$ regularization,
Signal Process. 167 (2020), Article No. 107325, 14 pp.

\vspace{-0.3cm}

\bibitem{PMGC12} V. M. Patel, R. Maleh; A. C. Gilbert and R. Chellappa, Gradient-based image recovery methods
from incomplete Fourier measurements, IEEE Trans. Image Process. 21 (2012),  94--105.

\vspace{-0.3cm}

\bibitem{DL98} T. Pham Dinh and H. A.  Le Thi, A d.c. optimization
algorithm for solving the trust-region subproblem. SIAM J. Optim.
8 (1998), 476--505.

\vspace{-0.3cm}

\bibitem{P15} C. Poon, On the role of total variation in compressed sensing, SIAM J. Imaging Sci. 8 (2015),  682--720.

\vspace{-0.3cm}

\bibitem{ROF} L. Rudin,  S. Osher and E. Fatemi, Nonlinear total variation based noise removal algorithms, Phys. D 60 (1992), 259--268.

\vspace{-0.3cm}

\bibitem{SS08} S. Setzer and G. Steidl,  Variational methods with higher-order
derivatives in image processing, in:
Approximation Theory XII: San Antonio 2007, pp. 360--385,
Nashboro Press, Brentwood, TN, 2008.

\vspace{-0.3cm}

\bibitem{ss04}
L. Shen and Q. Sun,
Biorthogonal wavelet system for high-resolution image reconstruction,
IEEE Trans. Signal Process. 52 (2004), 1997--2011.

\vspace{-0.3cm}

\bibitem{SHB16} Y. Shen,  B. Han  and E. Braverman, Adaptive
frame-based color image denoising, Appl. Comput. Harmon. Anal.
41 (2016), 54--74.

\vspace{-0.3cm}

\bibitem{SX24} Z. Shuang and J. Xiao, Evolutionary weighted Laplace equations with applications in signal decomposition,
SIAM J. Imaging Sci. 17 (2024),  2015--2052.

\vspace{-0.3cm}

\bibitem{SCBP15} E. Y. Sidky, R. Chartrand, J. Boone, X. Pan,
Constrained TV$_p$ minimization for enhanced exploitation
of gradient sparsity: application to CT image reconstruction,
IEEE J. Transl. Eng. Health Med. 30 (2015), Article Sequence No.
1800418, 18 pp.

\vspace{-0.3cm}

\bibitem{SJP12} E. Y. Sidky, J. H. J{\o}rgensen and X. Pan, Convex optimization problem
prototyping for image reconstruction in computed tomography with the Chambolle-Pock algorithm,
Phys. Med. Biol. 57 (2012), 3065--3091.

\vspace{-0.3cm}

\bibitem{S12} Q. Sun,  Recovery of sparsest signals
via $\ell_q$-minimization, Appl. Comput. Harmon. Anal. 32 (2012),
329--341.

\vspace{-0.3cm}

\bibitem{s14}
Q. Sun,
Localized nonlinear functional equations and two sampling
problems in signal processing,
Adv. Comput. Math. 40 (2014), 415--458.

\vspace{-0.3cm}

\bibitem{TPHD20} D. N. H. Thanh,  V. B. S. Prasath, L. M. Hieu and  S. Dvoenko,
An adaptive method for image restoration based on high-order total variation and inverse gradient,
Signal Image and Video Process. 14 (2020),  1189--1197.

\vspace{-0.3cm}

\bibitem{TLZHJN25} Z. Tu, J. Lu, H. Zhu, W. Hu, Q. Jiang and M. K. Ng,
Fully-connected tensor network decomposition and group sparsity for multitemporal images cloud removal,
Inverse Probl. Imaging 19 (2025),  59--86.

\vspace{-0.3cm}

\bibitem{tlzphjl} Z. Tu, J. Lu, H. Zhu, H. Pan, W. Hu, Q. Jiang
and Z. Lu, A new nonconvex low-rank tensor approximation
method with applications to hyperspectral images denoising,
Inverse Problems 39 (2023), Paper No. 065003, 27 pp.

\vspace{-0.3cm}

\bibitem{WBSS04} Z. Wang, A. C. Bovik, H. R. Sheikh and E. P. Simoncelli,  Image
quality assessment: From error visibility to structural similarity, IEEE
Trans. Image Process.  13 (2004),  600--612.

\vspace{-0.3cm}

\bibitem{WDZ15}  J. Wen, D. Li and F. Zhu,   Stable recovery of
sparse signals via $\ell_p$-minimization, Appl. Comput. Harmon. Anal.
38 (2015), 161--176.

\vspace{-0.3cm}

\bibitem{W97} P. Wojtaszczyk,  A Mathematical Introduction
to Wavelets, London Mathematical Society Student Texts 37,
Cambridge University Press, Cambridge, 1997, xii+261 pp.

\vspace{-0.3cm}

\bibitem{WC13}  R. Wu and D.-R. Chen, The improved bounds of
restricted isometry constant for recovery via $\ell_p$-minimization,
IEEE Trans. Inform. Theory 59 (2013), 6142--6147.

\vspace{-0.3cm}

\bibitem{XY25} J. Xiao and C. Yue, A trace principle for fractional Laplacian with an application to image processing,
Matematica 4 (2025),  84--109.

\vspace{-0.3cm}

\bibitem{XY26} J. Xiao and C. Yue,  Fractional $(L_q,\text{BMO})$ interpolating and imaging,
Ann. Math. Sci. Appl. 11 (2026),  229--266.

\vspace{-0.3cm}

\bibitem{XYMZHW12} Q. Xu, H. Yu, X. Mou, L. Zhang, J. Hsieh and G, Wang,
Low-dose X-ray CT reconstruction via dictionary learning, IEEE Trans. Med. Imaging, 31 (2012), 1682--1697.

\vspace{-0.3cm}


\bibitem{XZMB14} W. Xue, L.  Zhang, X.  Mou, A. C. Bovik,  
Gradient magnitude similarity deviation: A
highly efficient perceptual image quality index,  IEEE Trans. Image Process. 23 (2014), 684--695.

\vspace{-0.3cm}

\bibitem{YL15} J. Yan and W.-S. Lu, Image denoising by generalized total variation regularization and least squares fidelity.
Multidimens. Syst. Signal Process. 26 (2015), 243--266.

\vspace{-0.3cm}

\bibitem{YLHX15} P. Yin, Y. Lou, Q. He and J. Xin, Minimization of
$\ell_{1-2}$  for compressed sensing,
SIAM J. Sci. Comput. 37 (2015),  A536--A563.

\vspace{-0.3cm}

\bibitem{YSS08} J. Yuan, G. Steidl and C. Schn\"{o}rr,
Convex hodge decomposition of image flows, in: Pattern Recognition,
pp. 416--425, Lecture Notes in Comput. Sci. 5096, Springer, Berlin, 2008.

\vspace{-0.3cm}

\bibitem{Z18} A. I. Zayed, Sampling of signals bandlimited to a disc
in the linear canonical transform domain,
IEEE Signal Process. Lett.  25 (2018), 1765--1769.

\vspace{-0.3cm}

\bibitem{Z19} A. I. Zayed, A new perspective on the 
two-dimensional fractional Fourier transform
and its relationship with the Wigner distribution, 
J. Fourier Anal. Appl. 25 (2019),  460--487.

\vspace{-0.3cm}

\bibitem{Z24} A. I. Zayed, Fractional Integral 
Transforms--Theory and Applications, CRC Press, Boca Raton, FL, 2024. xvi+263 pp.

\vspace{-0.3cm}

\bibitem{ZL18} R. Zhang and S. Li, A proof of conjecture on restricted
isometry property constants $\delta_{tk}$ $(0<t<\frac43)$,
IEEE Trans. Inform. Theory 64 (2018),  1699--1705.

\vspace{-0.3cm}

\bibitem{ZL19} R. Zhang and S. Li, Optimal RIP bounds for sparse
signals recovery via $\ell_p$ minimization,
Appl. Comput. Harmon. Anal. 47 (2019), 566--584.

\vspace{-0.3cm}

\bibitem{ZX17} S. Zhang and J. Xin,  Minimization of transformed $L_1$ penalty:
closed form representation and iterative thresholding algorithms,
Commun. Math. Sci. 15 (2017), 511--537.

\vspace{-0.3cm}

\bibitem{ZX18} S. Zhang and J. Xin,   Minimization of transformed $L_1$
penalty: theory, difference of convex function algorithm, and robust
application in compressed sensing, Math. Program. 169 (2018), 307--336.

\vspace{-0.3cm}

\bibitem{ZZWL26} K. Zhao, H. Zhang, J. Wang and Y. Lou, 
Transformed $\ell_1$ regularizations for robust principal
component analysis: toward a fine-grained understanding, 
J. Mach. Learn. (2026), doi.org/10.4208/jml.251003.

\vspace{-0.3cm}

\bibitem{ZC12} W. Zhu and T. Chan, Image denoising using 
mean curvature of image surface, SIAM J. Imaging Sci.
5 (2012), 1--32.

\end{thebibliography}
\end{document}